\documentclass[11pt, draft, onecolumn ]{IEEEtran}
\usepackage{color}     
\usepackage{epsf,psfrag,amssymb,amsfonts,latexsym,cite,verbatim,enumerate,subcaption}
\usepackage{srcltx,amsmath,cases,graphicx}
\usepackage{times, multirow, array}
\usepackage[mathscr]{eucal}
\usepackage{cancel}
 \usepackage{algorithm}
\usepackage{algpseudocode}
\usepackage{algcompatible}
\usepackage{setspace}

\def\psfancypar#1#2{\begingroup\def\par{\endgraf\endgroup\lineskiplimit=0pt}
               \setbox2=\hbox{\large\sc #2}
               \newdimen\tmpht \tmpht \ht2 \advance\tmpht by \baselineskip
               \font\hhuge=Times-Bold at \tmpht
               \setbox1=\hbox{{\hhuge #1}}
               \count7=\tmpht \count8=\ht1
               \divide\count8 by 1000 \divide\count7 by \count8 
               \tmpht=.001\tmpht\multiply\tmpht by \count7 
               \font\hhuge=Times-Bold at \tmpht
               \setbox1=\hbox{{\hhuge #1}}
               \noindent
                \hangindent1.05\wd1
               \hangafter=-2 {\hskip-\hangindent
               \lower1\ht1\hbox{\raise1.0\ht2\copy1}%
                \kern-0\wd1}\copy2\lineskiplimit=-1000pt}

\newcommand{\E}{\mbox{{\rm E}}}

\newcommand{\abf}{\mbox{${\bf a}$}}

\def\boxit#1{\vbox{\hrule\hbox{\vrule\kern3pt
        \vbox{\kern3pt#1\kern3pt}\kern3pt\vrule}\hrule}}

\def\reals{ { {\rm  I \kern-0.15em R }  } }
\def\complex{ {\,{{\rm C} \kern-0.50em \raise0.20ex {  |}}\, }}

\def\mubf{\hbox{\boldmath$\mu$\unboldmath}}

\def\Sigmabf{\hbox{$\bf \Sigma$}}

\def\Pibf{{\bf \Pi}}
\def\abf{{\bf a}}
\def\bbf{{\bf b}}

\def\fbf{{\bf f}}
\def\gbf{{\bf g}}
\def\hbf{{\bf h}}

\def\nbf{{\bf n}}

\def\vbf{{\bf v}}
\def\wbf{{\bf w}}
\def\xbf{{\bf x}}
\def\ybf{{\bf y}}

\def\xbf{{\bf x}}
\def\ybf{{\bf y}}
\def\Abf{{\bf A}}
\def\Bbf{{\bf B}}
\def\Cbf{{\bf C}}
\def\Dbf{{\bf D}}

\def\Fbf{{\bf F}}

\def\Hbf{{\bf H}}
\def\Ibf{{\bf I}}

\def\Rbf{{\bf R}}

\def\Ubf{{\bf U}}
\def\Vbf{{\bf V}}
\def\Wbf{{\bf W}}

\def\Ac{{\cal A}}
\def\Bc{{\cal B}}
\def\Cc{{\cal C}}
\def\Dc{{\cal D}}

\def\Gc{{\cal G}}
\def\Hc{{\cal H}}

\def\Kc{{\cal K}}

\def\Nc{{\cal N}}

\def\Pc{{\cal P}}

\def\Rc{{\cal R}}

\def\Wc{{\cal W}}

\def\be{\vskip .3cm \begin{equation}}
\def\ee{\end{equation} \vskip .4cm \noindent}
\newcommand{\R}{\mbox{$\hat {\bf R}_{N}$}}

\def\Rxx{\Rbf_{\ssstyle X\kern-.1em X}}

\let\ssstyle=\scriptscriptstyle

\def\Kout{\setbox1=\hbox{\Huge\bf K}\hbox to
1.05\wd1{\hspace{.05\wd1}
}}
\def\Sout{\setbox1=\hbox{\Huge\bf S}\hbox to 1.05\wd1{\hspace{.05\wd1}
}}

  \ifx\LabelFigloaded\MYundefined\relax
  \else
   \fi

  \def\LabelFigloaded{\relax}

  \chardef\LabelFigCatAt\the\catcode`\@
  \catcode`\@=11

 \let\LabelFigwlog@ld\wlog
 \def\wlog#1{\relax}

 \ifx\\\MYundefined@
    \let\\\relax
 \fi

  \def\ms@g{\immediate\write16}

 \def\N@wif{\csname newif\endcsname }
 \def\Temp@ {\N@wif\ifIN@}
 \ifx\INN@\MYundefined@
    \else \let\Temp@\relax
 \fi
 \Temp@

  \def\IN@{\expandafter\INN@\expandafter}
  \long\def\INN@0#1@#2@{\long\def\NI@##1#1##2##3\ENDNI@
    {\ifx\m@rker##2\IN@false\else\IN@true\fi}%
     \expandafter\NI@#2@@#1\m@rker\ENDNI@}
  \def\m@rker{\m@@rker}
 
  \newtoks\Initialtoks@  \newtoks\Terminaltoks@
  \def\SPLIT@{\expandafter\SPLITT@\expandafter}
  \def\SPLITT@0#1@#2@{\def\TTILPS@##1#1##2@{%
     \Initialtoks@{##1}\Terminaltoks@{##2}}\expandafter\TTILPS@#2@}

 \def\Shifted@@#1#2#3{\setbox0=\hbox{#3}%
   \raise -\dp0\vbox {\kern-#2%
       \hbox {\kern#1\unhbox0\kern-#1}%
           \kern#2}}

 \newcount\gridcount
 \newbox\auxGridbox@ \newbox\hGridbox@ \newbox\vGridbox@
 \newbox\Labelbox@ \newbox\auxLabelbox@
 \newbox\Coordinatebox@
 \newtoks\Labeltoks@
 \newdimen\Wdd@ \newdimen\Htt@
 \newdimen\Wddd@ \newdimen\Httt@
 
 \def\Wr@{\immediate\write16}

 \newdimen\GL@wd
 \def\GridLineWidth#1{\GL@wd=#1}

 \def\gobble#1{}
 \def\EdgeErr@{\Wr@{}%
      \Wr@{\string\Edges\space argument
      1, 10, 100 or 1000 please\string!}%
      }

 \newcount\Edgect@

 \def\Sweepup#1\endSweepup{}

 \def\SetEdges@{%
    \edef\Zr@@s{\expandafter\gobble\number\Edgect@\empty}%
        \count255=0\Zr@@s\relax
        \ifnum\count255=\z@\else\EdgeErr@\show\tailtest\fi
        \count255=1\Zr@@s\relax
        \ifnum\count255=\Edgect@\relax\else\EdgeErr@\show\leadtest\fi
    \EdgGl@b\edef\Zr@s{\expandafter\gobble\Zr@@s\empty}
    \ifnum\Edgect@>\@ne\relax\EdgGl@b\let\L@Dc\empty
        \else\EdgGl@b\edef\L@Dc{\string.}\fi
    \ifnum\Edgect@>\@ne\relax
        \EdgGl@b\edef\Edgescale@##1{\divide##1 by \Edgect@}%
        \else\EdgGl@b\edef\Edgescale@##1{}\fi
    }

 \def\Edges#1{\Edgect@=#1\relax
     \let\EdgGl@b\global \SetEdges@}

 \Edges{1}

 \def\hhrule{\hrule height \GL@wd\vskip-.\GL@wd}

 \def\hRule@{%
   \advance\gridcount -2%
   \vfil\hhrule\vfil
   \llap{\smash{\raise -2.5pt
     \hbox{\L@Dc\number\gridcount\Zr@s\kern2pt}}}%
   \hhrule
   }

\def\vvrule{\vrule width \GL@wd \kern-\GL@wd}

 \def\vRule@{\advance\gridcount 2%
   \hfil\vvrule\hfil
   \setbox\auxGridbox@=\vbox to 0pt
      {\vskip \Htt@\vskip 2pt
        \hbox to 0pt{\hss\L@Dc\number\gridcount\Zr@s\hss}\vss}%
      \wd\auxGridbox@=0pt \box\auxGridbox@
   \vvrule
   }

 \def\PlaceGrid@@{\gridcount=10 
  \setbox\hGridbox@=\hbox{%
        \hbox{%
             \hskip-.4pt\vrule
             \vbox to \Htt@{%
               \offinterlineskip\parindent=\z@\relax
               \hbox to \Wdd@{\hfil}
               \hRule@\hRule@\hRule@\hRule@
               \vfil\hhrule\vfil}%
             \vrule\hskip-.4pt}
    }%
  \gridcount=0%
  \setbox\vGridbox@=\hbox{%
      \vbox{\offinterlineskip\parindent=0pt\hsize=0pt
         \vskip-.4pt\hrule%
         \hbox to \Wdd@{%
                 \vtop to \Htt@{\vfil}%
                 \vRule@\vRule@\vRule@\vRule@
                 \hfil\vvrule\hfil}%
         \hrule\vskip-.4pt}}%
  \wd\hGridbox@=0pt\ht\hGridbox@=0pt
  \wd\vGridbox@=0pt\ht\vGridbox@=0pt
  \hbox{\box\hGridbox@\box\vGridbox@}%
  }

 \def\LabelsGlobal{\def\LabGl@b{\global}}
 \def\LabelsLocal{\def\LabGl@b{}}
 \LabelsGlobal 

 \def\SetLabels#1\endSetLabels{%
   \LabGl@b\Labeltoks@={#1()\\}%
   }

 \LabGl@b\Labeltoks@={()\\}

 \def\ShowGrid{\LabGl@b\let\PlaceGrid@\PlaceGrid@@}
 \def\HideGrid{\LabGl@b\let\PlaceGrid@\relax}
 \def\Grids{\ShowGrid\LabGl@b\let\GridSwitch@\ShowGrid}
 \def\noGrids{\HideGrid\LabGl@b\let\GridSwitch@\HideGrid}

 \noGrids

 \def\bAdjust@@{%
     \setbox\auxLabelbox@=\hbox{\raise \dp\auxLabelbox@
            \box\auxLabelbox@}}
 \def\bAdjust@{\let\vAdjust@\bAdjust@@}

 \def\eAdjust@@{\dimen0=-.5\ht\auxLabelbox@
     \advance\dimen0 by .5\dp\auxLabelbox@
     \setbox\auxLabelbox@=
            \hbox{\raise\dimen0\box\auxLabelbox@}}
 \def\eAdjust@{\let\vAdjust@\eAdjust@@}

 \def\tAdjust@@{%
     \setbox\auxLabelbox@=\hbox{\raise-\ht\auxLabelbox@
            \box\auxLabelbox@}}
 \def\tAdjust@{\let\vAdjust@\tAdjust@@}

 \let\vAdjust@\relax

 \def\lAdjust@{\let\hAdjust@\rlap}
 \def\rAdjust@{\let\hAdjust@\llap}

 \let\hAdjust@\relax\let\vAdjust@\relax

 \def\FetchLabel@#1(#2)#3\\{%
     \IN@0#2@@\ifIN@
        \setbox0=\hbox{\ignorespaces#1#3\unskip}%
        \ifdim\wd0>0pt
           \ms@g{}%
           \ms@g{ !!! Bad label(s)? !!!}%
           \message{ #1(#2)#3}%
        \fi
        \def\LabelMole@##1\endFetchLabel@{%
            \IN@0()\\@##1@%
            \ifIN@\def\Temp@{\FetchLabel@##1\endFetchLabel@}%
            \else\def\Temp@{}%
            \fi
            \Temp@
           }%
     \else
       \ignorespaces#1\unskip
       \setbox\auxLabelbox@=%
         \hbox to 0pt{\hss\ignorespaces\hAdjust@
          {\ignorespaces#3\unskip}\hss}%
       \vAdjust@
       \let\hAdjust@\relax\let\vAdjust@\relax
       \AugmentLabelBox@@{#2}%
       \ht\Labelbox@=0pt\dp\Labelbox@=0pt
       \let\LabelMole@\FetchLabel@%
     \fi\LabelMole@}

 \newtoks\XYSep@ 
 \def\SetXYSeparator#1{%
     \IN@0#1@@\ifIN@\XYSep@{*}%
     \else
     \XYSep@{#1}%
     \fi
     }

 \SetXYSeparator*

 \def\AugmentLabelBox@@#1{%
     \IN@0\the\XYSep@ @#1@\ifIN@
       \SPLIT@0\the\XYSep@ @#1@%
       \setbox\Labelbox@=\hbox to 0pt{%
         \unhbox\Labelbox@
         \Shifted@@{\the\Initialtoks@\Wddd@}%
         {\the\Terminaltoks@\Httt@}%
         {\box\auxLabelbox@}}%
     \else
         \ms@g{}%
         \ms@g{ !!! Bad insertion point. !!!}%
         \message{ (#1\ this point was rejected.)}%
     \fi
    }

 \def\FetchOption@#1[#2]#3\endFetchOption@{%
    \def\temp{#1}
    \ifx\temp\empty
       \Edgect@=#2\relax
       \let\EdgGl@b\relax
       \SetEdges@
       \Cleaner@#3%
    \fi}

 \def\Cleaner@#1[@]{\Labeltoks@{#1}}
     
 \def\PlaceLabels@@{\mathsurround=0pt
     \def\Cr@{\\}%
     \let\L\lAdjust@\let\R\rAdjust@
     \let\B\bAdjust@\let\E\eAdjust@\let\T\tAdjust@
     \expandafter\FetchOption@\the\Labeltoks@[@]\endFetchOption@
     \Wddd@=\Wdd@ \Edgescale@\Wddd@ 
     \Httt@=\Htt@ \Edgescale@\Httt@
     \expandafter\FetchLabel@\the\Labeltoks@\endFetchLabel@
     \box\Labelbox@
     }%

 \let \PlaceLabels@\PlaceLabels@@

 \def\AffixLabels#1{\setbox\Coordinatebox@=\hbox{#1}%
      \Wdd@=\wd\Coordinatebox@ \Htt@=\ht\Coordinatebox@
      \advance\Htt@ \dp\Coordinatebox@
      \hbox{\copy\Coordinatebox@\kern-\Wdd@ 
           \Shifted@@{0pt}{-\dp\Coordinatebox@}%
           {\PlaceLabels@\PlaceGrid@}%
           \kern\Wdd@}%
      \GridSwitch@ 
      \LabGl@b\Labeltoks@{()\\}%
      }
 
   \let\wlog\LabelFigwlog@ld   
   \catcode`\@=\LabelFigCatAt  

                                By

              Raymond S\'eroul <A18645@FRCCSC21.BITNET>
                                and 
              Laurent Siebenmann <lcs@topo.math.u-psud.fr>
    
              VERSIONS: July 1991, Oct 1991, Jan 1992, July 1992

INTRODUCTION

      This labelling package is intended for TeX users who
rely on non-TeX sources for for their graphics inserts.  It
provides means for adding TeX labels to such inserts with a
minimum of fuss. 

       For most labels, TeX users have in the past found it
reasonably convenient to rely on non-TeX sources. Typical
occasions when an inescapable need for TeX labels seemed to
arise are

 (a) when the graphics program lacks certain exotic or complex
mathematical symbols

 (b) when the very highest typographical quality is wanted for the
labels

 (c) when labels included with the graphics fail to print, 
 and you cannot figure out why (cf. boxedeps.doc).  The labels
 provided by labelfig.tex are 100

       Since this package first appeared, many users, who in the
past scarcely dreamed of using TeX labels, have come to use
nothing but.  So it is now appropriate to add

Intoxication Warning:  TeX labels may be addictive and expensive. 

     If you have a fast preview you may disagree, and even find
that this package provides an agreeable paste-up environment; see
extra applications at end.

     Note to publishers: It is possible and convenient to ultimately
export the TeX labels produced by labelfig.tex to become an integral
part of the EPS file. This is often desired by a publisher who typically
uses an "upmarket" graphics or page layout program, with which the
staff is skilled in perfecting figures.  See Appendix I for
a recipe.

     The authors are grateful to Patrick Ion of Math Reviews for
helpful comments and encouragement.

BASIC INSTRUCTIONS

    After reading in the macro file using

preview or proof your figure with a coordinate grid printed on
top, by typing the following:

    \ShowGrid  
    \AffixLabels{<the graphics insertion>}

Here <the graphics insertion> is what you would type to insert
the graphics object alone without the grid.  This must provide
for the space around it. For example <the graphics insertion>
might well be \BoxedEPSF{MyFigure scaled 700} using the
boxedeps.tex macro package (from same source); this provides a
TeX box containing the encapsulated PostScript insert specified by
the file MyFigure. \AffixLabels{...} provides the grid (supposing
\ShowGrid is present) and later, once you have specified labels
using the grid, it will "tack on" the labels.

     The grid is a sort of (usually elongated) checkerboard of
ten rows and ten columns and its (internal) partitions are by
default numbered  .1, ... ,.9  both horizontally (X-coordinate
running left to right) and vertically (Y-coordinate running bottom
to top).  Thus the points enclosed by the grid correspond to the
points of the unit square in the cartesian "X-Y" plane, the lower
left corner corresponding to the origin (0,0).  By extrapolation,
the full page corresponds to a larger rectangle in the plane.

     These coordinates serve to position labels as follows.
Before the \AffixLabels{...} command type label specifications:

  \SetLabels
   (<X-coordinate>*<Y-coordinate>) <first label> \\
   .
   .
   .
   (<X-coordinate>*<Y-coordinate>)  <last label> \\
  \endSetLabels

Each row specifies one label and is terminated by \\.  In each
row, the position indicator comes first; it is written as a
standard cartesian point except that the X- and Y- coordinates
are separated by * rather than a comma because TeX allows a
comma as decimal point. There are no dimension units to specify
as the unit is the grid itself.

     By default, this cartesian point specifies where the middle
of the baseline of the label will be located.  However if you precede
the point by \L [or \R] the left [or right] edge of the baseline will
be located there. Similarly you may also precede the point by \T, \E,
or \B to vertically align the top equator or bottom of the label box
at the specified point.  This gives nine standard positions of
the label with respect to the insertion point --- corresponding to
the eight principle points of the compas and the center

                     \L\T     \T      \R\T

                     \L\E     \E      \R\E

                     \L\B     \B      \R\B

But this neglects the default "baseline" level of TeX,
giving potentially three more positions

                     \L    <no tag>   \R

For text, the baseline level is often the preferred. Its relation to
the others is variable. It will often coincide with the bottom level,
as happens for "X".  But it is often distinct, as for "g", in which
case you have in all 12 distinct positions rather than 9.

     It is convenient to think of this specification of label
position as attaching the label by a thumb-tack to the coordinate
grid. There are up to twelve positions of the thumb-tack on the
label, while the position of the thumb-tack on the coordinate grid is
arbitrary.  Normally, one choses the position of the thumb-tack on
the label to be the one that is the closest to the item being
labeled.  There are good reasons for this "rule of thumb":

   (a)  It facilitates correct positioning at first try.

   (b)  If the scale of the figure must be altered after labels
have been affixed, the labels have a good chance of remaining well
positioned.

   (c)  The visible grid need not extend beyond the "bounding box"
for the figure, because the best preferred position is always
(at least almost) within the bounding box .

The second reason is particularly important. Indeed it often
happens that scale has to be altered after labelling begins, in
order to either provide space for the labels, or to adjust
proportions between the labels and the figure.  (The size of labels
is unaffected by scaling.)

     Here is an artificial but self-contained test which uses
TeX rules to make a graphics object.

TEST

    Do not skip this!

 \def\FrameIt#1{\hbox{\vrule$\vcenter {\hrule\kern3pt%
             \hbox {\kern3pt #1\kern3pt}%
               \kern3pt\hrule}$\relax\vrule}}

 \def\Caption#1#2{\FrameIt{%
       \vtop {\hsize=#1\relax \parindent=0pt
         \leftskip=0pt \rightskip=0pt plus15pt
         \parfillskip=0pt
         \lineskip=1pt\baselineskip=0pt
         #2}}}

 \def\FirstQuadrant{\hbox to 100pt{\vrule\vbox to 100pt{%
        \hbox to 100pt{\hfil}\vfil\hrule}\hss}}

  \SetLabels
    \R(.5*.2) $\zeta\,\cdot$\\
    (.9*-.10) $\xi$\\
    \R(-.03*.9) $\eta$\\
    \T(.5*.9) \Caption{70pt}{%
          \it The norm of
          $g(\xi+i\eta)$ is indicated on
          contours of this invisible surface.}\\
  \endSetLabels

  \AffixLabels{\FirstQuadrant}

  \end

  Note that the coordinates to use for labels are indicated on the
edges of the grid (when visible) corresponding to the conventional
x- and y- axes of the Cartesian plane. By default the grid is
1-by-1. However, by the command \Edges{100}, you can change this
to 100-by-100 and many users find this alternative most
convenient. Place the command \Edges{...} in your style file (or
header) since its effect is is global. Other possible edge values
are 10 and 1000.

  If you use the command \Edges{...} at all, do so with care.  For
if you accidentally delete an \Edges{...} command your labels will
abruptly be badly misplaced and may logically but mysteriously
generate "dimension too big" errors under TeX and "off page" errors
under your driver.  

  You can dictate the edgescale for an individual figure by giving
the scale in brackets immediately after \SetLabels.  Thus, to
import into an article using say \Edge{100} a figure labelled using
another edgescale, say the original 1-by-1 default, you can use
\SetLabels[1]...\endSetLabels.

GETTING IT DOWN PAT

     Complicated labeling deserves the same respect as
complicated mathematics.  Do not expect it to come out perfect the
first time!  What is needed in either case is a mechanism to
repeatedly typeset troublesome pieces.

     One mechanism is always available.  One does complicated
labelling in a separate "test" file involving just the figure being
labelled;  a texpert will know how to \dump TeX's current state as
a temporary format that restarts rapidly at each retry.  Usually,
one then pastes the completed labelled figure back into the main
TeX file, but, of course, one can also \input it as an auxiliary
file.

     If you do not have a TeXpert at handy, here is a first
approximation to an efficient setup. By deletions reduce a copy
of your article to just a few lines before and after the figure.
Now label the figure, and finally, copy and paste the labelled
figure to the original article. Then copy the next figure to label
into this testbed and repeat. The TeXpert can improve the  speed
at which TeX starts up, by compiling a format specifically for
your article; just one caution: best NOT include in the format
ephemeral details of setup like \Set<mydriver>ArtSpecials (from
boxedeps.tex because this reads  figure dimensions which you may
change during your work session.

     An improved mechanism to repeatedly typeset troublesome
pieces is now available on the Macintosh; it is called LinoTeX;
see the same ftp sources.  It could be set up on many types
of computer.

     Before using labelfig.tex to attach labels to a graphics
object inserted using boxedeps.tex or BoxedArt.tex, make it a
firm rule to carefully adjust the bounding box using the trimming
commands of these packages, and also at least tentatively scale
and position the object. Beware of changing the grid inadvertently
after the labels have been positioned.  For example, correcting
the bounding box of a PostScript graphics object can foul up the
labels by changing the coordinate grid to which the labels are
attached. This is particularly true for the trimming  commands of
boxedeps.tex and BoxedArt.tex. However, as noted already, change
of scale is much less disruptive, and modest adjustments should be
well tolerated.

     Sometimes the labels protrude so far from the bounding box
of a figure that the figure has to be repositioned.  Best do this
by ad hoc spacing, say using \hglue and \vglue; altering the
bounding box would create a vicious circle.

     Remember that you are responsible for preventing labels
from overlapping. You are responsible for all label typography
including size and style. A label is really just about anything
that can be put in a TeX box. Note that spaces at the beginning
and end of labels will normally be suppressed; if you really want
them you must protect them with TeX braces.

     This package temporarily sets the \mathsurround parameter
of TeX to zero  while the labels are being affixed. This is done
because nonzero \mathsurround space would influence the position
of left and right aligned labels; then, when a texpert or printer
modifies mathsurround, diagram labeling might be disastrously
altered. There is a small price to pay involving labels that are
formatted as caption boxes including mathematics: you  may want or
need to specify an explicit mathsurround space within the caption
box; it will not influence anything outside.

     Those hostile to the use of * as separator between
the X and Y coordinates of label insertion points, are free to
impose another using \SetXYSeparator{<the new separator>}.  
Americans may prefer "," to "*" since they never use a 
comma as a decimal point; on the other hand, * may be more visible.

APPENDIX (I)  MERGING labelfig.tex LABELS INTO AN EPSF GRAPHICS OBJECT.

     As promised in the introduction, here is a recipe useful for
publishers. It works at least on Macintosh and at least for vectorized
graphics and Adobe type1 fonts.  (There is surely a similar recipe for
PCs under MSWindows.)

 (a)  Use boxedeps.tex utility to integrate the figure given by the eps
file, "x.eps" say, with a visible frame around it.  See
\ShowDisplacementBoxes command in boxedeps.tex.  To get precise results
automatically it is important to use the \Trim... commands of
boxedeps.tex making the "DisplacementBox" neatly fit the figure.

 (b)  Use the TeX printer driver and LaserWriter (versions >= 8.1.1) to
export to an EPSF the DVI page containing the integrated, labelled
figure. You now have an EPS file  "xx.eps"  that contains too much, and at
the wrong scale, and at wrong position.

 (c)  Convert the EPSF to an Adode Illustrator format EPSF using
the shareware utility called epsConvert by Sam Weiss
1993-- (currently $25).

 (d)  In Illustrator (or a compatible program), group the labels and the
"DisplacementBox"; copy them to the clipboard and paste them into "x.ps".
This step requires that all the label fonts be "visible to the Macintosh.

 (e)  Translate and scale the pasted group consisting of the labels plus
the "DisplacementBox" so as to make the "DisplacementBox" the bounding
box of (labelless) figure represented by "x.eps".  At this point the
labels will be correctly placed on the figure "x.eps".

 (f)  Ungroup and delete the "DisplacementBox".  The result is the
desired single EPS file, "x+.eps" say, It contains the original figure
plus its labels.  

     Using grouping and ungrouping appropriately in "x+.eps", a
publisher's staff can very efficiently improve label positions etc.

APPENDIX II)  SOME EXOTIC APPLICATIONS

     The grid of labelfig.tex is analogous to a light-table in
classical page makeup with wax or latex glue.  In principle, you
can use it to compose any page from its indivisible parts.  This
even has some of the artisanal charm of classical paste-up
provided you have a fast screen preview to make the process
"interactive".

     In practice labelfig.tex is a tool for nonstandard jobs.
Here are a few going beyond the labelling already discussed.

(I)  GRAPHICS INTEGRATION.

     This is accomplished by treating the imported graphics
objects as labels.  The underlying graphics object is then
typically an empty  \vbox to <dimension>{\vfill} in a TeX
\midinsert...\endinsert construction.  A label line
might be of the form

   (.1*.1) \\

The exact form of the special command varies from driver to
driver.  However, in the case of encapsulated PostScript graphics
(EPSF norm), by relying on boxedeps.tex, one can have the
following standard syntax (independant of driver  (see
boxedeps.doc for details.
  
  (.1*.1) \BoxedEPSF{MyFigure scaled <scale in mils>}\\

This may be slow since it requires TeX to read the PostScript
file to read bounding box using many complex macros.  So you
may want to try

  (.1*.1) \EPSFSpecial{MyFigure}{<scale in mils>}\\

which is fast and driver independant, but it squashes the
bounding box, normally to its lower left corner.

     Similarly for graphics of the Macintosh PICT norm ---
using BoxedArt.tex (same sources) in place of boxedeps.tex.

     This approach to integration is to be recommended when
one is assembling a composite graphics object.

 (II)  COMMUTATIVE DIAGRAM ENHANCEMENT

     Commutative diagrams or arrays of mathematical objects
connected by arrows of various sorts are common in mathematics.
The mathematical objects require the use of TeX.  Recently TeX
acquired a good collection of arrows of all slopes --- that of
LamSTeX --- plus pwerful macros to build the diagrams.

     However, even the LamSTeX collection is often
inadequate; it lacks for example double shafted arrows, dotted
arrows and curved arrows. Fortunately it is possible to produce
such arrows on an individual basis using sophisticated graphics
programs such as Illustrator and AldusFreehand (both serving
the EPSF norm) or using Metafont (with its public domain norm).
Since the creation of each new arrow is a work of love, you
probably want to limit the number of arrows by using LamSTeX
for most arrows. The 40K commutative diagram module of LamSTeX
has been adapted to work with AmSTeX and a copy may be posted
with LabelFig and related files. Unfortunately no one has yet
offered a version that works with Plain TeX or LaTeX.

       Suffice it here to say that when the exotic arrow has
been somehow imported into TeX, labelfig.tex treats it as a
label that one affixes to the commutative diagram.  Two other
steps will be treated in separate notes, namely the matter of
extracting the dimension specifications for the arrow and the
construction of the arrow --- for these steps are far from
unique and often depend intimately on your computer environment. 
Notes for the Macintosh-Textures-Illustrator combination are
found in the file ExoticArrows.doc.

 (III) NESTING 

Ingenuity pays off in exploiting labelfig.tex. One can
mix graphics and typography quite freely.  labelfig.tex is good
for freeform or overlapping arrangements, while boxedeps.tex (or
BoxedArt.tex) is best for regimented non-overlapping
arrangements --- and the two can be combined.

     The default behavior of labelfig.tex is not ideal 
for nesting objects, because to prevent trouble for beginners
the register for labels is globally cleared when \AffixLabels
concludes.  But there are switches available

      \LabelsGlobal      \LabelsLocal

which change this.  To understand this, extend the above test 
by something like:

 \LabelsLocal

 \SetLabels
    (.5*.5) AAA\\
 \endSetLabels

 {
 \SetLabels
    (.5*.5) ZZZ\\
 \endSetLabels
   \AffixLabels{\FirstQuadrant}
 }

   \AffixLabels{\FirstQuadrant}

     There are however potential pitfalls.  Neither
labelfig.tex nor boxedeps.tex has been tested under extreme
conditions. Problems may occur if their procedures are
indiscriminately nested. For boxedeps.tex (not labelfig.tex)
there is a precise cause for worry, namely many of its
variables are "global", which means that TeX braces will not
provide the protection one might expect.

COMMAND SUMMARY FOR labelfig.tex

  Here [...] means optional (one or zero)
       [...]* means any number of such constructs

  \SetLabels
    [[<P>](<X><Sep><Y>) <label> \\]*
  \endSetLabels
  \ShowGrid  
  \AffixLabels{<the figure>}

   --- <P> is tack position, one of eleven or empty
              order irrelevant

                   \L\T      \T      \R\T

                   \L\E      \E      \R\E

                     \L               \R

                   \L\B      \B      \R\B

   --- (<X><Sep><Y>) insertion point;
  <Sep> is separator, = * by default;
  \SetXYSeparator{<Sep>} changes it.
   <X> and <Y> are real numbers

  --- <label> a label to attach 

  --- <the figure> the figure to label 

  \GlobalLabels (default)     
  \LocalLabels  setting for nested constructs.

 \Grids makes ALL grids appear; \HideGrid then makes just next disappear.
 \noGrids returns to default.  The commands are always global.

 \GridLineWidth{<dimension>} adjusts width of grid lines. Default is very
small, to give "hairline" effect. If your grid lines are missing try
setting \GridLineWidth{1pt}.

 \Edges#1 globally changes the edge size of all grids to the numerical 
value #1, which must be 1, 10, 100, or 1000.  The default is 1.

VERSION HISTORY.
 --- Jan 1993: \Edges#1 and [??] option after \SetLabels
 --- July 1992: \Grids, \noGrids, \HideGrid;
       Gridlines become hairlines; \GridLineWidth{<dimension>}.
 --- Oct 1991, Jan 1992: \SetXYSeparator{<Sep>},  \LabelsGlobal,
       \LabelsLocal.
 --- July 1991: first release

Address for bugs and other feedback:

        Raymond S\'eroul
        IREM and Lab. de Typographie Informatise
        Univ. Rene Descartes
        Strasbourg

    Tel 33-88-41-63-45
    Email:  A18645@FRCCSC21.BITNET

        Laurent Siebenmann
        Mathematique, Bat. 425,
        Univ de Paris-Sud,
        91405-Orsay,
        France

    Tel 33-1-6941-7949; 
    Email: lcs@topo.math.u-psud.fr

\def\scalefig#1{\epsfxsize #1\textwidth}

\newcommand {\Ebb}{{\mathbb{E}}}

\newtheorem{theorem}{Theorem}

\newtheorem{lemma}{Lemma}

\newtheorem{corollary}{Corollary}
\newtheorem{example}{Example}

\newtheorem{proposition}{Proposition}

\algnewcommand\algorithmicinput{\textbf{Step}}
\algnewcommand\Step{\item[\algorithmicinput]}

\title{\huge{A High-Diversity Transceiver Design for MISO Broadcast Channels}}

\author{
 Junyeong Seo, {\em Student~Member, IEEE}, Youngchul
Sung$^\dagger$\thanks{$^\dagger$Corresponding author}, {\em
Senior~Member, IEEE}, \\
\thanks{J. Seo and Y. Sung are with Dept. of Electrical Engineering,  KAIST, Daejeon 305-701, South
Korea, and H. Jafarkhani is with Center for Pervasive Communications \& Computing, UC Irvine, CA, USA.
E-mail: jyseo@kaist.ac.kr, ysung@ee.kaist.ac.kr, and hamidj@uci.edu.
This work was supported in part by Basic Science Research Program through the National Research Foundation of Korea (NRF) funded by the Ministry of Education (2013R1A1A2A10060852) and  supported in part by the NSF Award CCF-1526780. This work is from a part of \cite{Seo18Thesis}.}
 and Hamid Jafarkhani, {\em Fellow, IEEE}
}

\begin{document}

\maketitle

\begin{abstract}
In this paper, the outage behavior and diversity order of the mixture transceiver architecture for multiple-input single-output broadcast channels are analyzed. The mixture scheme groups users with closely-aligned channels and applies superposition coding and successive interference cancellation decoding to each group composed of users with closely-aligned channels, while applying zero-forcing beamforming across semi-orthogonal user groups. In order to enable such analysis, closed-form lower bounds on the achievable rates of a general multiple-input single-output broadcast channel with superposition coding and successive interference cancellation are newly derived. By employing channel-adaptive user grouping and proper power allocation, which ensures that the channel subspaces of user groups have angle larger than a certain threshold, it is shown that the mixture transceiver architecture achieves full diversity order in multiple-input single-output broadcast channels and opportunistically increases the multiplexing gain while achieving full diversity order. Furthermore, the achieved full diversity order is the same as that of the single-user maximum ratio transmit beamforming. Hence, the mixture scheme  can provide reliable communication under channel fading for ultra-reliable low latency communication.  Numerical results validate our analysis and show the outage superiority of the mixture scheme over conventional transceiver designs for multiple-input single-output  broadcast channels.
\end{abstract}


\section{Introduction}

{\em The multiple-input single-output (MISO)  broadcast channel (BC) model} is an important channel model which captures modern cellular downlink communication in which a base station (BS) equipped  with multiple transmit antennas simultaneously serves multiple receivers each equipped with a single receive antenna  by using the spatial domain. Due to its importance it has been investigated extensively  for more than a decade and major current  wireless communication standards  support MISO BC downlink communication\cite{Weingartenetal06IT,Sharif&Hassibi:05IT,Yoo&Goldsmith,LTEMUMIMOphy}.
It is known that  the capacity region of a MISO BC can be achieved by dirty paper coding (DPC) \cite{Weingartenetal06IT}.  However, because of the unavailability of practical dirty paper codes,  simple linear downlink beamforming such as zero-forcing (ZF) beamforming is widely considered and used in practice\cite{Spencer&Swindlehurst&Haardt:04IT,Yoo&Goldsmith}. Although such simple linear beamforming is not a capacity-achieving scheme, it can yield good performance when it is combined with multi-user diversity and user scheduling\cite{Sharif&Hassibi:05IT,Yoo&Goldsmith,LeeSung18COM,LeeSungSeo,LeeSungKountouris}.
 That is, when the number of users in the cell is sufficiently large as compared to the number $N$ of transmit antennas, the BS
can select $N$ users with nearly orthogonal channel vectors so
that linear ZF downlink
beamforming is sufficient.
However,  such orthogonality-based user scheduling for linear downlink beamforming may not be appropriate in certain cases. One example is the case in which the number of transmit antennas is large under rich scattering environments since it is difficult to simultaneously select multiple users with roughly orthogonal channels  in this case\cite{Huh:12IT,LeeSungSeo,LeeSungKountouris}. Thus, for
a MISO BC with a large number of transmit antennas  it was proposed
that the BS selects the users for simultaneous service arbitrarily and applies linear
ZF beamforming\cite{Huh:12IT}.
 Another emerging important example is  {\em ultra-reliable low-latency  communication (URLLC)} for fast machine-type communication in 5G. 
  In the case of URLLC,
  such orthogonality-based user scheduling
induces extra delay in communication since the users requiring immediate data transmission may not have channel vectors nearly orthogonal to each other or to other on-going overlapping data users under spatial multiplexing.  Hence, it is  preferred that  the BS immediately schedules the users requiring low-latency  data transmission regardless of their channel vectors' mutual orthogonality. In both examples, the channel vectors of the scheduled users are not guaranteed to be nearly orthogonal and the performance of linear ZF beamforming can be severely degraded since
 the
channel vectors of some of the scheduled users can be closely
aligned and the channel alignment causes  poor conditioning of the channel matrix for ZF inversion.

Recently,  inspired by the  usefulness of superposition coding and successive interference cancellation (SIC) decoding in non-orthogonal multiple access  (NOMA) \cite{Saito:13VTC,mao2018rate}, a mixture (or hybrid) transceiver architecture was considered for MISO BCs  to
overcome  the drawback of the fully linear ZF downlink
beamforming based on {\em  user grouping} and {\em mixture of linear and non-linear reception}\cite{Seo&Sung:17SPAWC,Chen&Ding&Dai:16Access}.
The basic idea of the {}{mixture} transceiver architecture is as follows.
Under the assumption of independent and identically distributed (i.i.d.) realization of $K$ channel vectors
in  $K$-user MISO downlink,   if the channel
vectors of some users are closely aligned, the performance of ZF beamforming is severely degraded. However, if we group the closely-aligned users
and apply  superposition coding and non-linear SIC decoding
for each closely-aligned user group while applying ZF beamforming
across roughly-orthogonal user-groups, the performance degradation
by the full ZF beamforming can be alleviated. Preliminary study on such user grouping and  mixture transreception was performed on the two-user grouping case, where  intra-group rate analysis is rather simple\cite{Seo&Sung:17SPAWC,Chen&Ding&Dai:16Access}.
In \cite{Seo&Sung:17SPAWC}, Pareto-optimal beam design is considered for the two-user grouping case, the beam vectors and corresponding rates are numerically obtained, and  the performance of the mixture scheme is compared with the full ZF beamforming numerically.
In \cite{Chen&Ding&Dai:16Access}, under the assumption of two users in each group, closed-form beam vectors are obtained  to minimize the transmit power under a  signal-to-interference-plus-noise ratio (SINR) constraint for each user based on quasi-degradation, and it was shown that such a mixture architecture based on two-user grouping increases the diversity order by one as compared to the conventional ZF downlink beamforming. Although such two-user grouping for the mixture transceiver architecture is  tractable, it has limitation in diversity order improvement. (The related idea of hierarchical coding and user grouping  was  discussed  in the dual scenario of multiple access channel in \cite{Hamid08COM}, {}{and the idea of user grouping and inter-group zero forcing was also considered in \cite{Adhikary13IT} using the intra-group processing of a classical spatial multiplexing from a capacity perspective.})

In this paper, we fully generalize the mixture transceiver architecture for general MISO BCs. The contributions of the paper are summarized as follows:

\noindent $\bullet$ {In order to enable analysis of
 the outage probability and diversity order of the  mixture transceiver architecture,  we derive a new  lower bound on the achievable rate of each user in closed form in terms of each user's channel norm for a  MISO BC with superposition coding and  SIC decoding with an arbitrary number of users.}

\noindent $\bullet$ {We propose a channel-adaptive user grouping method which ensures a condition for the channel subspace angle property for the constructed user groups and a  power allocation method necessary for achievability of full diversity order.}

\noindent $\bullet$  {Combining the newly derived  achievable rate result and the property of the proposed adaptive  user grouping method, we derive the diversity order of the mixture transceiver architecture, and show that  {\em the mixture transceiver architecture achieves  full diversity order in MISO BCs, which is the same as that of the single-user maximal ratio transmit (MRT) beamforming}, and furthermore it  opportunistically increases multiplexing gain.}

\noindent $\bullet$ {We further investigate the related issues such as diversity-and-multiplexing trade-off associated with the mixture scheme, impact of imperfect channel state information (CSI), etc.}

\textit{Notations:}  Vectors and matrices are written in boldface
with matrices in capitals. All vectors are column vectors. For a
matrix $\Abf$, $\Abf^*$, $\Abf^H$, $\Abf^T$ and $\mbox{Tr}(\Abf)$ indicate the
complex conjugate, conjugate transpose,   transpose and trace of $\Abf$,
respectively, and  $\Cc(\Abf)$ and $\Cc^\perp(\Abf)$ denotes the
linear subspace spanned by the columns of $\Abf$ and its orthogonal
complement, respectively.
 $\Pibf_\Abf$ and
$\Pibf_\Abf^\bot$ are the projection matrices to $\Cc(\Abf)$ and
$\Cc^\perp(\Abf)$,
 respectively.
 $[\abf_1,\cdots,\abf_n]$ denotes the matrix composed of column vectors $\abf_1,\cdots,\abf_n$.
 $||\abf||$ represents the 2-norm of vector $\abf$. $\Ibf_n$ denotes the identity matrix of size $n$ (the subscript is omitted when unnecessary).
$\xbf\sim\Cc\Nc(\mubf,\Sigmabf)$ means that random vector $\xbf$ is
circularly-symmetric complex Gaussian distributed with mean vector
$\mubf$ and covariance matrix $\Sigmabf$.

\section{The Channel Model and Preliminaries}

\subsection{The Channel Model}  \label{subsec:ChannelModel}

In this paper, we consider a Gaussian MISO BC composed of a transmitter with $N$ transmit antennas and $K$ single-antenna users (i.e., receivers), where the number of users is less than or equal to the number of transmit antennas, i.e., $K \le N$. The received signal $y_k$ at the $k$-th user is given by
\begin{equation}  \label{eq:channelModel}
y_k = \hbf_k^H \xbf + n_k, ~~~k=1,2,\cdots, K,
\end{equation}
where $\xbf$ is the $N\times 1$ transmit signal vector at the transmitter with the total transmit power $P_t = {\mathbb{E}} \{ \xbf\xbf^H\}$, $n_k$ is the additive white Gaussian noise (AWGN) at the $k$-th user, i.e., $n_k \sim \Cc\Nc(0,\sigma^2)$ with $\sigma^2$ set to $1$ for simplicity, and $\hbf_k$ is the $N \times 1$ (conjugated) channel vector from the transmitter to the $k$-th user following independent Rayleigh fading, i.e.,
\begin{equation}  \label{eq:hbfkinmodel}
\hbf_k=[h_{k1},h_{k2},\cdots,h_{kN}]^T \stackrel{i.i.d.}{\sim} \Cc\Nc(\mathbf{0}, 2\Ibf).
\end{equation}
Here, we set $2\Ibf$ as the covariance matrix for convenience  so that both real and imaginary components of each element of $\hbf_k$ have variance one and thus $||\hbf_k||^2$ has the chi-square distribution of degrees of freedom $2 N$.
Different scaling can be absorbed into the transmit power.
Concatenating all the received signals $y_1,\cdots,y_K$, we can write the matrix model for the received signals as
\begin{equation}
\ybf = \Hbf^H \xbf + \nbf,
\end{equation}
where $\ybf=[y_1,y_2,\cdots,y_K]^T$, $\nbf=[n_1,n_2,\cdots,n_K]^T$, and
$\Hbf = [\hbf_1,\hbf_2,\cdots,\hbf_K]$.
{}{We assume that the channel state information (CSI) $\Hbf$ is available at the transmitter.}
Due to the assumption of $K \le N$, the $K\times N$ overall channel matrix $\Hbf^H$ is a fat  matrix and hence it is right-invertible so that conventional ZF transmit beamforming is feasible.
Design of the signal vector $\xbf$ and receiver processing based on $\{y_1,y_2,\cdots,y_K\}$ will be explained in the subsequent sections.

\subsection{Preliminaries: Reliability and Diversity Order}
\label{subsec:prelim}

Channel fading is inherent in wireless communication, and communication reliability under channel fading is  dependent on the diversity order of the communication channel.
Consider the well-known single-user
 MRT beamforming with multiple transmit antennas. The corresponding channel model is given by  the channel model \eqref{eq:channelModel} with only a single user, i.e., $K=1$. For MRT beamforming, we have $\xbf = \frac{\hbf_1}{||\hbf_1||}\sqrt{p_1}s_1$ with ${\mathbb{E}}\{|s_1|^2\}=1$. The resulting equivalent single-input single-output (SISO) channel and rate are respectively given by
\begin{equation}  \label{eq:MRTsisoEQ}
y_1 = ||\hbf_1|| \sqrt{p_1}s_1 + n_1 ~~~\mbox{and}~~~
R_1= \log ( 1 + ||\hbf_1||^2 \mathrm{SNR}), ~~~\mathrm{SNR}:= \frac{p_1}{\sigma^2},
\end{equation}
where the probability density function (pdf) of $||\hbf_1||^2 = |h_{11}|^2 + \cdots + |h_{1N}|^2$ is given by the chi-square distribution with degree of freedom $2N$ since it is the sum of the squares of $2N$ standard normal random variables:
\begin{align}  \label{eq:ChiSquare2N}
f_{||\hbf_1||^2}(x) &= \frac{1}{2^N(N-1)!}x^{N-1}e^{-x/2}=\frac{1}{2^N(N-1)!}x^{N-1} + o(x^{N-1}), ~\mbox{as}~x\rightarrow 0,
\end{align}
where $o(\cdot)$ is the small o notation.
Communication outage is defined as the event that the channel cannot support a given target rate $R^{th}$, and  the corresponding outage probability is given by
$P_{out} = \mathrm{Pr}\{ R_1 < R^{th}   \}$\cite{Tse:book}.
Then, the diversity of order of the channel is defined as  \cite{Tse:book}
\begin{equation}
D := - \lim_{\mathrm{SNR}\rightarrow \infty} \frac{\log P_{out}}{\log\mathrm{SNR}}.
\end{equation}
In the single-user MRT beamforming case, the outage probability  is given by
$P_{out} = \mathrm{Pr}\left\{ ||\hbf_1||^2  \le \frac{2^{R^{th}}-1}{\mathrm{SNR}} \right\}$ $\approx \frac{(2^{R^{th}}-1)^N}{2^N N! \mathrm{SNR}^N}$\cite{Tse:book},
 and hence the diversity order in this case is  $N$.  That is, the outage probability decays as $\mathrm{SNR}^{-N}$, as SNR increases. Note that in the case of a Rayleigh-fading SISO channel with a single transmit antenna $N=1$, the pdf \eqref{eq:ChiSquare2N} reduces to
$f_{|h_{11}|^2}(x)= \frac{1}{2}e^{-x/2}$,
and the diversity order reduces to one. Hence, MRT beamforming with $N$ transmit antennas  increases the diversity order by $N$ times as compared to the SISO case.

Now, consider the general  Gaussian MISO BC \eqref{eq:channelModel} with ZF downlink beamforming for $K=N$.  In the ZF beamforming case, the overall transmit signal $\xbf$ is given by
$\xbf = \wbf_1^{ZF} \sqrt{p_1}s_1 + \cdots + \wbf_K^{ZF} \sqrt{p_K}s_K$,
where $\wbf_k^{ZF}$ and $s_k$ are the ZF beam vector and data symbol for the $k$-th user with $||\wbf_k^{ZF}||^2=1$ and $\Ebb\{|s_k|^2\}=1$, respectively. Here, the ZF beam vector $\wbf_k^{ZF}$ lies in $\Cc^\perp([\hbf_1,\cdots,$ $\hbf_{k-1},\hbf_{k+1},\cdots,\hbf_K])$
so that $\hbf_i^H \wbf_k^{ZF} =0$ for all $i \ne k$. Then, the resulting  SISO channel for the $k$-th user is given by
\begin{equation}  \label{eq:ZFsisoEQ}
y_k = \hbf_k^H \wbf_k^{ZF} \sqrt{p_k} s_k + n_k.
\end{equation}
In the case of independent Rayleigh fading, the channel vector $\hbf_k$ and the remaining  $\{ \hbf_1,\cdots,\hbf_{k-1},\hbf_{k+1},$  $\cdots,\hbf_K\}$ are independent. Hence, the one-dimensional subspace $\Cc^\perp([\hbf_1,\cdots,\hbf_{k-1},\hbf_{k+1},\cdots,\hbf_K])$ is also   independent of $\hbf_k$, and hence $\hbf_k$ is circularly-symmetric Gaussian distributed over ${\mathbb{C}}^N$ with respect to a reference direction of  $\Cc^\perp([\hbf_1,\cdots,\hbf_{k-1},\hbf_{k+1},\cdots,\hbf_K])$.  Therefore,  taking the inner product between $\hbf_k$ and the unit-norm vector $\wbf_k^{ZF} \in \Cc^\perp([\hbf_1,\cdots,\hbf_{k-1},\hbf_{k+1},\cdots,\hbf_K])$ is equivalent to taking only one component out of $N$ complex Gaussian components, and thus  $|\hbf_k^H\wbf_k^{ZF}|^2$ has the same pdf as  $f_{|h_{11}|^2}(x)= \frac{1}{2}e^{-x/2}$. Hence,
the corresponding diversity order for  the $k$-th user is simply one for all $k$\cite{Chen&Ding&Dai:16Access} as in the SISO Rayleigh fading channel.  Thus, ZF downlink beamforming for MISO BCs loses the diversity gain possibly obtainable from multiple transmit antennas.

Note that if $\hbf_k$ is perfectly orthogonal to $\hbf_1,\cdots,\hbf_{k-1},\hbf_{k+1},\cdots,\hbf_K$, then  $\Cc^\perp([\hbf_1, \cdots,\hbf_{k-1},\hbf_{k+1},$ $\cdots,\hbf_K])$ is perfectly aligned with $\hbf_k$ and hence in this case we have
$\hbf_k^H\wbf_k^{ZF} = ||\hbf_k||$.
In this case, the resulting SISO channel for the $k$-th user is the same as that of the MRT beamforming single-user channel in \eqref{eq:MRTsisoEQ}.  Furthermore, suppose that the angle between $\hbf_k$ and one-dimensional subspace $\Cc^\perp([\hbf_1, \cdots,\hbf_{k-1},\hbf_{k+1},$ $\cdots,\hbf_K])$ is equal to or less than a certain fixed threshold $\alpha$. Then, we have
$|\hbf_k^H\wbf_k^{ZF}| \ge ||\hbf_k||\cos \alpha$.
Since $\cos \alpha$ is a constant, the pdf of $|\hbf_k^H\wbf_k^{ZF}|^2$ is a certain scaled version of that of $||\hbf_k||^2$ (the meaning of this statement will become clear in later sections), and the outage behavior for the $k$-th user in this case should be the same as that of the MRT single-user case as SNR increases without bound.
Reflecting this, one can recognize that the degradation of diversity order of ZF beamforming for a MISO BC with independent channel fading results from the uncontrolled and arbitrary angle between
$\hbf_k$ and  $\Cc^\perp([\hbf_1, \cdots,\hbf_{k-1},\hbf_{k+1},$ $\cdots,\hbf_K])$.

\section{The {}{Mixture} Transceiver Architecture }
\label{sec:systemmodel}

In this section, motivated by the discussion in the previous section, {}{we consider the mixture transceiver architecture for  Gaussian MISO BCs \cite{Seo&Sung:17SPAWC,Chen&Ding&Dai:16Access} in order to overcome the diversity drawback of ZF downlink beamforming. The mixture architecture is  based on user grouping and mixture of linear ZF and non-linear SIC reception. First, user grouping is performed to group users with closely-aligned channel vectors. Then,  superposition coding and SIC are applied to the users with closely-aligned channel vectors in each group, whereas ZF beamforming is applied across groups.
In order to  fully enhance the diversity order of the resulting individual user channel,
we generalize the mixture architecture by adopting adaptive user grouping, which yields channel-dependent groups
and
enforces the angle between the subspace of each group and the orthogonal complement of the union of all other groups' subspaces to be less than a certain threshold  so that inter-group ZF beamforming does not harm the overall diversity order.}

From here on, we explain the {}{mixture transceiver architecture with the proposed user grouping method} in detail. We consider the MISO BC explained in Section \ref{subsec:ChannelModel} as our channel model.   We assume the following for our transceiver architecture:

{\em A.1 (User Grouping):} First, we group the $K$ users into $N_g$ groups. The constructed groups are denoted by the sets $\Gc_1, \Gc_2,\cdots, \Gc_{N_g}$ such that $\Gc_i \cap \Gc_j = \emptyset$ for $i\ne j$ and $\bigcup_{j=1}^{N_g} \Gc_j = \{1,2,\cdots,K\}$.
User grouping is adaptive in the sense that the number of groups can vary and the number of  members in each group can vary from one to $K$, depending on the channels such that $\sum_{j=1}^{N_g} |\Gc_j|=K$.  The constructed groups satisfy a certain subspace angle property in order to apply inter-group ZF beamforming without degrading the diversity order.
The detailed method for user grouping will be presented in Section \ref{sec:User Grouping Algorithm}.

{\em A.2 (Inter-Group Beamforming):} With the constructed groups,  in order to control inter-group interference,
we apply ZF beamforming across the constructed groups. With this inter-group ZF beamforming,
the inter-group interference across the groups is zero.

{\em A.3 (Intra-Group Processing: Superposition Coding and SIC):}
With the constructed groups, for intra-group processing we apply superposition  coding and SIC decoding
to each and every group with more than one user.

Under the aforementioned transceiver architecture, the transmit signal $\xbf$ of the transmitter  can be expressed as
\begin{equation}
 \mathbf{x} = \sum_{j=1}^{N_g} \mathbf{\Pi}^{(j)}\sum_{i \in \mathcal{G}_j}\sqrt{p_i^{(j)}} \mathbf{w}_i^{(j)}
    s_i^{(j)}, \label{eq:new_arch}
\end{equation}
where $s_i^{(j)}$ is the transmit symbol from  $\mathcal{CN}(0,1)$
for User $i$ in group $\mathcal{G}_j$, $\wbf_i^{(j)}$ is the $N \times 1$ intra-group beamforming vector for User $i$ in group
$\mathcal{G}_j$ out of the feasible set $\widetilde{\mathcal{W}} := \{ \wbf ~|~ \| \mathbf{\Pi}^{(j)} \wbf\|^2 \leq 1\}$, $p_i^{(j)}$ is the power
assigned to User $i$ in group $\mathcal{G}_j$, and $\Pibf^{(j)}$ is the inter-group ZF projection matrix for
group $\mathcal{G}_j$.  We assume that the total transmit power $P_t$ is divided such that $|\mathcal{G}_j| \times P_t/K$ is allocated to group $\mathcal{G}_j$. Then, from
 \eqref{eq:channelModel}
 the received signal at User $i$ in group $\mathcal{G}_j$ can be written as
\begin{align}
    y_i^{(j)} &= \hbf_i^{(j)H} \left(\Pibf^{(j)} \sum_{i \in \mathcal{G}_j}\sqrt{p_i^{(j)}} \mathbf{w}_i^{(j)}
    s_i^{(j)}\right) + n_i^{(j)} \stackrel{(a)}{=}  \left(\Pibf^{(j)}\hbf_i^{(j)}\right)^H \left(\sum_{i \in \mathcal{G}_j}\sqrt{p_i^{(j)}} \mathbf{w}_i^{(j)}
    s_i^{(j)}\right) + n_i^{(j)}, \label{eq:rxModel2}
\end{align}
where $\hbf_i^{(j)}$ is  the $N \times 1$ channel vector between the transmitter and User $i$ in group $\mathcal{G}_j$, and
$n_i^{(j)}\sim \mathcal{CN}(0,1)$ is the AWGN at User $i$ in group $\mathcal{G}_j$ (here, the single user index $k$ in  \eqref{eq:channelModel} is properly mapped to the two indices: intra-group user index $i$ and group index $j$).  The inter-group ZF projection matrix $\Pibf^{(j)}$ is given by
$\Pibf^{(j)} = \Pibf_{\tilde{\Hbf}_j}^\bot$,
where  $\tilde{\Hbf}_j$ is the matrix composed of all channel vectors except the channel vectors of the users in group $\Gc_j$, i.e.,
\begin{equation}  \label{eq:tildeHbfj}
\tilde{\Hbf}_j:= [\hbf_1^{(1)} \cdots \hbf_{|\mathcal{G}_1|}^{(1)} , \cdots, \hbf_1^{(j-1)} \cdots \hbf_{|\mathcal{G}_{j-1}|}^{(j-1) }, \hbf_1^{(j+1)} \cdots \hbf_{|\mathcal{G}_{j+1}|}^{(j+1) } , \cdots, \hbf_1^{(N_g)} \cdots \hbf_{|\mathcal{G}_{N_g}|}^{(N_g) }  ].
\end{equation}
Due to the inter-group ZF beamforming, there is no inter-group interference in \eqref{eq:rxModel2}, and the property of an orthogonal projection matrix, $(\Pibf_{\tilde{\Hbf}_j}^\bot)^H=\Pibf_{\tilde{\Hbf}_j}^\bot$, is used in Step (a) in \eqref{eq:rxModel2}.

\subsection{Intra-Group Beam Design and the Corresponding Rates}
\label{Beam Design for the Proposed Structure}

In this subsection, we consider intra-group beam vector design for the {}{mixture}
transceiver architecture and analyze the achievable rates of the intra-group processing. First, consider each group $\Gc_j$ with one user. In this case, the received signal \eqref{eq:rxModel2}
reduces to
\begin{equation}  \label{eq:received_1user}
    y_1^{(j)}     =  \left(\Pibf_{\tilde{\Hbf}_j}^\bot \hbf_1^{(j)}\right)^H  \sqrt{p_1^{(j)}} \mathbf{w}_1^{(j)}
    s_1^{(j)} + n_1^{(j)}, \quad |\mathcal{G}_j|  = 1,
\end{equation}
and the design of the  intra-group beam vector  $\wbf_1^{(j)}$ is simple. The optimal intra-group beam vector
$\mathbf{w}_1^{(j)*}$ is the MRT beam matched to the projected effective channel vector $\Pibf_{\tilde{\Hbf}_j}^\bot \hbf_1^{(j)}$, i.e.,
   $\sqrt{p_1^{(j)}} \mathbf{w}_1^{(j)*} = \sqrt{P_t/K} {\Pibf_{\tilde{\Hbf}_j}^\bot \hbf_1^{(j)}}/{\|\Pibf_{\tilde{\Hbf}_j}^\bot \hbf_1^{(j)}\|}$.
In this case, the optimal beam vector is equivalent to the ZF-beamforming vector with power $P_t/K$.

Next, consider the intra-group beam design for each group with more than one user. As aforementioned, we apply
 superposition
coding and SIC in this case. Suppose that group $\mathcal{G}_j$ consists of $L$ users ($L > 1$). Then, with the group  index $(j)$ omitted for convenience, the received signal for User $i$, $i=1,\cdots,L$,  in group $\mathcal{G}_j$ is given by
\begin{equation} \label{eq:BC_l}
    y_i = \gbf_i^H \left(\sum_{i =1}^L\sqrt{p_i} \mathbf{w}_i
    s_i\right) + n_i, \quad i = 1 \cdots, L ,
\end{equation}
where $\sum_{i=1}^L p_i \le  P$ with $P$ being the total group power allocated to  group $\Gc_j$ (i.e., $P = L \times P_t/K$), and
$\gbf_i$ is the projected effective channel of User $i$ given by
\begin{equation}  \label{eq:IntraEffectCh}
    \gbf_i = \Pibf_{\tilde{\Hbf}_j}^\bot \hbf_i^{(j)}, \quad i = 1, \cdots, L.
\end{equation}
We assume that the intra-group beam vector $\wbf_i$ is designed based on the projected effective channels $\gbf_1,\cdots,\gbf_L$.
Then, the feasible set for intra-group beam vector $\wbf_i$ is given by $\mathcal{W} := \{ \wbf ~|~ \| \wbf\|^2 \leq 1\}$ from the
fact that the beam design space for $\wbf_i$ is the linear subspace spanned by $\{\gbf_1, \cdots, \gbf_L\}$. (The beam component not in the subspace spanned by   $\{\gbf_1, \cdots, \gbf_L\}$ does not affect the signal or the interference. Hence, it just wastes power.) Since $\wbf_i \in \Cc([\gbf_1,\cdots,\gbf_L])$, we have $\wbf_i \in \Cc^\perp(\tilde{\Hbf}_j)$ by \eqref{eq:IntraEffectCh} and hence for the actual beam power constraint $\|\Pibf^{(j)} \wbf_i \|^2 \le 1$, we have
$\|\Pibf^{(j)} \wbf_i \|^2=\|  \Pibf_{\tilde{\Hbf}_j}^\bot \wbf_i\|^2 =\|\wbf_i\|^2 \le 1$ in this case. So, we have the feasible set $\Wc$ for $\wbf_i$.

Note that with inter-group ZF beamforming, the intra-group signal model is separated from group to group based on the projected effective channels, and the system model  \eqref{eq:BC_l} is a conventional  MISO BC with $L$-user superposition coding beamforming.
For superposition coding and SIC,
we assume that the in-group users are ordered according to their channel norms as
$\|\gbf_1\|^2 \geq \|\gbf_2\|^2 \geq \cdots  \geq \|\gbf_L\|^2$.
With this assumption, SIC at the in-group receivers is applied such that
User $i$ decodes and cancels the interference from Users $L, L-1, \cdots, i+1$ sequentially.\footnote{{}{The considered decoding order may not be optimal if we consider the  design of $\{\wbf_i\}$ but is sufficient for our purpose of analytic derivation of the diversity order of the mixture transceiver architecture.}} (Note that since User $i$ has a better channel than Users $L,L-1,\cdots,i+1$, User $i$ can decode the messages intended for Users
$L,L-1,\cdots,i+1$.)
 Then, the rates of the in-group users can be expressed  as
\begin{equation} \label{eq:rateLMISOBCSIC1}  
    R_1 =  \log_2 \left( 1 + p_1|\gbf_1^H \wbf_1|^2\right), ~~~
    R_i =  \log_2 \left( 1 +  \min \left\{ \mathrm{SINR}_1^i, \cdots, \mathrm{SINR}_i^i \right\}\right), ~ i = 2, \cdots, L,
\end{equation}
where $\mathrm{SINR}_j^i$ is the SINR  when User $j$ decodes the message intended for User $i$, given by
\begin{equation}
    \mathrm{SINR}_j^i = \frac{p_i|\gbf_j^H \wbf_i|^2}{\sum_{m=1}^{i-1} p_{m}|\gbf_j^H \wbf_{m}|^2 + 1}
\end{equation}
The achievable rate region $\mathcal{R}$ of the
MISO BC with superposition coding and SIC decoding is defined as the union of achievable rate-tuples:
\begin{equation}  \label{eq:AchieRegionMISOBCSIC}
    \mathcal{R} := \bigcup\limits_{ \substack{(\wbf_1, \cdots, \wbf_L) \in \mathcal{W}^L \\ (p_1, \cdots, p_L) | p_i >0, \forall i,~\sum_{i=1}^L p_i = P  } } ( R_1, R_2 \cdots, R_L ),
\end{equation}
where $(R_1,\cdots,R_L)$ is  from \eqref{eq:rateLMISOBCSIC1}. 
The Pareto boundary of the region $\mathcal{R}$ is the outer boundary of $\mathcal{R}$ and  can be obtained
by maximizing $R_L$ for each feasible target rate-tuple $(R_1^* , \cdots, R_{L-1}^*)$.
The maximization problem  for given $(R_1^* , \cdots, R_{L-1}^*)$  can be solved  by a convex programming approach based on reformulation \cite{Hanif16:TSP}
 and the convex concave procedure (CCP)\cite{yuille03:Concave}.
However, difficulty lies in knowing  the feasible target rate-tuple set for the MISO BC with superposition coding and SIC
since the rates depend on the beam vectors and channel vectors of all in-group users, although some induction approach for this  was proposed in \cite{Seo&Sung:18SP}.
{}{The difficulty to find the feasible rate tuple for Users $1,\cdots,L-1$ can be circumvented by formulating the problem as weighted sum rate maximization  based on the rate-profile approach \cite{Zhang&Cui:10SP}. However,} the existing algorithms for the Pareto-optimal design problem  numerically provide rates  based on numerically obtained beam vectors. Hence, these existing design algorithms do not provide  closed-form rate expressions for general MISO BCs with superposition coding and SIC decoding which is necessary for our analytical derivation of the diversity order.
In order {}{to obtain desired closed-form expressions} for  the achievable rates of the  MISO BC with superposition coding and SIC decoding, we consider beam design under the following constraint:
\begin{eqnarray}
\wbf_1&=&\wbf_2=\cdots=\wbf_L= \wbf,~~~||\wbf||^2 \le 1, ~\mbox{i.e.}~\wbf \in \Wc \nonumber\\
p_i &=& \delta_i P, ~~~i=1,2,\cdots,L, \label{eq:optimization_alter_cond}
\end{eqnarray}
where $(\delta_1, \cdots, \delta_L)$ is a power ratio-tuple out of
the feasible power ratio-tuple set
$\Dc:=\{ (\delta_1, \cdots, \delta_L) ~ | ~ \delta_i \ge 0~ \forall i, ~\sum_{i=1}^L \delta_i = 1\}$.
Here, $\delta_i$ is the ratio of the total group power $P$ to the power allocated to
User $i$, i.e., $p_i=\delta_iP$ is assigned to User $i$.  Note that the   constraint \eqref{eq:optimization_alter_cond}
satisfies the original beam design constraint in  \eqref{eq:AchieRegionMISOBCSIC}.
 Based on the restricted  constraint \eqref{eq:optimization_alter_cond},
 the following proposition provides simple closed-form lower bounds on the achievable rates  for
the MISO BC with superposition coding and SIC:

\begin{proposition} \label{pro:proposition2}
    In the MISO BC \eqref{eq:BC_l} with $L$ users adopting superposition coding and SIC decoding with  channel vectors $\gbf_1, \cdots, \gbf_L$ with ordering
$\|\gbf_1\|^2 \geq \|\gbf_2\|^2 \geq \cdots  \geq \|\gbf_L\|^2$ and total group power $P$,  for an arbitrary given power ratio-tuple  $(\delta_1, \cdots, \delta_L)$ out of the feasible power ratio-tuple set $\Dc$, 
the achievable rates  $(R_1, R_2, \cdots, R_L)$   are lower bounded as
\begin{align}
    R_1 &\geq \log_2 \left( 1 + \frac{1}{c} \delta_1 \|\gbf_1\|^2 P \right)  \label{eq:propRate1}\\
    R_i &\geq \log_2 \left( 1 + \frac{\delta_i}{\sum_{m=1}^{i-1} \delta_m} \frac{1}{1 + \left(\frac{1}{c} \|\gbf_i\|^2\sum_{m=1}^{i-1} \delta_m P \right)^{-1}}\right),
    \quad i = 2, \cdots, L, \label{eq:propRatei}
\end{align}
 where the constant $c$ is given by
\begin{equation}  \label{eq:Prop1constanC}
   c = \left\{ \begin{array}{ll}
                  L & \mbox{if}, \quad L \leq 3, \\
                  8L^2 & \mbox{if}, \quad L > 3.
                \end{array}
    \right.
\end{equation}
\end{proposition}

{\em Proof:} See Appendix A.

Note that the power ratio-tuple set $\Dc$ does not
depend on the beam vectors and the channel vectors, and it is just a simplex. Thus, the rate lower bounds \eqref{eq:propRate1}
    and \eqref{eq:propRatei}  with sweeping $(\delta_1,\cdots,\delta_L)$ within $\Dc$
yield an inner region of  the achievable rate region  $\Rc$ defined in \eqref{eq:AchieRegionMISOBCSIC}.

\begin{figure}[ht]
\begin{psfrags}
    \centerline{ \scalefig{0.55} \epsfbox{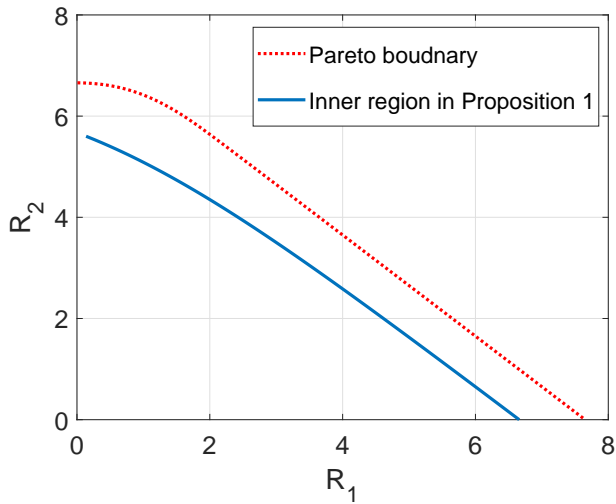} }
    \caption{{Rate region: Pareto-boundary versus Propostion 1 ($K=2$ with $4\times 1$ MISO)}}
    \label{fig:rate_region}
\end{psfrags}
\end{figure}

{The rate lower  bound in Proposition \ref{pro:proposition2} was evaluated and compared with the Pareto-boundary obtained from \eqref{eq:AchieRegionMISOBCSIC} for an example case of a  MISO BC of two users with SIC. The system setup is as follows: It was a MISO BC with four transmit antennas and one receive antenna, $10 \log_{10} P/1 = 10$, $||\gbf_1||^2 = 20$, $||\gbf_2||^2 = 10$, and $|\gbf_1^H \gbf_2|^2/(||\gbf_1||^2 ||\gbf_2||^2) = 0.5$, where $\gbf_1$ is the $4 \times 1$ channel vector of the strong user and $\gbf_2$ is the $4\times 1$ channel vector of the weak user.
The rate region is shown in Fig. \ref{fig:rate_region}. It is seen that the inner rate region by  Proposition \ref{pro:proposition2} is not very close to the Pareto-boundary, but it still achieves quite a good portion of the Pareto-region.}
The key point in the derived lower bounds \eqref{eq:propRate1} and
 \eqref{eq:propRatei}  on the achievable rates $(R_1,\cdots,R_L)$  is that {\em the lower bound on the rate $R_i$ of User $i$ in the superposition-and-SIC group is expressed only in terms of User $i$'s channel norm square $||\gbf_i||^2$ and the power distribution factors $(\delta_1,\cdots,\delta_L)$.} This enables us to analyze the distribution of $R_i$ via the distribution of $||\gbf_i||^2$ and to derive the diversity order of the mixture scheme in Section  \ref{sec:outageAnalysis}.

\subsection{Adaptive User Grouping}
\label{sec:User Grouping Algorithm}

Now, we  consider  user grouping, which should be  done properly for good diversity performance of the mixture transceiver architecture.
Since we apply inter-group ZF beamforming, a level of orthogonality across the constructed groups is required to guarantee high reliability, as discussed in Section \ref{subsec:prelim}.  Note that the  channel orthogonality among the users within a group is not required since  superposition coding and SIC are applied to the users in each group.
There can exist many user grouping methods that guarantee certain orthogonality among the constructed groups. In this section, we provide one example for such user grouping.  {}{The main difference between our user grouping method and several previous user grouping methods proposed for NOMA\cite{Chen&Ding&Dai:16Access,Seo&Sung:18SP,AliHossainKim17Access} is that the number of groups and the number of members in each group are not predetermined and the  angle between the channel subspaces of any two user groups is not less than a certain threshold in our user grouping method, whereas the number of groups and the number of members in each group are predetermined and fixed for the previous methods\cite{Chen&Ding&Dai:16Access,Seo&Sung:18SP,AliHossainKim17Access}. This angle property is necessary for derivation of the diversity order of the mixture architecture.
}

To measure the orthogonality across groups, we define a new subspace angle metric
$\theta(\cdot, \cdot)$, which captures the angle between the subspaces  $\Cc(\Abf)$ and $\Cc(\Bbf)$ spanned
by the columns of matrices $\Abf$ and $\Bbf$ as
\begin{equation}  \label{eq:thetaDef}
    \theta(\Abf,\Bbf) := \left\{\begin{array}{ll}
    \max (\{\phi(\Abf,\bbf_i), \forall~ i\} \cup \{\phi(\Bbf,\abf_j), \forall~ j\} ), & \mbox{if $\Abf$ and $\Bbf$ are non-empty matrices}\\
    0, & \mbox{if}~ \Abf ~\mbox{or}~\Bbf~\mbox{is an empty matrix},
    \end{array}
    \right.
\end{equation}
where $\abf_i$ is the $i$-th column of $\Abf$,  $\bbf_j$ is the $j$-th column of $\Bbf$,  and $\phi(\cdot,\cdot)$ is another newly-defined
 angle metric which captures the angle between the vector $\bbf$ and the subspace $\Cc(\Abf)$, defined as
\begin{equation}  \label{eq:phiAbfbbf}
    \phi(\Abf, \bbf) := \frac{\|\Abf(\Abf^H\Abf)^{-1}\Abf^H \bbf\|^2}{\|\bbf\|^2}.
\end{equation}
In case of $\Abf = [\abf]$ is a vector, $\phi$ reduces to the square of the angle cosine  of two vectors $\abf$ and $\bbf$:
\begin{equation}
    \phi(\abf, \bbf) = \frac{|\abf^H \bbf|^2}{\|\abf\|^2\|\bbf\|^2} = \cos^2 \angle (\abf,\bbf) ~\in [0,1].
\end{equation}
When $\Bbf=[\bbf]$ in   \eqref{eq:thetaDef}  is a vector,  $\theta(\Abf,\bbf)$ simply reduces to $\phi(\Abf, \bbf)$ because $\phi(\Abf,\bbf) \ge \phi(\bbf,\abf_j)$ for all $j$, i.e., the angle between $\bbf$ and $\Cc(\Abf)$ is smaller than or equal to the angle between $\bbf$ and individual column $\abf_j$ of $\Abf$, as illustrated in Fig. \ref{fig:phiDef}(a). When $\theta=0$,  two subspaces $\Cc(\Abf)$ and $\Cc(\Bbf)$ are mutually orthogonal. When $\theta=1$, on the other hand,  there exists either at least a column of $\Abf$ contained in $\Cc(\Bbf)$ or at least a column of $\Bbf$ contained in $\Cc(\Abf)$, and the two subspaces $\Cc(\Abf)$ and $\Cc(\Bbf)$ are not separated.

\begin{figure*}[t!]
\centerline{
 \begin{psfrags}
  \psfrag{a1}[l]{\small $\abf_{1}$} %
  \psfrag{a2}[l]{\small $\abf_{2}$} %
  \psfrag{b}[l]{\small $\bbf$} %
  \psfrag{pb}[l]{\footnotesize $\Pibf_\Abf \bbf$} %
  \psfrag{ca}[l]{\small $\Cc(\Abf=[\abf_1,\abf_2])$} %
  \psfrag{l1}[l]{\small $l_{o}$} %
  \psfrag{l2}[l]{\small $l_{p}$} %
  \psfrag{phi}[l]{\small $\phi(\Abf,\bbf)={l_p^2}/{l_o^2}$} %
  \psfrag{caa}[c]{\footnotesize $\Cc(\Abf)$} %
  \psfrag{cape}[l]{\footnotesize $\Cc^\perp(\Abf)$} %
  \psfrag{cpab}[r]{\footnotesize $\Cc^\perp([\Abf,\Bbf])$} %
  \psfrag{cpabx}[r]{\scriptsize $\Pibf_{[\Abf,\Bbf]}^\perp \xbf$} %
  \psfrag{cb}[l]{\footnotesize $\Cc(\Bbf)$} %
  \psfrag{x}[l]{\small $\xbf$} %
  \psfrag{cab}[l]{\footnotesize $\Cc(\Pibf_\Abf^\perp \Bbf)$} %
  \psfrag{cax}[l]{\scriptsize $\Pibf_\Abf^\perp \xbf$} %
  \psfrag{pipiab}[l]{\scriptsize $\Pibf_{\Pibf_{\Abf}^\perp \Bbf}\Pibf_\Abf^\perp \xbf$} %
  \psfrag{aaa}[c]{\small (a)} %
  \psfrag{bbb}[c]{\small (b)} %
    \scalefig{1}\epsfbox{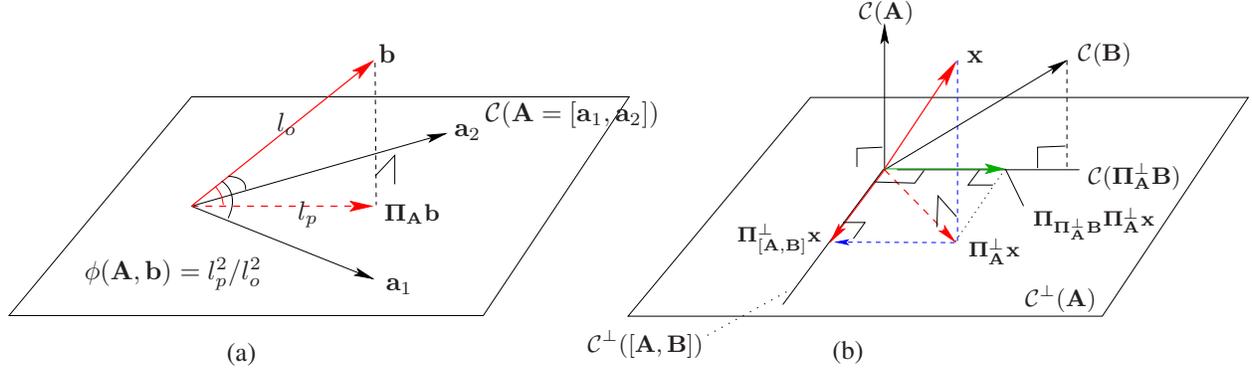}
\end{psfrags}
} \caption{(a) an illustration of $\phi(\Abf,\bbf)$ and (b) an illustration of sequential orthogonal projection}
\label{fig:phiDef}
\end{figure*}

\begin{algorithm}[!t]
        \setstretch{1.2}
   \caption{\textbf{: The Proposed User Grouping Algorithm}}\label{al:algoirhtm1}
    \begin{algorithmic}[1]
       \State \textbf{Initialization}:
       \State A threshold value $\theta^{th} \in (0,1)$ is given.
       \State Initially set $\fbf_1 \cdots, \fbf_K$ as the actual channel vectors $\hbf_1,\cdots,\hbf_K$ of the $K$ users.
       \State Set $\mathcal{K} \leftarrow \{1,\cdots,K\}$ (initial candidate set)
       \State Set $i_g \leftarrow 0$ (group index)
       \State Set $n_g \leftarrow 1$ (number of users in group)
       \State Set $\Fbf_\Kc \leftarrow [\fbf_1, \cdots, \fbf_K]$.
             \vspace{1em}
       \State \textbf{Execution}:

       \State \textbf{While } $n_g \le K$
       \State \quad Find a group of users $\{u_1^*, \cdots, u_{n_g}^*\}$ with cardinality $n_g$ such that $\Cc(\Fbf_{\{u_1^*, \cdots, u_{n_g}^*\}})$ and $\Cc(\Fbf_{ \mathcal{K} \setminus \{u_1^*, \cdots, u_{n_g}^*\}})$ satisfy
                     \begin{equation}  \label{eq:Algor1theta}
                        \theta(\Fbf_{ \mathcal{K} \setminus \{u_1^*, \cdots, u_{n_g}^*\}},\Fbf_{\{u_1^*, \cdots, u_{n_g}^*\}}) \leq \theta^{th}.
                    \end{equation}
       \State \quad \quad\textbf{If} we find such a group of users $\{u_1^*, \cdots, u_{n_g}^*\}$,
       \State \quad \quad\quad\quad $i_g\leftarrow i_g+1$ (increase the group index by one).
       \State \quad \quad\quad\quad $\mathcal{G}_{i_g} \leftarrow \{u_1^*, \cdots, u_{n_g}^*\}$ (construct one group).
         \State \quad \quad\quad\quad $\mathcal{K} \leftarrow \mathcal{K}\setminus \{u_1^*, \cdots, u_{n_g}^*\}$ (update $\Kc$ by removing the selected users from the candidate set).
       \State \quad  \quad\quad\quad Update the vector $\fbf_u$ as $\fbf_u \leftarrow \left(\Ibf - \Fbf_{\mathcal{G}_{i_g}}(\Fbf_{\mathcal{G}_{i_g}}^H\Fbf_{\mathcal{G}_{i_g}})^{-1}\Fbf_{\mathcal{G}_{i_g}}^H\right)\fbf_u, ~ \forall u \in \mbox{updated}~\mathcal{K}$
       \State \quad \quad\quad\quad Construct new $\Fbf_{\mathcal{K}}$ with the updated $\fbf_u, \forall u \in \mbox{updated}~\mathcal{K}$.
       \State \quad \quad\textbf{Else}
       \State \quad  \quad\quad\quad $n_g \leftarrow n_g + 1$
       \State \quad \quad \textbf{Endif}
       \State \textbf{Endwhile}
       \State $N_g \leftarrow i_g$.
       \State (Throughout the algorithm, $\Fbf_{\mathcal{S}}$ means the submatrix of current $\Fbf_\Kc$ composed of $\{\mbox{current}~\fbf_u, \forall u \in \mathcal{S}\subset \Kc\}$.)
\end{algorithmic}
\end{algorithm}

The proposed user grouping algorithm based on $\theta(\cdot,\cdot)$  is presented in Algorithm \ref{al:algoirhtm1}. Before explaining the algorithm, we introduce a useful lemma regarding sequential orthogonal projection necessary to explain the algorithm.

\begin{lemma} \label{lemma:sop}
        For a vector $\xbf$ and matrices $\Abf$ and $\Bbf$ such that $[\Abf,\Bbf]$ is a tall matrix, the following equality holds:
            $\Pibf^\bot_{[\Abf,\Bbf]} \xbf =  (\Ibf-\Pibf_{\Pibf_\Abf^\perp \Bbf})\Pibf_\Abf^\perp \xbf = \Pibf^\bot_\Abf \xbf
            - \Pibf^\bot_{\Abf} \Bbf [(\Pibf^\bot_{\Abf} \Bbf)^H  \Pibf^\bot_{\Abf} \Bbf ]^{-1}(\Pibf^\bot_{\Abf} \Bbf)^H  \Pibf^\bot_\Abf \xbf$.
\end{lemma}

{\em Proof}: See Appendix B.

Lemma  \ref{lemma:sop} states that the projection of $\xbf$ onto the orthogonal space of $\Cc([\Abf,\Bbf])$ can be accomplished in two steps first by projecting $\xbf$ onto the orthogonal space of $\Cc(\Abf)$  and then by projecting this projected vector onto the orthogonal space of $\Cc(\Pibf_\Abf^\perp \Bbf)$ (not $\Cc(\Bbf)$), as illustrated in Fig. \ref{fig:phiDef}(b).   By successively applying Lemma  \ref{lemma:sop}, we can obtain $\Pibf^\bot_{[\Abf_1,\Abf_2,\cdots,\Abf_n]} \xbf$ in a successive manner, where $[\Abf_1,\cdots,\Abf_n]$ is a tall matrix. That is, we first project $\xbf$ onto $\Cc^\perp(\Abf_1)$ to obtain $\Pibf_{\Abf_1}^\perp\xbf$, and project the subspace matrices $\Abf_2,\Abf_3,\cdots,\Abf_n$ onto $\Cc^\perp(\Abf_1)$ to obtain $\Pibf_{\Abf_1}^\perp \Abf_2, \Pibf_{\Abf_1}^\perp \Abf_3, \cdots, \Pibf_{\Abf_1}^\perp \Abf_n$. Then, we project  $\Pibf_{\Abf_1}^\perp\xbf$ onto $\Cc^\perp(\Pibf_{\Abf_1}^\perp \Abf_2)$ to obtain $(\Ibf - \Pibf_{\Pibf_{\Abf_1}^\perp \Abf_2})\Pibf_{\Abf_1}^\perp\xbf $, and also project $\Pibf_{\Abf_1}^\perp \Abf_3, \cdots, \Pibf_{\Abf_1}^\perp \Abf_n$ onto $\Cc^\perp(\Pibf_{\Abf_1}^\perp \Abf_2)$ to obtain $(\Ibf - \Pibf_{\Pibf_{\Abf_1}^\perp \Abf_2})\Pibf_{\Abf_1}^\perp \Abf_3, \cdots, (\Ibf - \Pibf_{\Pibf_{\Abf_1}^\perp \Abf_2})\Pibf_{\Abf_1}^\perp \Abf_n$. Then, we project $(\Ibf - \Pibf_{\Pibf_{\Abf_1}^\perp \Abf_2})\Pibf_{\Abf_1}^\perp\xbf $ onto $\Cc^\perp((\Ibf - \Pibf_{\Pibf_{\Abf_1}^\perp \Abf_2})\Pibf_{\Abf_1}^\perp \Abf_3)$, and project the remaining  subspace matrices  $(\Ibf - \Pibf_{\Pibf_{\Abf_1}^\perp \Abf_2})\Pibf_{\Abf_1}^\perp \Abf_4, \cdots, (\Ibf - \Pibf_{\Pibf_{\Abf_1}^\perp \Abf_2})\Pibf_{\Abf_1}^\perp \Abf_n$ correspondingly. We continue this process until step $n$ is reached. Then, this gives us $\Pibf^\bot_{[\Abf_1,\Abf_2,\cdots,\Abf_n]} \xbf$.

Algorithm \ref{al:algoirhtm1} tries to find single-user groups first (line 6). If the algorithm cannot find any single-user group further, it increases the number of users in group to two (lines 17 and 18), and tries to find two-user groups. It continues this process until $n_g$ becomes $K$ (line 9). Suppose that no group is found up to $n_g = K-1$. Then, at $n_g=K$, one argument in $\theta(\cdot,\cdot)$ in   \eqref{eq:Algor1theta} becomes an empty matrix, $\theta$ becomes zero by the definition \eqref{eq:thetaDef}, and hence the condition \eqref{eq:Algor1theta} is satisfied. Thus, in this case the whole set $\{1,2,\cdots,K\}$ becomes a single group. Let us explain Algorithm \ref{al:algoirhtm1} by using a specific example  below:

\begin{example}  \label{example:SOP}
Suppose that initial $\Kc=\{1,2,3,4,5,6,7\}$ and  suppose that initially user $1$ satisfies $\theta([\hbf_2,\cdots,\hbf_7],\hbf_1) \leq \theta^{th}$. Then, we update $\Gc_1 = \{1\}$ (line 13) and   $\Kc=\{2,3,4,5,6,7\}$ (line 14), and project the channel vectors $\hbf_2,\cdots,\hbf_7$ onto $\Cc^\perp([\hbf_1])$ to obtain the projected channel vectors $\Pibf_{\hbf_1}^\perp \hbf_2, \cdots,$ $\Pibf_{\hbf_1}^\perp \hbf_7$ (line 15).  Next, suppose that $\theta([\Pibf_{\hbf_1}^\perp \hbf_3,\cdots, \Pibf_{\hbf_1}^\perp \hbf_7], \Pibf_{\hbf_1}^\perp \hbf_2) \le \theta^{th}$. (Note that at this point we compute $\theta(\cdot,\cdot)$ using the {\em projected} channels (lines 15 and 16).)  Then, we update $\Gc_2 =\{2\}$ and $\Kc=\{3,4,5,6,7\}$ and project $\Pibf_{\hbf_1}^\perp \hbf_3, \cdots, \Pibf_{\hbf_1}^\perp \hbf_7$ onto $\Cc^\perp(\Pibf_{\hbf_1}^\perp\hbf_2)$ to obtain the further projected channels $(\Ibf - \Pibf_{\Pibf_{\hbf_1}^\perp \hbf_2})\Pibf_{\hbf_1}^\perp \hbf_3, \cdots, (\Ibf - \Pibf_{\Pibf_{\hbf_1}^\perp \hbf_2})\Pibf_{\hbf_1}^\perp \hbf_7$.  Now, suppose that we cannot find a single-user group further and that at $n_g=2$ only one pair of users $\{3,4\}$ satisfies $\theta ( [(\Ibf - \Pibf_{\Pibf_{\hbf_1}^\perp \hbf_2})\Pibf_{\hbf_1}^\perp \hbf_5, \cdots, (\Ibf - \Pibf_{\Pibf_{\hbf_1}^\perp \hbf_2})\Pibf_{\hbf_1}^\perp \hbf_7] ,[(\Ibf - \Pibf_{\Pibf_{\hbf_1}^\perp \hbf_2})\Pibf_{\hbf_1}^\perp \hbf_3,  (\Ibf - \Pibf_{\Pibf_{\hbf_1}^\perp \hbf_2})\Pibf_{\hbf_1}^\perp \hbf_4]  ) \le \theta^{th}$. Then, we update $\Gc_3=\{3,4\}$ and  $\Kc=\{5,6,7\}$, and  the further projected channels for users $\{5,6,7\}$ are obtained by projecting $(\Ibf - \Pibf_{\Pibf_{\hbf_1}^\perp \hbf_2})\Pibf_{\hbf_1}^\perp \hbf_5, \cdots, (\Ibf - \Pibf_{\Pibf_{\hbf_1}^\perp \hbf_2})\Pibf_{\hbf_1}^\perp \hbf_7$ onto $\Cc^\perp([(\Ibf - \Pibf_{\Pibf_{\hbf_1}^\perp \hbf_2})\Pibf_{\hbf_1}^\perp \hbf_3,  (\Ibf - \Pibf_{\Pibf_{\hbf_1}^\perp \hbf_2})\Pibf_{\hbf_1}^\perp \hbf_4])$. These final projected channels for users $\{5,6,7\}$ are the same as the ZF projected channels $\Pibf_{[\hbf_1,\cdots,\hbf_4]}^\perp \hbf_5, \cdots,$ $\Pibf_{[\hbf_1,\cdots,\hbf_4]}^\perp \hbf_7$ by Lemma  \ref{lemma:sop}. At the next iteration, $n_g$ becomes 3 since we assumed that there is no further two-user group;  one argument of $\theta(\cdot,\cdot)$ becomes an empty matrix since $\Kc=\{5,6,7\}$ and $n_g=3$; hence $\Gc_4=\{5,6,7\}$;  no user is left in the candidate set $\Kc$ after update (line 14); no further channel projection in line 15 occurs since updated $\Kc=\emptyset$; and the algorithm stops.

Now, let us consider the norm  property of the projected ZF channels associated with the constructed groups in the example, which is the key aspect of the proposed user grouping algorithm. Consider user 1 in  firstly-constructed $\Gc_1$.  Since $\theta([\hbf_2,\cdots,\hbf_7],\hbf_1) \leq \theta^{th}$, by the definition of $\theta(\cdot,\cdot)$ in \eqref{eq:thetaDef}, we have
\begin{equation} \label{eq:GroupingExplain1}
 \phi(\tilde{\Hbf}_1=[\hbf_2,\cdots,\hbf_7],\hbf_1) = \frac{\|\tilde{\Hbf}_1(\tilde{\Hbf}_1^H\tilde{\Hbf}_1)^{-1}\tilde{\Hbf}^H_1 \hbf_1\|^2}{\|\hbf_1\|^2} \le \theta^{th}
\end{equation}
Hence, we have
\begin{align*}
||\Pibf_{\tilde{\Hbf}_1}^\perp \hbf_1||^2 &= ||(\Ibf - \tilde{\Hbf}_1(\tilde{\Hbf}_1^H\tilde{\Hbf}_1)^{-1}\tilde{\Hbf}^H_1 )\hbf_1||^2 \\
&=(1- \phi(\tilde{\Hbf}_1,\hbf_1))||\hbf_1||^2~\mbox{by the Pythagorean theorem}\\
&\ge (1-\theta^{th})||\hbf_1||^2.
\end{align*}
Next, consider the norm of the ZF effective channel for User 2 in $\Gc_2$.  Due to the construction of $\Gc_1$ based on \eqref{eq:GroupingExplain1}, $\hbf_1$ and $\hbf_2$ satisfy the following:
\begin{align}
||\Pibf_{\hbf_1}^\perp \hbf_2||^2 &= (1-\phi(\hbf_1,\hbf_2))||\hbf_2||^2 \nonumber \\
&\ge(1- \phi(\tilde{\Hbf}_1,\hbf_1))||\hbf_2||^2, ~~\mbox{since $\tilde{\Hbf}_1$ includes $\hbf_2$} \nonumber\\
&\ge (1-\theta^{th})||\hbf_2||^2. \label{eq:GroupingExaplain2}
\end{align}
By Lemma \ref{lemma:sop}, $\Pibf_{[\hbf_1,\hbf_3,\cdots,\hbf_7]}^\perp\hbf_2$ can be obtained by sequential orthogonal projection as
\[
\Pibf_{[\hbf_1,\hbf_3,\cdots,\hbf_7]}^\perp\hbf_2=(\Ibf - \Pibf_{[\Pibf_{\hbf_1}^\perp\hbf_3,\cdots,\Pibf_{\hbf_1}^\perp\hbf_7]})   \Pibf_{\hbf_1}^\perp\hbf_2,
\]but $\Gc_2$ was constructed such that $\Pibf_{\hbf_1}^\perp\hbf_2$ and $[\Pibf_{\hbf_1}^\perp\hbf_3,\cdots,\Pibf_{\hbf_1}^\perp\hbf_7]$ satisfy the threshold $\theta^{th}$ requirement.  Combining this fact and \eqref{eq:GroupingExaplain2}, we have
\begin{align*}
||\Pibf_{[\hbf_1,\hbf_3,\cdots,\hbf_7]}^\perp\hbf_2||^2 &\ge(1- \theta^{th})^2||\hbf_2||^2.
\end{align*}
Then, consider  User 3 in $\Gc_3=\{3,4\}$. (The same applies to User 4 in $\Gc_3$.) By Lemma \ref{lemma:sop}, we have
\begin{align}
\Pibf_{[\hbf_1,\hbf_2]}^\perp\hbf_3&=  (\Ibf - \Pibf_{\Pibf^\perp_{\hbf_1}\hbf_2}) \Pibf_{\hbf_1}^\perp \hbf_3  \label{eq:exampleu31}\\
\Pibf_{[\hbf_1,\hbf_2,\hbf_5,\hbf_6,\hbf_7]}^\perp\hbf_3 &= (\Ibf - \Pibf_{[\Pibf_{[\hbf_1,\hbf_2]}^\perp\hbf_5,\Pibf_{[\hbf_1,\hbf_2]}^\perp\hbf_6,\Pibf_{[\hbf_1,\hbf_2]}^\perp\hbf_7]}) \Pibf_{[\hbf_1,\hbf_2]}^\perp\hbf_3. \label{eq:exampleu32}
\end{align}
In \eqref{eq:exampleu31}, $\Gc_1=\{1\}$ was constructed such that $\hbf_1$ and $\hbf_3$ satisfy the angle constraint, and $\Gc_2=\{2\}$ was constructed such that  $\Pibf^\perp_{\hbf_1}\hbf_2$ and  $\Pibf_{\hbf_1}^\perp \hbf_3$ satisfy the angle constraint. Hence,   we have $||\Pibf_{\hbf_1}^\perp \hbf_3||^2 \ge (1-\theta^{th})^2 ||\hbf_3||^2$. Furthermore, in \eqref{eq:exampleu32}, $\Gc_3=\{3,4\}$ was constructed such that
$[\Pibf_{[\hbf_1,\hbf_2]}^\perp\hbf_3,\Pibf_{[\hbf_1,\hbf_2]}^\perp\hbf_4]$ and the remaining $[\Pibf_{[\hbf_1,\hbf_2]}^\perp\hbf_5,\Pibf_{[\hbf_1,\hbf_2]}^\perp\hbf_6,\Pibf_{[\hbf_1,\hbf_2]}^\perp\hbf_7]$ satisfy the angle constraint. Combining these facts, we have
\begin{equation}
||\Pibf_{[\hbf_1,\hbf_2,\hbf_5,\hbf_6,\hbf_7]}^\perp\hbf_k||^2 \ge (1-\theta^{th})^3||\hbf_k||^2, ~~k=3,4.
\end{equation}
Finally, consider the norm of the ZF effective channels $\Pibf_{[\hbf_1,\cdots,\hbf_4]}^\perp \hbf_5$, $\cdots$, $\Pibf_{[\hbf_1,\cdots,\hbf_4]}^\perp \hbf_7$ of the last group $\Gc_4=\{5,6,7\}$. These vectors are obtained by three sequential orthogonal projections based on Lemma  \ref{lemma:sop}, and at each projection stage the threshold $\theta^{th}$ was kept for group splitting. Hence, we have
\begin{align*}
||\Pibf_{[\hbf_1,\cdots,\hbf_4]}^\perp \hbf_k||^2 &\ge(1- \theta^{th})^3||\hbf_k||^2, ~k=5,6,7.
\end{align*}
\end{example}
Note that in general the proposed user grouping algorithm satisfies the following norm reduction property for the ZF effective channels:
\begin{equation}
||\gbf_i^{(j)}||^2=||\Pibf_{\tilde{\Hbf}_j}^\perp \hbf_i^{(j)}||^2 \ge (1-\theta^{th})^{N_g-1} ||\hbf_i^{(j)}||^2,
\end{equation}
where $\Pibf_{\tilde{\Hbf}_j}^\perp$ is the ZF projection matrix for group $\Gc_j$, $\hbf_i^{(j)}$ is the channel vector of User $i$ in group $\Gc_j$, and $N_g$ is the number of constructed groups, which is bounded by $K$.  Since the number  of antennas $N$ and the number of users  $K ~(\le N)$  are fixed in  our MISO BC model with superposition coding and SIC,  the lower bound $(1-\theta^{th})^{K-1} \in (0,1)$ of $(1-\theta^{th})^{N_g-1}$ is a constant.

Now, let us define a useful quantity for further exposition:  We define the degrees of freedom of a fading channel $\hbf$  as
\begin{equation}
d:= \lim_{x \rightarrow 0} \frac{  \log \mathrm{Pr} ( \|\hbf\|^2  \leq x )}{ \log x}.
\end{equation}
This quantity captures the behavior of the tail probability of the random variable $||\hbf||^2$ in its lower tail, and the degrees of freedom $d$ for $\hbf$ means that $\mathrm{Pr} ( \|\hbf\|^2  \leq x )$ behaves as $x^{d}+o(x^{d})$, as $x\rightarrow 0$. This quantity is directly related to the diversity order of the SISO communication  channel with the channel gain $||\hbf||$. For example, a Rayleigh fading channel $\hbf \sim \Cc({\mathbf{0}},2\Ibf_N)$ has the degrees of freedom $N$ since
\begin{equation}  \label{eq:PropDistibutionPr}
\mathrm{Pr} ( ||\hbf||^2 \le x) = \int_0^x  f_{\|\hbf\|^2}(z) dz = \frac{1}{2^N N!}x^N + o(x^N), ~\mbox{as}~x\rightarrow 0
\end{equation}
and
$\lim_{x \rightarrow 0} \frac{  \log \mathrm{Pr} ( \|\hbf\|^2  \leq x )}{ \log x} = N$,
where $f_{\|\hbf\|^2}(z) dz$ is given in \eqref{eq:ChiSquare2N}.
Finally, we provide the main statement of this subsection regarding the degrees of freedom of the ZF effective channels associated with the proposed grouping method in the following proposition:

\begin{proposition} \label{pro:distribution}
With the {}{mixture} transceiver architecture and the user grouping method in Algorithm \ref{al:algoirhtm1}, the  projected effective channel $\gbf_j^{(i)} = \Pibf_{\tilde{\Hbf}_j}^\bot \hbf_i^{(j)}$ in \eqref{eq:rxModel2} resulting from inter-group ZF beamforming has the same degrees of freedom as the original channel $\hbf_i^{(j)}$, i.e.,
\begin{equation}  \label{eq:PropDistMain}
    \lim_{x \rightarrow 0} \frac{  \log \mathrm{Pr} ( \|\gbf_i^{(j)}\|^2  \leq x )}{ \log x} = \lim_{x \rightarrow 0} \frac{  \log \mathrm{Pr} ( \|\hbf_i^{(j)}\|^2  \leq x )}{ \log x}, ~~\forall~i,j.
\end{equation}
\end{proposition}

{\em Proof:} See Appendix C.

{}{{\em Complexity of Algorithm \ref{al:algoirhtm1}}:  Note that in the worst case the number of group searches is given by $K+\left(\begin{array}{c}K\\2\end{array} \right)+ \left(\begin{array}{c}K\\3\end{array} \right) + \cdots + \left(\begin{array}{c}K\\K\end{array} \right)$, which scales as $K^{K/2}$. For each group search, we need to compute $\theta(\cdot,\cdot)$ in  \eqref{eq:Algor1theta}, which requires inversion of $K\times K$ matrices in the worst case (see  \eqref{eq:Algor1theta} and the term $(\Abf^H\Abf)^{-1}$ in \eqref{eq:phiAbfbbf}).  Thus, Algorithm \ref{al:algoirhtm1} is not scalable for large $K$. Nevertheless, the algorithm is devised to prove the diversity-order optimality of the mixture architecture in this paper. Invention of more efficient user grouping algorithms for the mixture architecture for MISO BCs is a future work.
For one possible idea {for polynomial complexity}, please see Appendix E.
}

{{\em SIC Complexity}: Since in the proposed adaptive user grouping, each group can have one to $K$ members, it is required that each receiver be able to handle  SIC of $K-1$ users in the worst case. SIC for a general number of users has been investigated extensively for code-division multiple access systems \cite{JileiHou}.}

\section{Outage Analysis and Diversity Order of The Mixture Scheme}
\label{sec:outageAnalysis}

In this section, we present our main result  regarding the diversity order of the mixture transceiver architecture for MISO BCs.

\begin{theorem}  \label{theo:MainResult}
For the Gaussian MISO BC with $N$ transmit antennas and $K$ single-antenna users with independent Rayleigh fading described in Section \ref{subsec:ChannelModel},  let the channels be ordered as $\|\mathbf{h}_1\|^2 \geq \|\mathbf{h}_2\|^2 \cdots\geq  \| \mathbf{h}_K\|^2$ and let the $k$-th user be the user with the $k$-th largest channel norm.
Then, the diversity order for the  $k$-th user achievable by  {}{the mixture transceiver architecture with proper user grouping} is given by
\begin{equation}
   D_k = N \times (K - k + 1).
\end{equation}
Here, the  diversity order is defined as
$D_k := \lim_{P_t \rightarrow \infty} - \frac{\log Pr\{ R_k < R^{th} \}}{\log P_t }$,
where $R_k$ is the rate of the $k$-th user and $R^{th}$ is a  rate threshold.
Note that $P_t \rightarrow \infty$ is equivalent to $\mathrm{SNR}=P_t/\sigma^2\rightarrow \infty$ since we set the noise variance  $\sigma^2=1$ for simplicity.
\end{theorem}

\begin{proof} {}{For user grouping of the mixture architecture we adopt Algorithm \ref{al:algoirhtm1}.} The diversity provided by such a grouping will provide an achievable bound for the diversity as claimed in the theorem.  Proof is based on Propositions \ref{pro:proposition2} and   \ref{pro:distribution}.
In proof, we consider not only the distribution of the channel norm itself but also the order statistics resulting from the channel norm ordering.
    With the descending channel ordering $\|\mathbf{h}_1\|^2 \geq \|\mathbf{h}_2\|^2 \cdots\geq  \| \mathbf{h}_K\|^2 $,  the pdf of the $k$-th channel norm square  is given by order statistics as
\begin{equation}
        f_{\|\hbf_k\|^2}(x) =  \frac{K!}{(k-1)! (K-k)!} [F_{\|\hbf\|^2}(x)]^{K-k}[1- F_{\|\hbf\|^2}(x)]^{k-1} f_{\|\hbf\|^2}(x) \label{eq:pro2}
\end{equation}
where $f_{\|\hbf\|^2}(\cdot)$ and $F_{\|\hbf\|^2}$ are the pdf and cumulative distribution function (cdf) of chi-square distribution with degree of freedom $2 N$:
\begin{align}
            f_{\|\hbf\|^2}(x) &= \frac{1}{2^{N} (N - 1)!} x^{N-1} e^{-x/2}  = \frac{1}{2^{N} N!} x^{N - 1} + o( x^{N - 1}), \quad \mbox{as}~ x \rightarrow 0 \label{eq:pro1}\\
            F_{\|\hbf\|^2}(x) &= \frac{1}{2^{N} N!}   x^{N} + o( x^{N}), \quad \mbox{as}~ x \rightarrow 0.
\end{align}
Hence, we have for the $k$-th largest channel norm square $||\hbf_k||^2$
\begin{equation} \label{eq:degree_orderHbfk1}
f_{||\hbf_k||^2}(x) = c_k x^{N(K-k+1)-1} + o(x^{N(K-k+1)-1}), \quad \mbox{as}~ x \rightarrow 0,
\end{equation}
and thus
\begin{equation}\label{eq:degree_orderHbfk2}
        \lim_{x \rightarrow 0} \frac{ \log \mathrm{Pr} (\|\hbf_k\|^2 \leq x ) }{\log x} = N(K -k + 1).
\end{equation}
 The outage probability of the $k$-th user is expressed as
 \begin{align}
            \mathrm{Pr}(R_k < R^{th}) &= \sum_{j=1}^{N_g} \left.\left[\mathrm{Pr}( k \in \Gc_j )  \cdot  \mathrm{Pr} \left( R_k < R^{th} \right|  k \in \Gc_j \right)\right]  \\
                             &= \sum_{j=1}^{N_g} \left[ \mathrm{Pr}( k \in \Gc_j ) \cdot  \left\{ \mathrm{Pr}\left(|\Gc_j|  = 1 ~|~k \in \Gc_j\right) \cdot   \mathrm{Pr}\left(R_k < R^{th} ~|~ |\Gc_j|=1, k \in \Gc_j \right) \right.  \right.  \nonumber\\
            & \quad \quad \quad   \left. \left. +  \mathrm{Pr}( |\Gc_j| \neq 1 ~|~ k \in \Gc_j ) \cdot \mathrm{Pr}\left(R_k < R^{th} ~ | ~  |\Gc_j| \neq 1 ,k \in \Gc_j \right)  \right\} \right]. \label{eq:ProofTheo1_1}
 \end{align}

{\em i) Lower bound on the outage probability:}  We obtain a lower bound on the outage probability by considering only the event that the $k$-th user belongs to a group with cardinality one, i.e., the first term in the RHS of \eqref{eq:ProofTheo1_1}.
\begin{align}
            \mathrm{Pr}(R_k < R^{th}) &\geq \sum_{j=1}^{N_g} \left. \mathrm{Pr}( k \in \Gc_j )  \cdot    \mathrm{Pr}\left(|\Gc_j \right.|  = 1 ~|~k \in \Gc_j\right) \cdot   \left.\mathrm{Pr}\left(R_k < R^{th} ~ \right| ~ |\Gc_j|=1, k \in \Gc_j \right)   \nonumber \\
                              &= \sum_{j=1}^{N_g} \mathrm{Pr}\left( | \Gc_j | = 1 , k \in \Gc_j\right) \cdot   \left.\mathrm{Pr}\left(R_k < R^{th} ~ \right| ~ |\Gc_j|=1, k \in \Gc_j \right)   \nonumber  \\
                              &=  \sum_{j=1}^{N_g}  \mathrm{Pr}\left( | \Gc_j | = 1 , k \in \Gc_j\right) \cdot \mathrm{Pr}\left(\|\Pibf^{(j)} \hbf_k\|^2 < K \cdot (2^{R^{th}} -1) \cdot P_t^{-1} \right), \label{eq:Theo1Proof78}
\end{align}
where  \eqref{eq:Theo1Proof78} holds due to  the rate $R_k = \log( 1 + P_t ||\Pibf^{(j)} \hbf_k||^2/K)$ for a single-user group based on \eqref{eq:received_1user}  and the corresponding optimal beam. 
Then, we have
\begin{align}
         -D_k  &=\lim_{P_t \rightarrow \infty} \frac{ \log \mathrm{Pr}(R_k < R^{th})}{\log P_t} \label{eq:lower_start} \\
           &  \geq
           \lim_{P_t \rightarrow \infty} \frac{\log \left(\sum_{j=1}^{N_g} \left[ \mathrm{Pr}\left( | \Gc_j | = 1 , k \in \Gc_j\right) \cdot \mathrm{Pr}\left(\|\Pibf^{(j)} \hbf_k\|^2 < K \cdot (2^{R^{th}} -1) \cdot P_t^{-1} \right) \right]\right) }{ \log P_t} \\
             &  =
           \lim_{P_t^{-1} \rightarrow 0} \frac{\log \left(\sum_{j=1}^{N_g} \left[ \mathrm{Pr}\left( | \Gc_j | = 1 , k \in \Gc_j\right) \cdot \mathrm{Pr}\left(\|\Pibf^{(j)} \hbf_k\|^2 < K \cdot (2^{R^{th}} -1) \cdot P_t^{-1} \right) \right]\right) }{ -\log P_t^{-1}} \\
           &  = - N (K -k + 1). \label{eq:lower_end}
\end{align}
Here, \eqref{eq:lower_end} is valid because $||\hbf_k||^2$ has the channel order $N(K-k+1)$ by \eqref{eq:degree_orderHbfk1}
and \eqref{eq:degree_orderHbfk2};  the projected effective channel $\|\Pibf^{(j)} \hbf_k\|^2$ has the same channel order as $||\hbf_k||^2$ by Proposition \ref{pro:distribution}; and the linear combination of terms with the same order has the same order as each term. Note that $\mathrm{Pr}\left( | \Gc_j | = 1 , k \in \Gc_j\right)$ depends only on the joint distribution of $(\hbf_1,\cdots,\hbf_k)$ for the given user grouping algorithm not on the power $P_t$.

{\em ii) Upper bound on the outage probability:}

For the upper bound, we need to include the second term in the RHS of \eqref{eq:ProofTheo1_1} in addition to the first term in the RHS of \eqref{eq:ProofTheo1_1} considered in the lower bound. The second term in the RHS of \eqref{eq:ProofTheo1_1} is given by
\begin{align}
           & \sum_{j=1}^{N_g}  \mathrm{Pr} ( k \in \Gc_j )  \cdot    \mathrm{Pr} \left(|\Gc_j |  \neq 1 ~|~k \in \Gc_j\right) \cdot   \mathrm{Pr}\left(R_k < R^{th} ~ | ~ |\Gc_j| \neq 1, k \in \Gc_j \right)    \\
                              &\quad = \sum_{j=1}^{N_g}  \mathrm{Pr} \left( | \Gc_j | \neq 1 , k \in \Gc_j\right) \cdot   \mathrm{Pr} \left(R_k < R^{th} ~ | ~ |\Gc_j| \neq 1, k \in \Gc_j \right)
\end{align}
\begin{align}
                              &\quad = \sum_{j=1}^{N_g}\sum_{\ell =2}^{K} \mathrm{Pr} \left( | \Gc_j | = \ell , k \in \Gc_j\right) \cdot    \underbrace{\mathrm{Pr} \left(R_k < R^{th} ~ \left. \right| ~ |\Gc_j| = \ell, k \in \Gc_j \right)}_{(a)} .  \label{eq:Theo1Upper1}
\end{align}
Define the following notations:
\begin{align}
            E_{k,j,i} &:= \mbox{Event that the }  k\mbox{-th user is the $i$-th largest channel norm user in } \Gc_{j} \\
            P_{k,j,i} &:= \mathrm{Pr} (E_{k,j,i} ).
\end{align}
With these notations, the term (a) in \eqref{eq:Theo1Upper1} can be rewritten as
\begin{align}
\left.\mathrm{Pr}\left(R_{k} < R^{th} ~ \right| ~ |\Gc_j| = \ell , k \in \Gc_j \right) &= \sum_{i=1}^\ell P_{k,j,i} \cdot \left.\mathrm{Pr}\left(R_k < R^{th} ~ \right| ~ |\Gc_j| = \ell , k \in \Gc_j ,E_{k,j,i}  \right), \label{equaiton_conti}
\end{align}
where $R_{k}$ conditioned on the joint event $(|\Gc_j| = \ell , k \in \Gc_j ,E_{k,j,i})$ is lower bounded by Proposition \ref{pro:proposition2} as
\begin{align}
            R_{k}  &\geq \left\{   \begin{array}{ll}
            \log_2 \left( 1 + \frac{1}{c} \delta_1^{(j)}\|\Pibf^{(j)} \hbf_k\|^2  \frac{\ell P_t}{K}  \right) & \quad \mbox{if} ~ i = 1 \\
            \log_2 \left( 1 + \frac{\delta_i^{(j)}}{\sum_{m=1}^{i-1} \delta_{m}^{(j)}} \frac{1}{1 + \left(\frac{1}{c} \|\Pibf^{(j)} \hbf_k\|^2\sum_{m=1}^{i-1} \delta_{m}^{(j)} \frac{\ell P_t}{K} \right)^{-1}}\right) & \quad  \mbox{if} ~ i = 2 \cdots \ell.
            \end{array}
            \right.
\end{align}
where $c$ is given in \eqref{eq:Prop1constanC}, and $(\delta_1^{(j)}, \delta_2^{(j)}, \cdots \delta_\ell^{(j)})$ is the power ratio-tuple in group $\Gc_j$, i.e., power $\delta_i^{(j)}\ell P_t /K$ is assigned to User $i$ in group $\Gc_j$. ($\ell P_t/K$ is the total group power for group $\Gc_j$ with $|\Gc_j|=\ell$.) Therefore, the probability \eqref{equaiton_conti} is upper bounded as
\begin{align}
& \sum_{i=1}^\ell \left[P_{k,j,i} \cdot \left.\mathrm{Pr}\left(R_{k} < R^{th} ~ \right| ~ |\Gc_j| = \ell , k \in \Gc_j, E_{k,j,i}  \right) \right] \label{eq:TheoProof98} \\
&\quad \leq P_{k,j,1} \cdot \mathrm{Pr}\left(\log_2 \left( 1 + \frac{1}{c} \delta_1^{(j)}\|\Pibf^{(j)} \hbf_k\|^2  \frac{ \ell P_t}{K}  \right) < R^{th}\right) \nonumber \\
& \quad \quad + \sum_{i=2}^\ell \left[P_{k,j,i} \cdot \Pr \left( \log_2 \left( 1 + \frac{\delta_{i}^{(j)}}{\sum_{m=1}^{i-1} \delta_{m}^{(j)}} \frac{1}{1 + \left(\frac{1}{c} \|\Pibf^{(j)} \hbf_k\|^2\sum_{m=1}^{i-1} \delta_{m}^{(j)} \ell P_t/K \right)^{-1}}\right)< R^{th}\right) \right] \label{eq:Theo1proof99}\\
&\quad = P_{k,j,1} \cdot \Pr\left(\|\Pibf^{(j)} \hbf_k\|^2  < (2^{R^{th}} - 1) \cdot \frac{c}{\delta_1^{(j)}} \cdot  \frac{K}{\ell} P_t^{-1}\right) \nonumber \\
& \quad \quad + \sum_{i=2}^\ell \left[P_{k,j,i} \cdot \Pr \left( \|\Pibf^{(j)} \hbf_k\|^2 <  c \left( \frac{\delta_{i}^{(j)}}{2^{R^{th}}-1} - \sum_{m=1}^{i-1} \delta_{m}^{(j)} \right)^{-1}\cdot  \frac{K}{l} \cdot P_t^{-1} \right) \right], \label{outage_morethan2}
\end{align}
where the threshold for $\|\Pibf^{(j)} \hbf_k\|^2$ in the second term in \eqref{outage_morethan2} is obtained by  manipulation of the second term in \eqref{eq:Theo1proof99}.
By Lemma \ref{lem:ExistDeltasForTheo1} in Appendix D, there always exists a collection of in-group power distribution factors $(\delta_1^{(j)},\cdots,\delta_\ell^{(j)})$ such that $( \frac{\delta_{i}^{(j)}}{2^{R^{th}}-1} -$ $ \sum_{m=1}^{i-1} \delta_{m}^{(j)} )$ in  \eqref{outage_morethan2} is strictly positive for all $i=2,\cdots, \ell$. Set $(\delta_1^{(j)},\cdots,\delta_\ell^{(j)})$ as one of such collections.
Then, each probability term in  \eqref{outage_morethan2} behaves as $P_t^{-N(K-k+1)}$ as $P_t\rightarrow \infty$, since  $\|\Pibf^{(j)} \hbf_k\|^2$ has the same degrees of freedom of $N(K-k+1)$ in \eqref{eq:degree_orderHbfk1}
and \eqref{eq:degree_orderHbfk2}  as  $||\hbf_k||^2$ by Proposition \ref{pro:distribution}.  Hence, their linear combination \eqref{eq:TheoProof98} behaves as $P_t^{-N(K-k+1)}$ as $P_t\rightarrow \infty$, and furthermore the term \eqref{eq:Theo1Upper1} as a linear combination of terms  \eqref{eq:TheoProof98}   behaves as $P_t^{-N(K-k+1)}$ as $P_t\rightarrow \infty$.  Now, by adding \eqref{eq:Theo1Upper1} and the term in  \eqref{eq:Theo1Proof78}, we have the exact outage probability.
 We already showed that the term in  \eqref{eq:Theo1Proof78} behaves as $P_t^{-N(K-k+1)}$ as $P_t\rightarrow \infty$. Furthermore, the upper bound of \eqref{eq:Theo1Upper1} behaves as $P_t^{-N(K-k+1)}$ as $P_t\rightarrow \infty$. Hence, we have
 $-D_k  =\lim_{P_t \rightarrow \infty} \frac{ \log \mathrm{Pr}(R_k < R^{th})}{\log P_t} \le -N(K-k+1)$.
  Combining this upper bound result with the lower bound result, we have
\begin{equation}
             N (K - k + 1)    \leq   D_k= - \lim_{P_t \rightarrow \infty} \frac{\log \Pr(R_k < R^{th})}{\log P_t} \leq N (K - k + 1).
\end{equation}
\end{proof}

\begin{corollary}  \label{cor:systemDorder}
For the Gaussian MISO BC with $N$ transmit antennas and $K$ single-antenna users with independent Rayleigh fading described in Section \ref{subsec:ChannelModel},  the diversity order of the overall system  achievable by  the {}{mixture transceiver architecture with the proposed user grouping method} is given by
\begin{equation}
   D = N
\end{equation}
\end{corollary}
\begin{proof}
The decay rate of the overall outage probability is dominated by the worst decay rate.   The worst diversity order in Theorem  \ref{theo:MainResult} occurs  when $k=K$, and is given by $N$.
\end{proof}
Note that
the diversity order of the full ZF downlink beamforming is given by \cite{Chen&Ding&Dai:16Access}
\begin{equation}  \label{eq:diversityOrderZF}
D=N-K+1.
\end{equation}
Hence, a significant improvement in the diversity order is attained by the mixture scheme.
{}{Note that the possible maximum diversity order for user $k$ with channel $\hbf_k \sim \Cc\Nc(0,2\Ibf)$ is simply $N$.} Hence,
the mixture transceiver architecture  achieves the full diversity order $N$ in MISO BCs.

\subsection{Diversity and Multiplexing Trade-off}

With the cluster power factors $\{\delta_i > 0,i=1,\cdots,L\}$ fixed, as the total cluster power $P$ increases according to \eqref{eq:optimization_alter_cond} without bound,
 in each group only the rate of the first user scales as $\log$ SNR but the rates of all other users saturate to constants:
$\bar{R}_i^{(j)} = \log_2 \left( 1+ \frac{\delta_i^{(j)}}{\sum_{m=1}^{i-1}\delta_m^{(j)}} \right)$, $i=2,\cdots,L$, as seen in  \eqref{eq:propRatei}.\footnote{The target rates of User $2,\cdots, L$ for group $\Gc_j$ should be less than $\bar{R}_i^{(j)}$,  but $\delta_i$'s can be designed for a common target rate $R_{th}$ based on  \eqref{eq:deltaSolution}. Please see Section \ref{subsec:num_MG}.}
Hence, the multiplexing gain for one user group with superposition and SIC is one regardless of the number of users in the group. (A similar observation of multiplexing gain one per superposition-and-SIC user group was made in \cite{DingAdachiPoor16WC2}.)
Thus, the overall multiplexing gain of the mixture scheme with the adaptive user grouping is the same as the number of user groups $N_g$ which is less than or equal to $K(\le N)$. Note that in the case of $K=N$, the multiplexing gain of the ZF beamforming is $N$, whereas its diversity order is one.  Thus, diversity-order and multiplexing-gain trade-off known in single-user MIMO \cite{ZhengTse:03IT} occurs even in MISO BCs\cite{Mroueh08ISIT}.
In fact, it can be shown by replacing $L$ with $K$ in Proposition 1  and going through the proof of Theorem 1  with $N_g=1$ that the full diversity order  $N$ can be achieved by a single superposition-and-SIC group  containing all $K$ users without considering channel alignment and orthogonality at all. However, this single-group full superposition-and-SIC approach is not good since it  yields multiplexing gain one regardless of channel realization. This scheme can be considered as an antipodal scheme of the ZF beamforming in terms of diversity and multiplexing trade-off: The diversity order and multiplexing gain of the full superposition-and-SIC approach versus full ZF beamforming  are $(N,1)$ versus $(1,N)$ for MISO BCs with $K=N$.
On the other hand, the proposed user grouping method is adaptive and depends on the channel realization. The number of groups is not predetermined in the proposed user grouping method. The number of user groups can  be $K$ if all user channels are semi-orthogonal. The number of user groups can be one if all user channels are aligned.  Hence, the number $N_g$ of constructed user groups, i.e., the multiplexing gain of the mixture scheme with the adaptive user grouping method,  is adaptive to channel realization, while the full diversity order $N$ is always achieved.  So, we can view that the mixture scheme with such an adaptive user grouping method tries to opportunistically  increase the multiplexing gain while achieving the full diversity order.  Note that $N_g$ is a random variable under the assumption that $\hbf_k,k=1,\cdots,K$ are random, and it depends on the angle threshold between the user groups used in the adaptive user grouping algorithm. We were not able to compute an analytic form for the expectation of $N_g$ to evaluate the multiplexing gain loss as compared to the ZF beamforming, but a numerical assessment of the multiplexing gain loss as compared to the ZF beamforming is provided in Section V.

\section{Numerical Results}
\label{sec:NumericalResult}

In this section, we provide some numerical results to validate our theoretical analysis in the previous sections.  We considered the MISO BC described in Section \ref{subsec:ChannelModel}.
In each simulation scenario, we generated the $K$ channel vectors $\hbf_1, \cdots, \hbf_K$ of the system independently from the zero-mean complex Gaussian distribution $\mathcal{CN}(\mathbf{0}, 2 \Ibf)$
sufficiently many times to numerically compute outage probability.
For each channel realization, we ran the user grouping algorithm (Algorithm \ref{al:algoirhtm1}) with $\theta^{th} = 0.9$. With the constructed groups, we applied inter-group ZF beamforming and designed the intra-group beam vectors  according to the  constraint \eqref{eq:optimization_alter_cond}, i.e., $\wbf_1=\cdots=\wbf_L=\wbf^*$ with the solution  $\wbf^*$ to the max-min problem \eqref{eq:problem_SNR} used in the proof of Proposition \ref{pro:proposition2}.
The rate $R_k$ of the $k$-th user is obtained based on the designed beam vectors in this way. (Note that the beam vectors $\wbf_1=\cdots=\wbf_L=\wbf^*$ designed in this way yield rates larger than or equal to the lower bounds in
 \eqref{eq:propRate1} and \eqref{eq:propRatei}.)
For the intra-group beam design, the power distribution factors are chosen to satisfy the condition in Lemma \ref{lem:ExistDeltasForTheo1} in Appendix D.
The used values for  power distribution factors are  $(0.2,0.8)$ for every two-user group, $(0.05, 0.2, 0.75)$ for every three-user group in Figs. 2(a), 2(b) and 3(a), and are the solution of \eqref{eq:deltaSolution} with $C = 2$  in Fig. 3(b).  For computation of the outage probability  $\mathrm{Pr}( R_k \leq R_{th})$, we set the target rate threshold as $R_{th} = 1.5$ [bits/channel use] in all simulations.

\subsection{Diversity Order Considering Order Statistic}

\begin{figure}[t]
\centerline{ \SetLabels
\L(0.25*-0.1) (a) \\
\L(0.75*-0.1) (b) \\
\endSetLabels
\leavevmode
\strut\AffixLabels{
\scalefig{0.45}\epsfbox{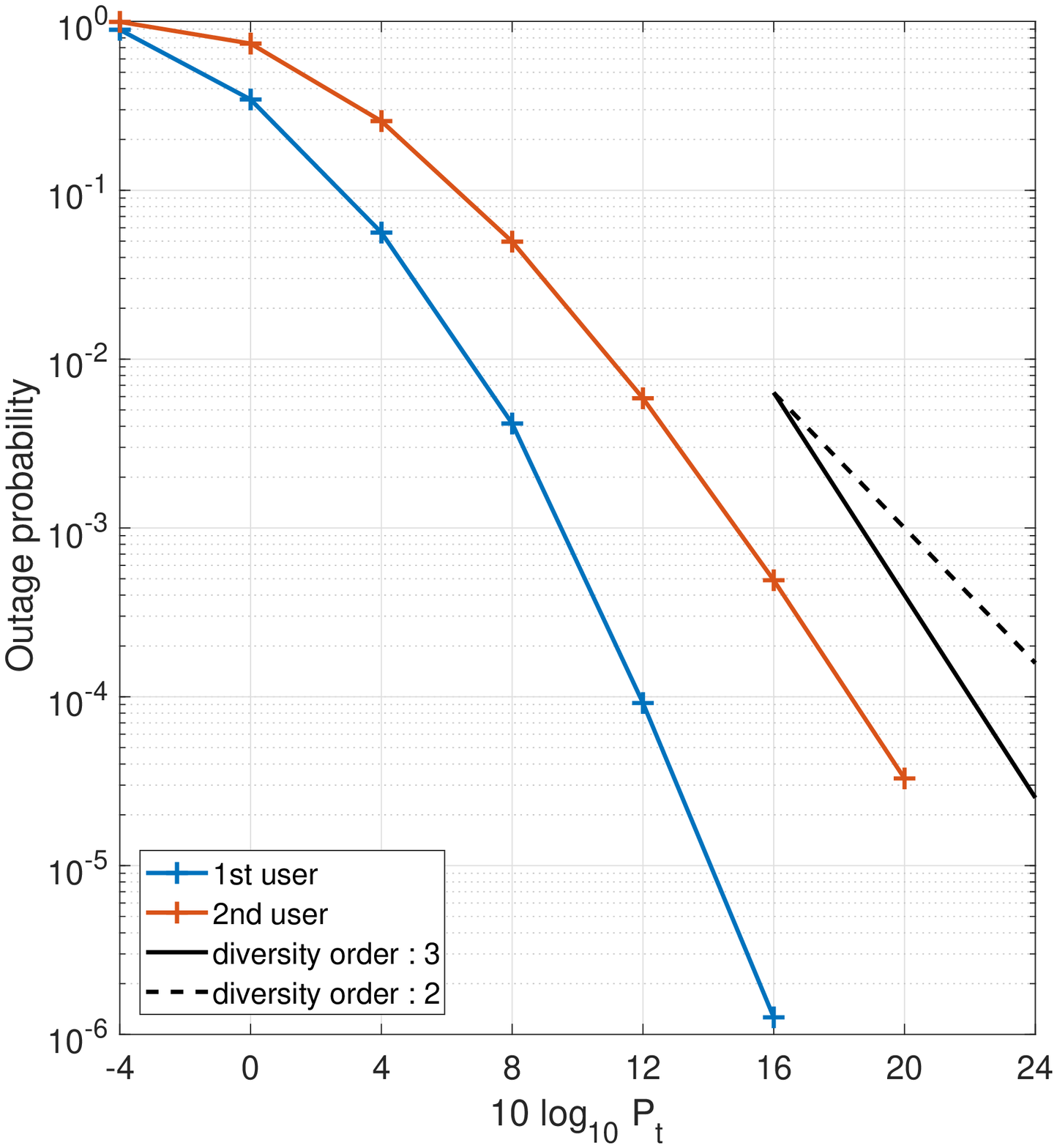}
\scalefig{0.45}\epsfbox{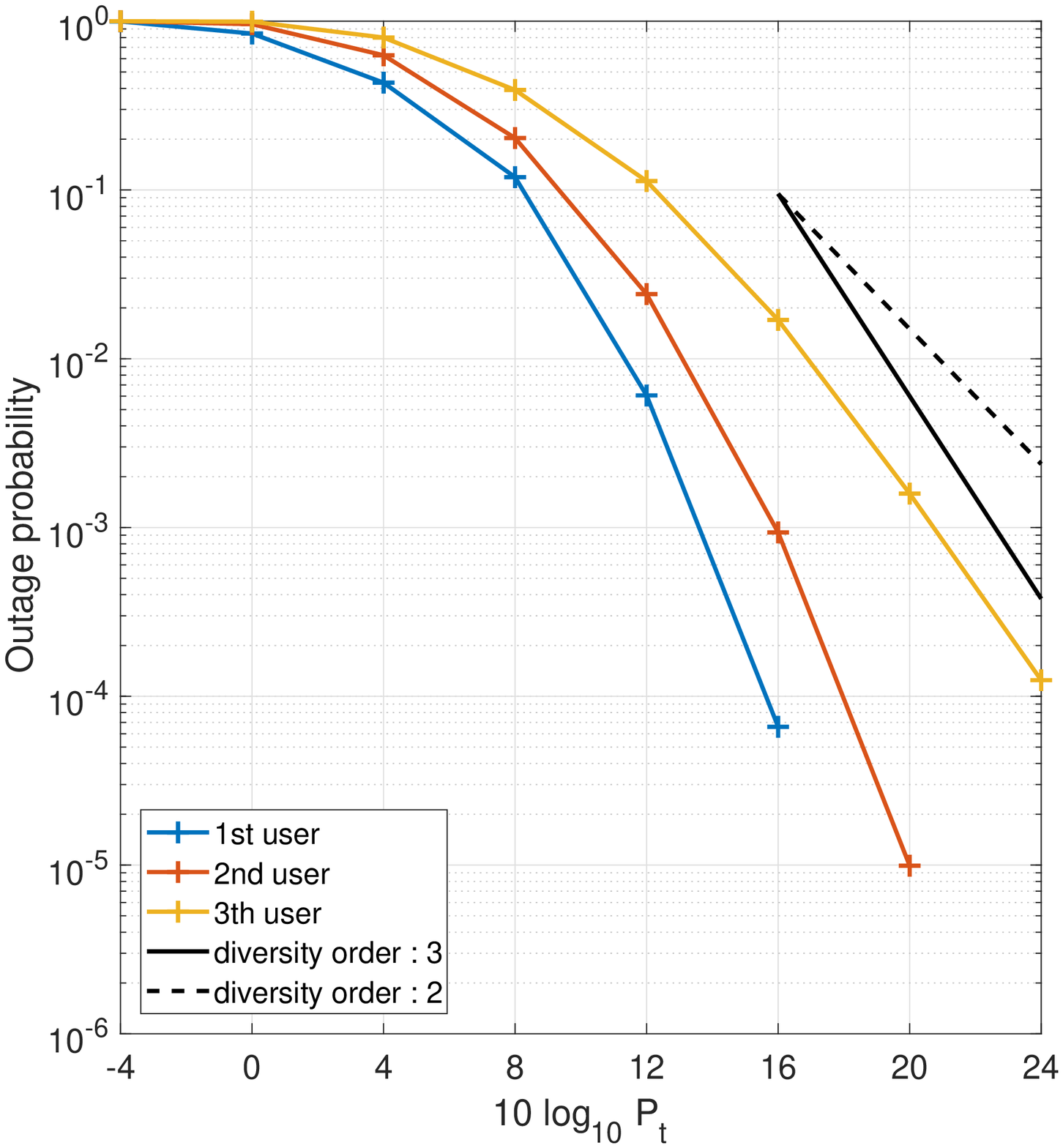}
} } \vspace{0.5cm} \caption{Outage probability  of the mixture transceiver architecture: (a) $N=3$, $K=2$
and (b) $N = 3$, $K = 3$} \label{fig:outage_probability}
\end{figure}

First, we numerically evaluated
the outage probability and diversity order of each user of the mixture transceiver architecture with considering channel norm ordering. Fig. \ref{fig:outage_probability} shows the outage probability of the mixture transceiver architecture in two cases: (a) $N=3$, $K=2$ and (b) $N=3$ and $K=3$, where
User $k$ is defined as the user with the $k$-th largest channel norm (i.e. $\|\hbf_1 \|^2 \geq \|\hbf_2 \|^2 \geq
\cdots \geq \|\hbf_K \|^2$).   In the case (a) of $N=3,K=2$, Theorem \ref{theo:MainResult} states that the diversity orders of Users 1 and 2 are 6 and 3, respectively.
It is seen in Fig.  \ref{fig:outage_probability}(a)  that the outage probability of User 2 has the slope corresponding to diversity order of 3, as SNR increases. It is also seen that the decay rate of User 1 is almost twice that of User 2. (In $\log_{10}$ y-scale, roughly User 1 has -4 and -5.9 and User 2 has -2.2 and -3.3 at $10\log P_t=$ 12 and 16, respectively.)
In the case (b) of $N=3,~K=3$, Theorem \ref{theo:MainResult} states that the diversity orders of Users 1, 2 and 3 are 9, 6, and 3, respectively. It is observed in Fig.  \ref{fig:outage_probability}(b)  that the outage probability of User 3 has the slope corresponding to diversity order of 3, as SNR increases.

\subsection{Overall Diversity Order}

\begin{figure}[t]
\centerline{ \SetLabels
\L(0.25*-0.1) (a) \\
\L(0.75*-0.1) (b) \\
\endSetLabels
\leavevmode
\strut\AffixLabels{
\scalefig{0.45}\epsfbox{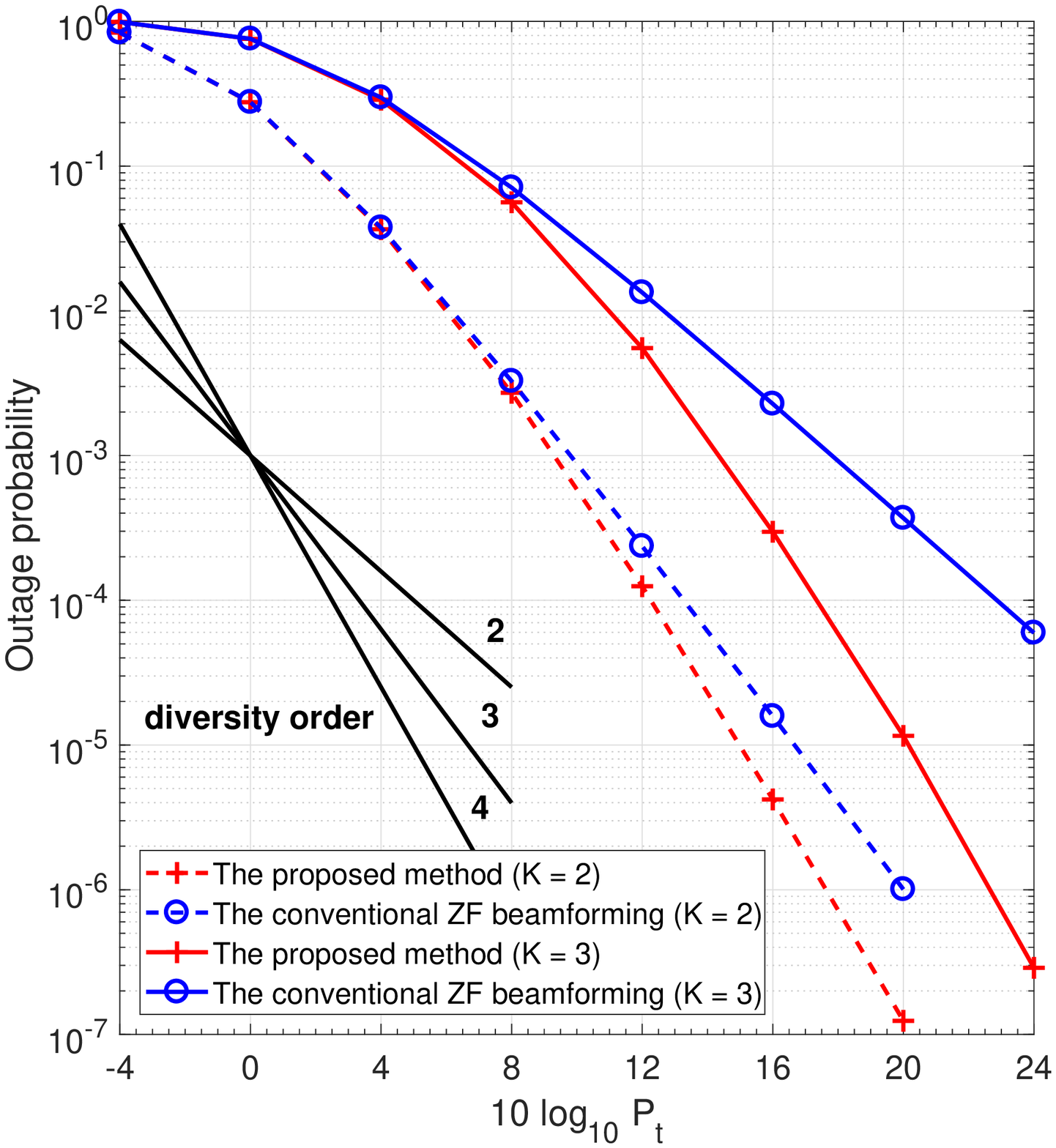}
\scalefig{0.45}\epsfbox{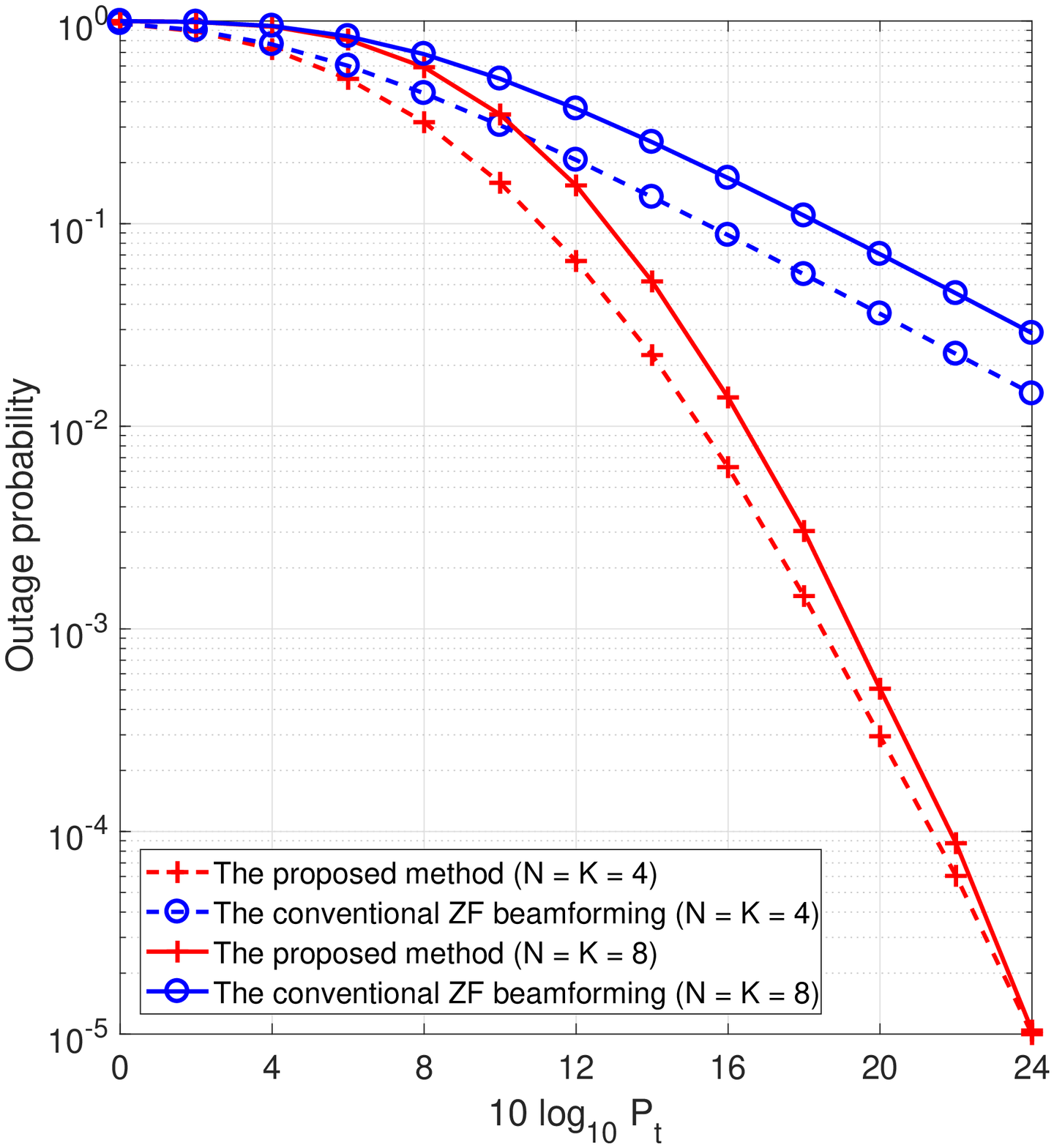}
} } \vspace{0.5cm} \caption{Overall outage probability : (a) $N = 4, ~K =2 ~\mbox{or}~ 3$ and (b) $N = K = 4$ and $N = K = 8$}    \label{fig:outage_probability_comparison}
\end{figure}

Then, we compared the mixture transceiver architecture with the full ZF downlink beamforming, based on the overall system diversity order. In order to see the overall diversity order, we computed overall outage probability. For this, we neglected channel norm ordering and computed the total number of outages occurred at all $K$ users  over all Monte Carlo runs.
 Fig. \ref{fig:outage_probability_comparison} shows the overall outage probability for the same channel statistics and the same rate threshold for the mixture scheme and the ZF downlink beamforming.
We considered four cases: {\em i)}  $N=4,K=2$ and {\em ii)} $N=4,K=3$ shown in  Fig. \ref{fig:outage_probability_comparison}(a) and  {\em iii)} $N=K=4$ and {\em iv)} $N=K=8$ shown in  Fig. \ref{fig:outage_probability_comparison}(b).
For the considered cases {\em i), ii), iii),} and {\em iv)}, the corresponding system diversity orders of the mixture scheme are 4, 4, 4 and 8  by Corollary  \ref{cor:systemDorder}, whereas the corresponding diversity orders of the ZF downlink beamforming are 3, 2, 1, and 1 by \eqref{eq:diversityOrderZF}.  It is seen in Fig. \ref{fig:outage_probability_comparison}(a)  that indeed the diversity orders of cases {\em i)} and {\em ii)} for the mixture scheme are the same as four. (The two red curves in Fig. \ref{fig:outage_probability_comparison}(a) seem to have the same slope with some offset, as SNR increases.)  On the other hand, it is seen that the diversity orders of the ZF downlink beamforming depends on $K$ for the same $N$, as expected.   The outage performance result for  the cases with more transmit antennas $N=K=4$ and $N=K=8$ is shown in Fig. \ref{fig:outage_probability_comparison}(b).  It is seen that the full ZF beamforming yields the same slope for the two cases $N=K=4$ and $N=K=8$, as expected, since it yields the diversity order of one in both cases by \eqref{eq:diversityOrderZF}.
On the other hand, it is seen that the diversity orders in the two cases $N=K=4$ and $N=K=8$ are different for the mixture scheme, as predicted by Corollary  \ref{cor:systemDorder}. Indeed, it is seen that the decay rate of the outage probability in the case of $N=K=8$ is larger than that of the case of $N=K=4$, although the outage probability of the case $N=K=8$ is higher than that of the case $N=K=4$ at low SNR. Note that the outage performance gain by the mixture scheme over the ZF beamforming is drastic in the case of  $N=K=4$ and $N=N=8$ for the meaningful range where the outage probability is below $10^{-2}$.

\subsection{Rate Distribution and Multiplexing Gain Loss}
\label{subsec:num_MG}

Next, we investigated the actual rate distribution and the multiplexing gain loss of the mixture scheme as compared to the ZF beamforming. For a numerical study, we again considered the case of $N=K=4$ considered in Fig. \ref{fig:outage_probability_comparison}(b).  For the power distribution factors $\delta_1,\cdots,\delta_4$, we used \eqref{eq:deltaSolution} with $R_{th}=1.5$ and $C=2$. (Other simulation setting is the same as that for Fig. \ref{fig:outage_probability_comparison}(b).)  We know that the common target rate should be smaller than
$\bar{R}_i^{(j)} = \log_2 \left( 1+ \frac{\delta_i^{(j)}}{\sum_{m=1}^{i-1}\delta_m^{(j)}} \right)$, $i=2,\cdots,|\Gc_j|$ since the rates of the users except the first user in group $\Gc_j$ saturate to $\bar{R}_i^{(j)}$, $i=2,\cdots,|\Gc_j|$. However, the power distribution factors $\delta_1^{(j)},\cdots,\delta_{|\Gc_j|}^{(j)}$ of group $\Gc_j$ can be designed for the target rate $R_{th}$  by using \eqref{eq:deltaSolution}. Note that channel realization does not satisfy the target rate $R_{th}$  with 100 percents and it is just a target rate. Hence, outage can still occur for the designed target rate with small probability.
\eqref{eq:deltaSolution} with $R_{th}=1.5$ and $C=2$ yields the following power distribution factor values:

\begin{itemize}
\item $\delta_1^{(j)}=1$ for the first user in any group with cardinality one.

\item  $\delta_1^{(j)}= 0.2071$, $\delta_2^{(j)}=0.7929$  for the first and second users in any group with cardinality two.

\item  $\delta_1^{(j)}=0.0429$, $\delta_2^{(j)}=0.1642$, $\delta_3^{(j)}=0.7929$  for the first, second and third users in any group with cardinality three.

\item  $\delta_1^{(j)}=0.0089$, $\delta_2^{(j)}= 0.0340$, $\delta_3^{(j)}= 0.1642$, $\delta_4^{(j)}=0.7929$ for the first,  second, third and fourth users in any group with cardinality four.

\end{itemize}

\noindent The corresponding $\bar{R}_i^{(j)}$ is given as follows:

\begin{itemize}

\item  $\bar{R}_2^{(j)}=2.2716$  for the second user in any group with cardinality two.

\item  $\bar{R}_2^{(j)}=2.2713$, $\bar{R}_3^{(j)}=2.2716$  for the second and third users in any group with cardinality three.

\item  $\bar{R}_2^{(j)}=2.2691$, $\bar{R}_3^{(j)}=2.2713$, $\bar{R}_3^{(j)}=2.2716$ for the second, third and fourth users in any group with cardinality four.

\end{itemize}
Note that \eqref{eq:deltaSolution} with $R_{th}=1.5$ and $C=2$ yields the power distribution factor values so that the rate upper bound $\bar{R}_i^{(j)}$ for non-first users is set just above the target rate $R_{th}$. The margin is controlled by the constant $C$. Hence, when a common target rate is given, we can design the power distribution factors $\delta_i^{(j)}$ such that the  rate upper bound $\bar{R}_i^{(j)}$ for non-first users is set just above the target rate $R_{th}$ by using \eqref{eq:deltaSolution}.

\begin{figure}[htbp]
\centerline{ \SetLabels
\L(0.5*-0.075) (a) \\
\endSetLabels
\leavevmode
\strut\AffixLabels{
\scalefig{0.55}\epsfbox{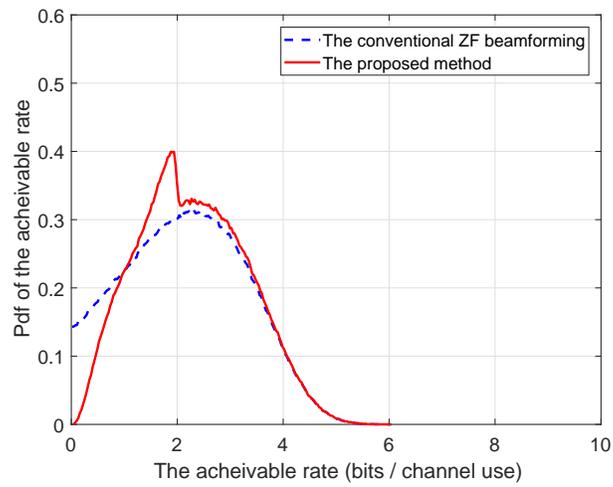}
} } \vspace{0.7cm}
 \centerline{ \SetLabels
\L(0.5*-0.075) (b) \\
\endSetLabels
\leavevmode
\strut\AffixLabels{
\scalefig{0.55}\epsfbox{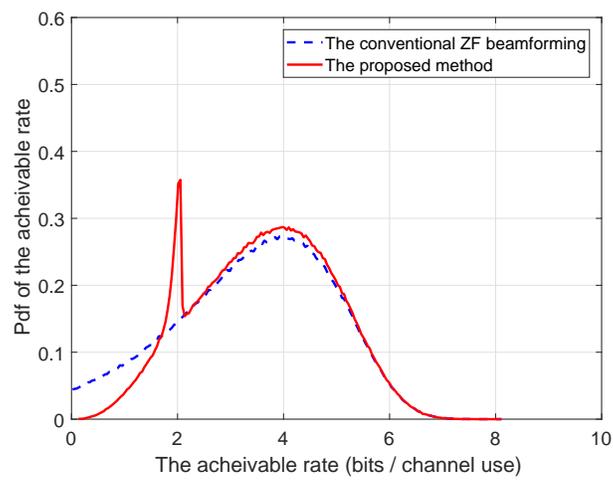}
} } \vspace{0.7cm}\centerline{ \SetLabels
\L(0.5*-0.075) (c) \\
\endSetLabels
\leavevmode
\strut\AffixLabels{
\scalefig{0.55}\epsfbox{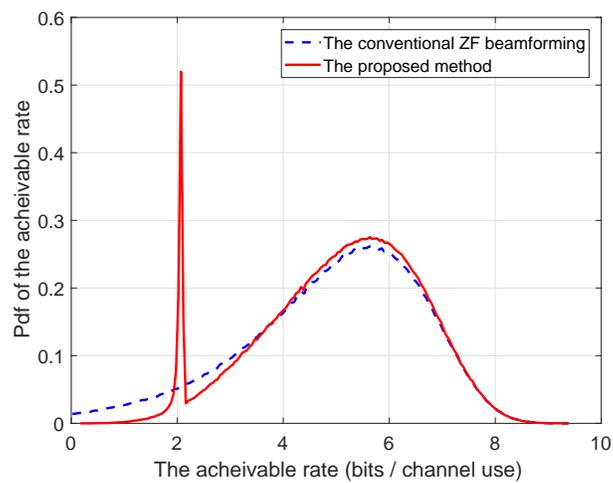}
} }
 \vspace{0.4cm} \caption{Rate distribution ($N=K=4$) (a) $10\log\frac{P_t}{1}=10$dB, (b) $10\log\frac{P_t}{1}=15$dB, (c) $10\log\frac{P_t}{1}=20$dB}
\label{fig:rateDistributionA}
\end{figure}

\begin{figure}[h]
\centerline{ \SetLabels
\L(0.25*-0.1) (a) \\
\L(0.75*-0.1) (b) \\
\endSetLabels
\leavevmode
\strut\AffixLabels{
\scalefig{0.45}\epsfbox{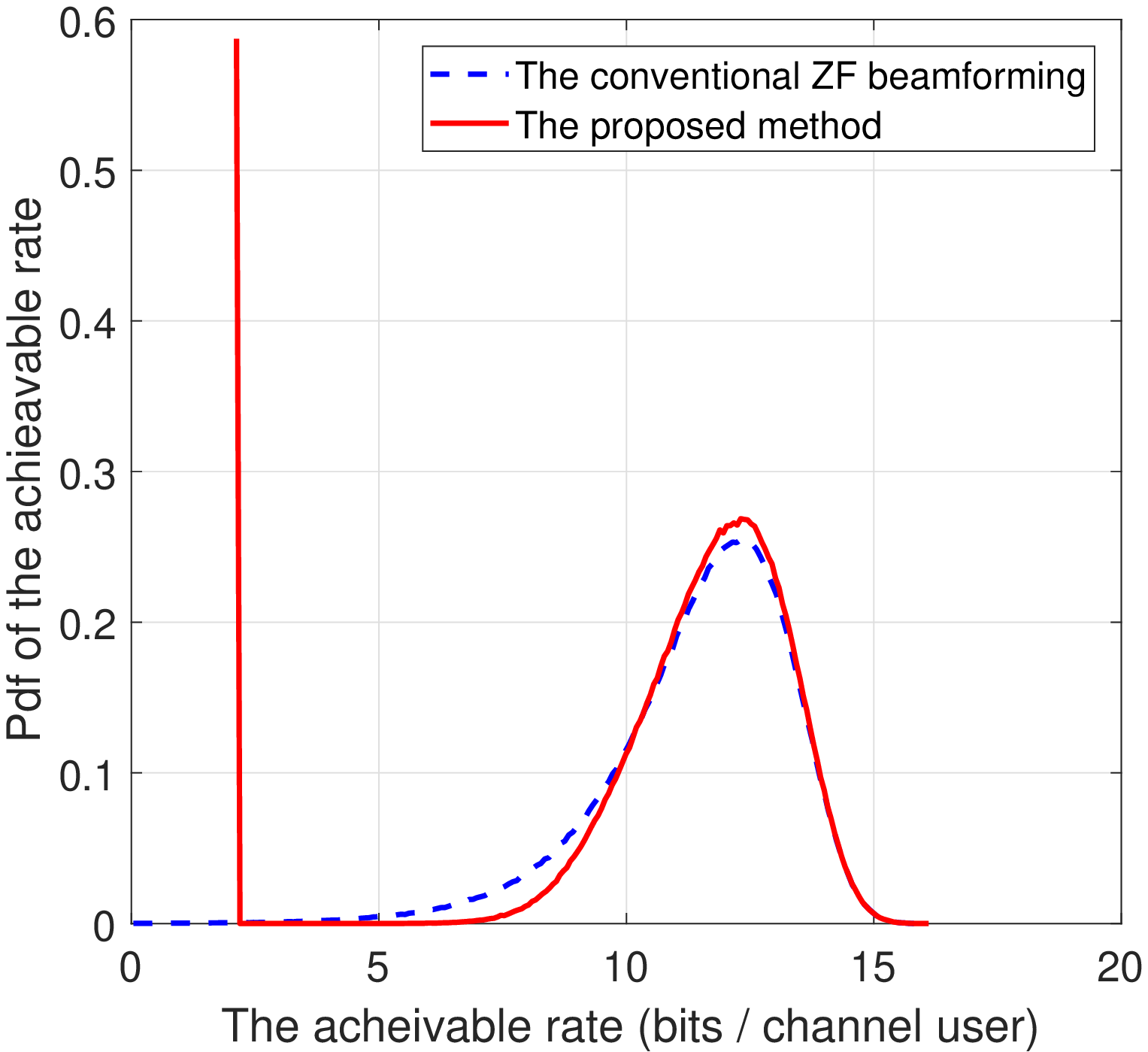 }
\scalefig{0.45}\epsfbox{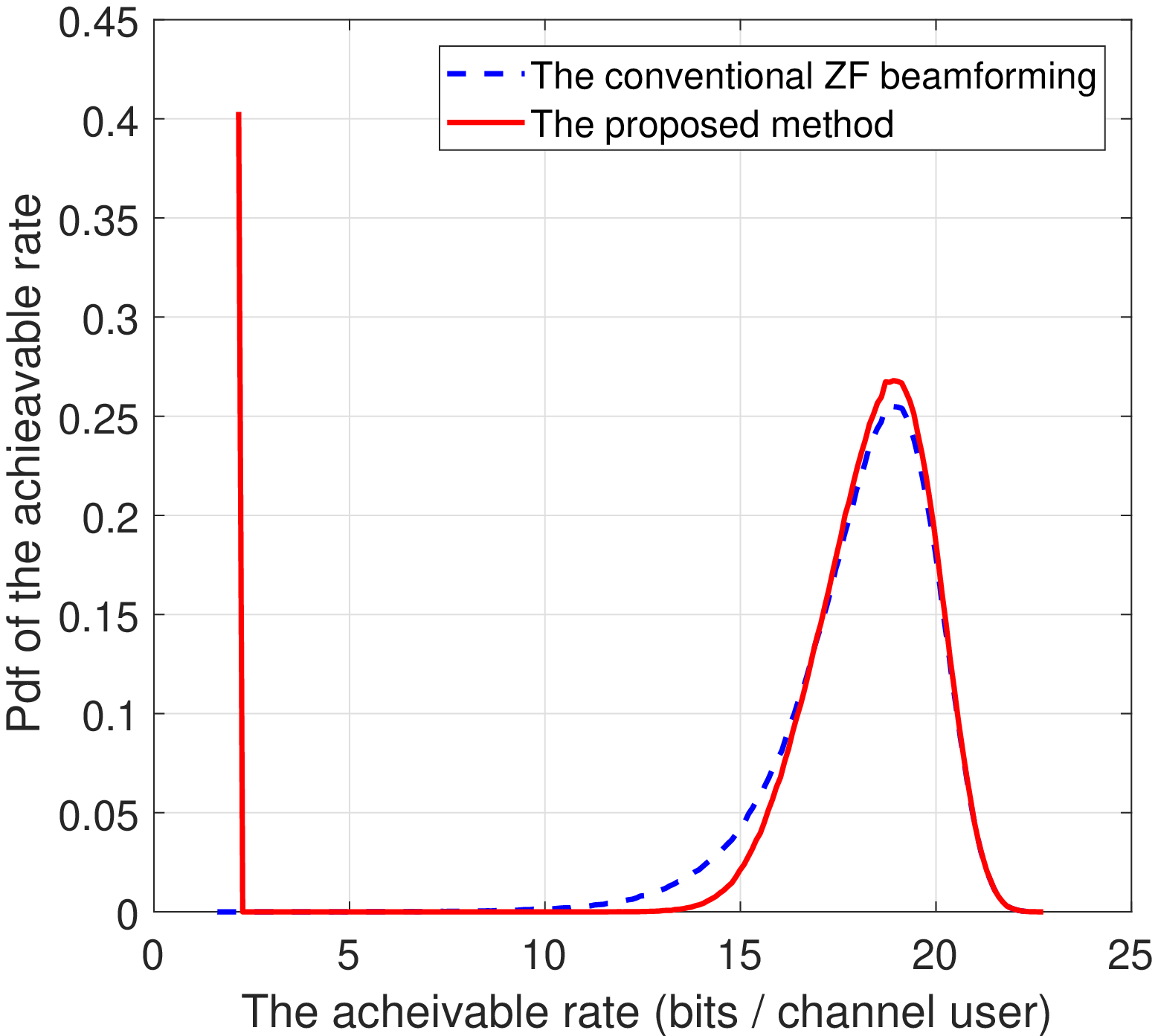 }
} } \vspace{0.5cm} \caption{Rate distribution ($N=K=4$) (a) $10\log\frac{P_t}{1}=40$dB and (b) $10\log\frac{P_t}{1}=60$dB}    \label{fig:rateDistributionB}
\end{figure}

For the $N=K=4$ system, we considered $10\log_{10}\frac{P_t}{1}=[10,15,20,40,60]$ dB, where one in the denominator is the noise variance. For each SNR point, we generated 500,000 channel realizations. For each channel realization, we applied the ZF beamforming and the mixture scheme and obtained the rates of the four users in the system. With the overall 4$\times$500,000 rate values, we obtained the rate distribution with the histogram method. The rate distribution results are shown in Figs. \ref{fig:rateDistributionA} and \ref{fig:rateDistributionB}.  Note that the mixture scheme with adaptive user grouping is opportunistic in  multiplexing gain and at least multiplexing gain of one is guaranteed since the number of groups is equal to or larger than one. It is observed that the rate distribution of the mixture scheme is a mixture of the first users' rate distribution and the non-first users' rate distribution. The distribution component of the first users' rates shows a similar distribution to that of the ZF scheme. That is, as SNR increases, the first users' rate distribution shifts to the right in the figures. We also see the component of the distribution of the non-first  users's rates. This component accumulates around $R\approx 2.2$ as predicted by the above values of $\bar{R}_i^{(j)}$. As SNR increases, the accumulation becomes sharper looking like a peak just below $\bar{R}_i^{(j)}$. Note the rate lower tail behaviors of the mixture scheme and the ZF scheme. At SNR = 10, 15, 20 dB, the mixture has much lighter tails. At SNR=20 dB, the mixture scheme yields most rates above the target rate $R_{th}=1.5$, whereas still quite a portion is below the target rate $R_{th}=1.5$ with the ZF scheme. Even at SNR=40dB, we can still see the non-zero tail around the origin for the ZF scheme, whereas for the mixture scheme the rate distribution starts from $R_{th}$ with a sharp peak. Note that the first users' rates of the mixture scheme almost match those of the ZF scheme. However, still there is a large peak around $\bar{R}_{i}^{(j)}$ due to the non-first users for the mixture scheme, and this reduces the multiplexing gain of the mixture scheme.

\begin{figure}[ht]
\begin{psfrags}
    \centerline{ \scalefig{0.6} \epsfbox{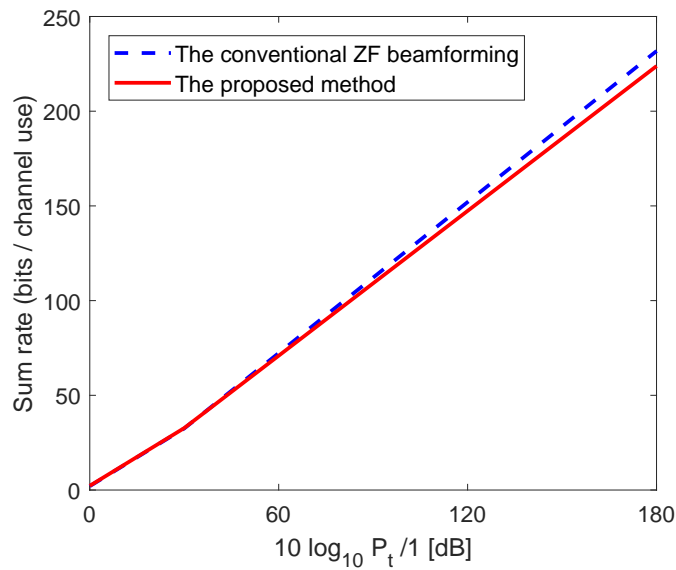} }
    \caption{Average sum rate: $K=N=4$}
    \label{fig:sum_rate_multiplexingGain}
\end{psfrags}
\end{figure}

Hence, we further investigated the multiplexing gain, i.e., the slope of rate increase with respect to SNR. For each SNR point, we averaged the rates of the four users in the system over channel realizations. Then, we plotted the average rates of the mixture scheme and the ZF scheme with respect to SNR. The result is shown in Fig. \ref{fig:sum_rate_multiplexingGain}. It is seen that  the multiplexing gain loss of the mixture scheme compared to the ZF scheme is insignificant at least in the case of $N=K=4$.

\begin{figure}[ht]
\begin{psfrags}
    \centerline{ \scalefig{0.6} \epsfbox{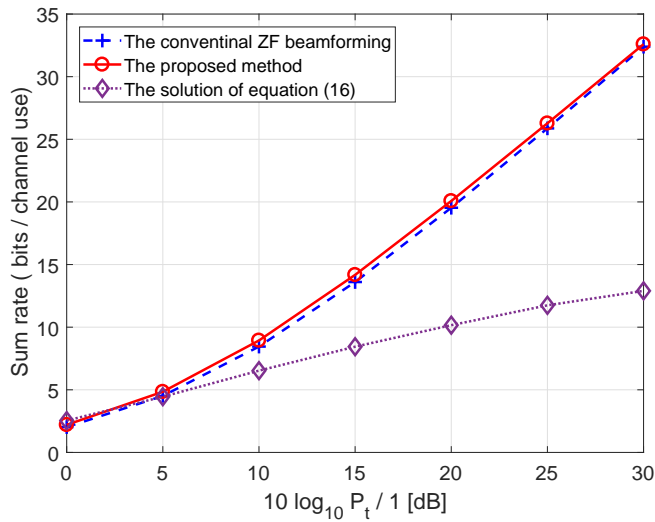} }
    \caption{{Average sum rate: $K=N=4$}}
    \label{fig:sum_rate_multiplexingGainProblem16}
\end{psfrags}
\end{figure}

{We further investigated the performance of the single-group approach with more sophisticated beam design  obtained by solving the problem \eqref{eq:AchieRegionMISOBCSIC}.  We solved the problem \eqref{eq:AchieRegionMISOBCSIC} for the $N=K=4$ system considered above as a single group with superposition and SIC.
Since the problem \eqref{eq:AchieRegionMISOBCSIC} is non-convex, several steps are needed.
First, we transform the problem \eqref{eq:AchieRegionMISOBCSIC} into a problem of maximizing
sum rate  with feasible power ratio-tuples.
Then, it is reformulated as maximizing the geometric mean
of SINRs with non-convex constraints\cite{Hanif16:TSP}.
Next, we approximate the non-convex constraints
using the convex concave procedure\cite{yuille03:Concave} and can solve
the problem in an iterative manner.
Sweeping $(p_1,p_2,\cdots,p_K)=(\delta_1P_t, \delta_2 P_t, \cdots, \delta_K P_t)$ yields a rate region. However, we did not perform this sweeping since our goal is not to obtain a rate region. Instead, we determined $(p_1,p_2,\cdots,p_K)=(\delta_1P_t, \delta_2 P_t, \cdots, \delta_K P_t)$ based on \eqref{eq:deltaSolution} with $C = 2$
and $R_{th} = 1.5$ and computed the corresponding sum rate of the problem \eqref{eq:AchieRegionMISOBCSIC}.
The corresponding rate-tuple point is on the boundary of the rate region of (16), although it may not be the sum-rate maximizing point.  The result is shown in Fig. \ref{fig:sum_rate_multiplexingGainProblem16}.  The curves of the proposed method and the conventional ZF method are the same as those in  Fig. \ref{fig:sum_rate_multiplexingGain}, and the curve of the solution of the problem (16) with $(p_1,p_2,\cdots,p_K)=(\delta_1P_t, \delta_2 P_t, \cdots, \delta_K P_t)$ determined based on \eqref{eq:deltaSolution} with $C = 2$
and $R_{th} = 1.5$ is newly added. Even though we solve the problem (16) optimally based on the aforementioned complicated procedure not based on user grouping, inter-group ZF, in-group simple superposition beamforming $\wbf_1=\cdots=\wbf_K$, the resulting rate of the problem (16) with a single-group approach is not good. Note that the corresponding slope is much smaller than that of  the ZF scheme and the proposed scheme. This is because as mentioned before, if we group all users in a single group and apply superposition and SIC, we have the multiplexing gain of only one, whatever sophisticated beam design and power allocation are used. Even if we adjust power allocation to yield maximum sum rate, this does not change the slope, i.e., the multiplexing gain.  On the other hand, the full ZF beamforming has the multiplexing gain of four and the mixture scheme with adaptive user grouping has the multiplexing gain from one to four. On average, the multiplexing gain of the mixture scheme with adaptive user grouping  slightly falls short of four, as seen in Fig.   \ref{fig:sum_rate_multiplexingGain}.  So, it is more important to group users properly to yield  as many  groups as possible, while maintaining minimum inter-group angle separation,  rather than to apply a sophisticated beam design method with one overall group from the perspective of the multiplexing gain, i.e., the sum rate.}

Now how to operate the mixture scheme is clear. Consider MISO-BC URLLC in which no retransmission is allowed due to latency constraint (one round-trip delay for retransmission is in the order of 10 ms, whereas URLLC requires 1ms delay) and low-latency low-data-rate packets should be delivered reliably.
First, we determine the angle threshold between group channel subspaces to be not too large so that we have as many groups as possible but we still avoid angle-wise very close groups. We determine the minimum target rate that should be satisfied by all users for URLLC. With the target rate, we design the power distribution factors, and run the adaptive user grouping. For the first users in the constructed groups, we can still apply rate adaptation based on modulation level and coding rate by exploiting the supportable rate channel quality indicator (CQI). (The distribution of the first users' rates is wide across the x-axis in Fig. \ref{fig:rateDistributionA} and \ref{fig:rateDistributionB}. We should exploit this.) But, for the non-first users we just transmit data with the target rate. In fact, we can control the first user in each group. In the case that a user wanted as the first user does not have maximum effective channel norm, we assign more power to the wanted user so that more power times its effective channel gain surpasses the largest effective channel norm of other user in the group. Then, we distribute the remaining group power according to \eqref{eq:deltaSolution} with the target rate $R_{th}$. With this, we can control the mixture system so that any user can be a high-rate first user while supporting the target rate with high reliability.

\subsection{Comparison with Other Advanced Transceiver Designs for MISO BCs}

We considered other advanced transceiver designs for MISO BCs, e.g., \cite{HuRusek:17WC,PeelHochwaldSwindlehurst:05COM}, devised to improve the performance over ZF downlink beamforming, and compared the outage performance of these advanced designs with the mixture architecture. The result is shown in Fig. \ref{fig:comwithAdvancedTRX}, where the setup is $N=K=4$ and other parameter setting is the same as that in Fig.  \ref{fig:outage_probability_comparison}(b) with $N=K=4$. It is seen that the advanced transceiver designs having the full multiplexing gain yield  the same diversity order as the ZF beamforming, which is worse than that of our scheme,  although they yield better rates compared to the ZF beamforming. Thus, these advanced designs are at the multiplexing-gain-optimal side in terms of diversity-and-multiplexing trade-off.

\begin{figure}[h]
\begin{psfrags}
    \psfrag{o}[l]{\small origin} %
    \psfrag{cth}[c]{\small $\theta_{\tau,1}$} %
    \psfrag{h}[l]{\small $\hbf_{g_{\hat{k}}}$} %
    \centerline{ \scalefig{0.5} \epsfbox{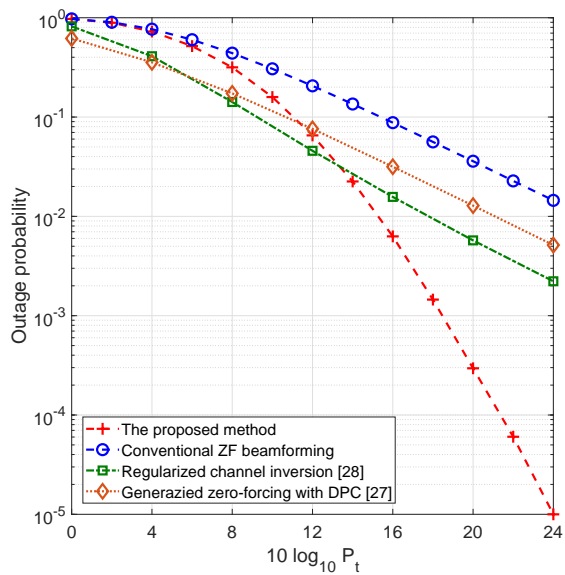} }
    \caption{Comparison with other advanced methods: $K=N=4$}
    \label{fig:comwithAdvancedTRX}
\end{psfrags}
\end{figure}

\subsection{Impact of Imperfect CSI}

\begin{figure}[h]
\begin{psfrags}
    \centerline{ \scalefig{0.5} \epsfbox{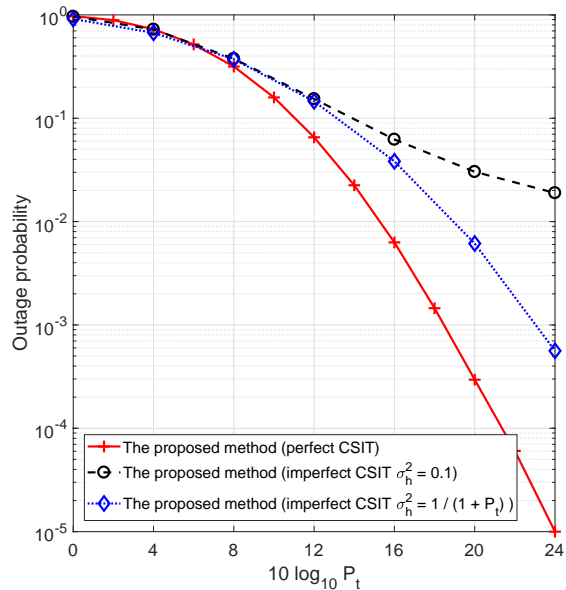} }
    \caption{Impact of imperfect CSI: $K=N=4$}
     \label{fig:imperfectCSI}
\end{psfrags}
\end{figure}

Although analysis  of the outage performance under imperfect CSIT is beyond the scope of this paper, we briefly investigated the impact of imperfect CSIT through simulation. Again we considered the case of $N=K=4$ with the same other setting as that in Fig. \ref{fig:outage_probability_comparison}(b).  It is known that the number of CSI feedback bits per user should increase linearly with respect to SNR  (or signal power for fixed noise variance) in log scale
in order to achieve full multiplexing gain for MISO BCs\cite{Jindal:06IT}.   For simulation the CSI error is assumed to be zero-mean Gaussian with variance $\sigma_{e}^2$, where we set 1) $\sigma_e^2$ as a fixed constant of 0.1 and 2) $\sigma_{e}^2 = \frac{1}{1+P_t}$ to be consistent with the result in \cite{Jindal:06IT}. The result is shown in Fig. \ref{fig:imperfectCSI}. It is seen that the fixed CSI quality with respect to SNR shows a floor for the outage probability as SNR increases. On the other hand, the CSI with quality $\sigma_e^2 = \frac{1}{1+P_t}$ does not show such a floor behavior. Indeed, it seems that the increasing CSI quality with respect to SNR is required to achieve the full diversity order although the exact increasing rate is not known yet.

\section{Conclusion}
\label{sec:conclusion}

In this paper,  we have considered the  mixture transceiver architecture   with channel-adaptive user grouping and mixture of linear and nonlinear SIC reception for MISO BCs, and {}{have shown that the mixture transceiver architecture opportunistically increases the multiplexing gain while achieving full diversity order for  MISO BCs.   The mixture transceiver architecture can provide far better outage  performance compared to the widely-used conventional ZF downlink
beamforming for MU-MISO BCs under channel fading environments. The gain in diversity order results from possible sacrifice of multiplexing gain through diversity-and-multiplexing trade-off, and thus the mixture scheme  provides an alternative transceiver architecture for MISO BCs to applications such as emerging URLLC in which reliability is more important than data rate.
Future research directions include optimization of angle threshold and power distribution, finding optimal diversity-and-multiplexing trade-off in MISO BCs,  finding faster grouping algorithms scalable with the number of users for large systems, application of the mixture architecture to the uplink\cite{Hamid08COM}, and application of more advanced transmit signaling\cite{ZengYetis:13SP}}.

\section*{Appendix A: Proof of Proposition \ref{pro:proposition2}}

    For given $(\delta_1, \cdots, \delta_L)$, in order to  obtain a lower bound on the achievable rate of each user, we simply set
    $\wbf_1 = \wbf_2 = \cdots = \wbf_L = \wbf$ with $||\wbf||^2 \le 1$ as in the constraint \eqref{eq:optimization_alter_cond}, i.e., we consider that all $L$ users use the same beam vector. Then,
     the rates in
     \eqref{eq:rateLMISOBCSIC1}  of the  MISO BC with superposition coding and SIC
     can be rewritten as
\begin{align}
    R_1 &=  \log_2 \left( 1 + \delta_1 P|\mathbf{g}_1^H \mathbf{w}|^2\right) \label{eq:rate_singlebeam111} \\
    R_i &=  \log_2 \left( 1 +  \min \left\{ \frac{ \delta_i P |\gbf_1^H \wbf|^2}{\sum_{m=1}^{i-1} \delta_m P|\gbf_1^H \wbf|^2 + 1}
    , \cdots ,\frac{ \delta_i P |\gbf_i^H \wbf|^2}{\sum_{m=1}^{i-1} \delta_m P|\gbf_i^H \wbf|^2 + 1}\right\}\right), \quad i = 2, \cdots, L, \nonumber \\
    &=  \log_2 \left(1 + \frac{\delta_i}{\sum_{m=1}^{i-1} \delta_m} \cdot \frac{1}{1+ 1/\left[\min{\{ |\gbf_1^H\wbf|^2,\cdots, |\gbf_i^H \wbf|^2 \}} (\sum_{m=1}^{i-1} \delta_m) P\right]} \right). \label{eq:rate_singlebeam}
\end{align}
Using Lemma \ref{Lemma:lemma1} below, we can bound the terms $|\mathbf{g}_1^H \mathbf{w}|^2$ in \eqref{eq:rate_singlebeam111} and $\min{\{ |\gbf_1^H\wbf|^2,\cdots, |\gbf_i^H \wbf|^2 \}}$ in  \eqref{eq:rate_singlebeam} as follows:
Using the optimal solution $\wbf^*$ to the max-min problem \eqref{eq:problem_SNR}, we have
\begin{align}
    |\gbf_1^H \wbf^*|^2 &\geq \min\left\{\left|\left(\frac{\gbf_1}{\|\gbf_1\|}\right)^H \wbf^*\right|^2, \cdots,  \left|\left(\frac{\gbf_L}{\|\gbf_L\|}\right)^H \wbf^*\right|^2\right\}  \|\gbf_1\|^2   \label{eq:prop1_inequal11} \\
    &\geq\frac{1}{c} \|\gbf_1\|^2,  \label{eq:prop1_inequal2}
\end{align}
where  \eqref{eq:prop1_inequal11} is valid since the minimum is taken over multiple terms including $|\gbf_1^H \wbf^*|^2$, and \eqref{eq:prop1_inequal2} is valid by Lemma \ref{Lemma:lemma1} below. Next, we have
\begin{align}
    \min\{|\gbf_1^H \wbf^*|^2, \cdots, |\gbf_i^H \wbf^*|^2\} &= \min\left\{\left|\left(\frac{\gbf_1}{\|\gbf_i\|}\right)^H \wbf^*\right|^2,\cdots, \left|\left(\frac{\gbf_i}{\|\gbf_i\|}\right)^H \wbf^*\right|^2\right\}  \|\gbf_i\|^2 \\
    &\geq  \min\left\{\left|\left(\frac{\gbf_1}{\|\gbf_1\|}\right)^H \wbf^*\right|^2, \cdots, \left|\left(\frac{\gbf_i}{\|\gbf_i\|}\right)^H \wbf^*\right|^2\right\}  \|\gbf_i\|^2 \label{eq_prop1_inequal21} \\
    &\geq  \min\left\{\left|\left(\frac{\gbf_1}{\|\gbf_1\|}\right)^H \wbf^*\right|^2, \cdots, \left|\left(\frac{\gbf_L}{\|\gbf_L\|}\right)^H \wbf^*\right|^2\right\}  \|\gbf_i\|^2 \label{eq_prop1_inequal22}\\
    &\geq\frac{1}{c} \|\gbf_i\|^2, \label{eq_prop1_inequal23}
\end{align}
where \eqref{eq_prop1_inequal21} is valid since $||\gbf_1||\ge \cdots \ge ||\gbf_L||$, \eqref{eq_prop1_inequal22} is valid since we increased the number of terms in the minimization including the previous terms, and \eqref{eq_prop1_inequal23} holds by Lemma \ref{Lemma:lemma1} below.
Substituting
\eqref{eq:prop1_inequal2} and   \eqref{eq_prop1_inequal23} into
\eqref{eq:rate_singlebeam111} and \eqref{eq:rate_singlebeam}, respectively, we have the rates  that can be achieved by the optimal solution  $\wbf^*=\wbf_1=\cdots=\wbf_L$ to the max-min problem \eqref{eq:problem_SNR}:
\begin{align}
    R_1 &\geq \log_2 \left( 1 + \frac{1}{c} \delta_1 \|\gbf_1\|^2 P \right) \\
    R_i &\geq \log_2 \left( 1 + \frac{\delta_i}{\sum_{m=1}^{i-1} \delta_m} \frac{1}{1 + \left(\frac{1}{c} \|\gbf_i\|^2\sum_{m=1}^{i-1} \delta_m P \right)^{-1}}\right),
    \quad i = 2, \cdots, L.
\end{align}
The considered design here of  $\wbf_1=\cdots=\wbf_L=\wbf^*$ with $||\wbf^*||^2 \le 1$ and $p_i = \delta_i P$ with $(\delta_1,\cdots,\delta_L) \in \Dc$, i.e., \eqref{eq:optimization_alter_cond}, satisfies the original beam design constraint $(\wbf_1, \cdots, \wbf_L) \in \mathcal{W}^L$  and $p_i >0, \forall i,~\sum_{i=1}^L p_i = P$ in  \eqref{eq:AchieRegionMISOBCSIC}.   Hence, the rates achieved by  $\wbf_1=\cdots=\wbf_L=\wbf^*$ with $(\delta_1,\cdots,\delta_L)$ are lower bounds on the achievable rates. \hfill{$\blacksquare$}

\begin{lemma}  \label{Lemma:lemma1}
Consider the following max-min optimization problem:
\begin{equation} \label{eq:problem_SNR}
    \begin{array}{ll}
      \max & \min{\left\{ \left| \left(\frac{\gbf_1}{\|\gbf_1\|}\right)^H \wbf\right|^2, \cdots, \left| \left(\frac{\gbf_L}{\|\gbf_L\|}\right)^H \wbf\right|^2\right\}} \\
      \mathrm{subject ~ to} & \|\wbf\|^2 \leq 1.
    \end{array}
\end{equation}
The optimal solution $\wbf^*$ to
the problem \eqref{eq:problem_SNR} satisfies the following:
\begin{equation}
        \min{\left\{ \left| \left(\frac{\gbf_1}{\|\gbf_1\|}\right)^H \wbf^*\right|^2, \cdots, \left| \left(\frac{\gbf_L}{\|\gbf_L\|}\right)^H \wbf^*\right|^2\right\}} \geq \frac{1}{c},
\end{equation}
where
\begin{equation}  \label{eq:Lemma:lemma1}
    c = \left\{ \begin{array}{ll}
                  L & \mbox{if} \quad L \leq 3, \\
                  8L^2 & \mbox{if} \quad L > 3.
                \end{array}
    \right.
\end{equation}
\end{lemma}

{\it Proof of Lemma  \ref{Lemma:lemma1}:} ~Define unit-norm $\vbf_i := \gbf_i / \|\gbf_i\|$ for $ i = 1,\cdots, L$. Then, \eqref{eq:problem_SNR} can be rewritten as
\begin{equation} \label{eq:problem_SNR_2}
    \begin{array}{ll}
      \max & \min{\left\{ \left| \vbf_1^H \wbf\right|^2, \cdots, \left| \vbf_L^H \wbf\right|^2\right\}} \\
      \mathrm{subject ~ to} & \|\wbf\|^2 \leq 1
    \end{array}
\end{equation}
The problem \eqref{eq:problem_SNR_2} can be reformulated as
\begin{equation}  \label{eq:Lemma1inverse}
\max ~ \frac{\min{\left\{ \left| \vbf_1^H \wbf\right|^2, \cdots, \left| \vbf_L^H \wbf\right|^2\right\}}}{||\wbf||^2} =  \min ~ \frac{||\wbf||^2}{\min{\left\{ \left| \vbf_1^H \wbf\right|^2, \cdots, \left| \vbf_L^H \wbf\right|^2\right\}}},
\end{equation}
where inversion of the cost function is taken in the right-hand side (RHS) of \eqref{eq:Lemma1inverse} . Thus,
it is known that the optimal value of the problem \eqref{eq:problem_SNR_2} is equivalent to the inverse of the optimal value of the following quadratic programming (QP) \cite{Sidiropoulos:06TSP}:
\begin{equation} \label{eq:problem_SNR_3}
    \begin{array}{ll}
      \min & \|\wbf\|^2 \\
      \mathrm{subject ~ to} & |\vbf_i^H \wbf|^2 \geq 1, \quad i = 1,\cdots, L.
    \end{array}
\end{equation}
The QP \eqref{eq:problem_SNR_3} can be solved by semi-definite relaxation of the rewritten form of \eqref{eq:problem_SNR_3}  \cite{Sidiropoulos:06TSP}:
\begin{equation}   \label{eq:problem_SNR_4}
    \begin{array}{ll}
      \min & \mathrm{Tr}(\Wbf) \\
      \mathrm{subject ~ to} & \mathrm{Tr}(\Vbf_i \Wbf) \geq 1, \quad i = 1,\cdots, L
    \end{array}
\end{equation}
where $\Wbf := \wbf \wbf^H$ and $\Vbf_i := \vbf_i\vbf_i^H$, $i = 1\cdots, L$.
Denote the optimal values of the optimization problems \eqref{eq:problem_SNR_3} and \eqref{eq:problem_SNR_4} by
$v_{qp}^*$ and $v_{sdp}^*$, respectively. Then, the relationship between $v_{qp}^*$ and $v_{sdp}^*$ is known as \cite{Luo:07JO}
\begin{equation}  \label{eq:inequal1lemma1}
    \begin{array}{ll}
    v_{qp}^* = v_{sdp}^*, & \mathrm{if} ~~ L \leq 3, \\
    v_{qp}^* \leq 8L \cdot v_{sdp}^*, &  \mathrm{if} ~~ L > 3.
    \end{array}
\end{equation}
Furthermore, note that $\Wbf' := \sum_{i=1}^L \Vbf_i$ is feasible
 for the problem \eqref{eq:problem_SNR_4} since $\mathrm{Tr}(\Vbf_i \Wbf') = \mathrm{Tr(\Vbf_i\sum_{i=1}^L \Vbf_i)})  \ge \sum_{i=1}^L \mathrm{Tr(\Vbf_i \Vbf_i)})\geq \mathrm{Tr}(\Vbf_i \Vbf_i) = 1  $,
and $\mathrm{Tr}(\Wbf') = L$. Hence, we have
\begin{equation}  \label{eq:inequal2lemma1}
    v_{sdp}^* \leq L.
\end{equation}
Hence, with  the optimal solution $\wbf^*$  to \eqref{eq:problem_SNR_2}, we have
\begin{align}
     \min{\left\{ \left| \vbf_1^H \wbf^*\right|^2, \cdots, \left| \vbf_1^H \wbf^*\right|^2\right\}} ~&\stackrel{(a)}{=}~  1/v_{qp}^*   ~\stackrel{(b)}{\ge}~ L/c \cdot 1/v_{sdp}^* ~\stackrel{(c)}{\ge}~  1/c, \label{eq:lemma1_last3}
\end{align}
where $c$ is given by \eqref{eq:Lemma:lemma1}.
Here,   Step (a)  is valid due to the relationship between the original problem
\eqref{eq:problem_SNR_2} and the QP \eqref{eq:problem_SNR_3};   Step (b)  is valid due to \eqref{eq:inequal1lemma1}; and   Step (c) is valid due to \eqref{eq:inequal2lemma1}.
\hfill{$\blacksquare$}

\section*{Appendix B:  Proof of Lemma \ref{lemma:sop}}

The block matrix inversion formula is given as follows:
       \begin{equation}
            \left[
              \begin{array}{cc}
                \Cbf & \Ubf \\
                \Vbf & \Dbf \\
              \end{array}
            \right] = \left[
              \begin{array}{cc}
                \Cbf^{-1} + \Cbf^{-1} \Ubf (\Dbf - \Vbf \Cbf^{-1} \Ubf )^{-1}\Vbf \Cbf^{-1}  &  -\Cbf^{-1}\Ubf(\Dbf - \Vbf \Cbf^{-1}\Ubf)^{-1} \\
                -(\Dbf - \Vbf \Cbf^{-1} \Ubf)^{-1} \Vbf \Cbf^{-1} & (\Dbf - \Vbf \Cbf^{-1} \Ubf)^{-1} \\
              \end{array}
            \right],
        \end{equation}
        which is used in Step (a) in the below.
         \begin{align*}
            &\Pibf_{[\Abf,\Bbf]}^\bot =
            \Ibf - [\Abf ~\Bbf] \left[
              \begin{array}{cc}
                \Abf^H \Abf & \Abf^H \Bbf \\
                \Bbf^H \Abf & \Bbf^H \Bbf \\
              \end{array}
            \right]^{-1} \left[
              \begin{array}{c}
                \Abf^H \\
                \Bbf^H \\
              \end{array}
            \right]     \\
             &  \stackrel{(a)}{=}\Ibf - [\Abf ~ \Bbf] \left[
              \begin{array}{c}
                (\Abf^H \Abf)^{-1} + (\Abf^H \Abf)^{-1}\Abf^H \Bbf (\Bbf^H\Bbf - \Bbf^H \Abf (\Abf^H \Abf)^{-1}\Abf^H \Bbf)^{-1} \Bbf^H \Abf (\Abf^H \Abf)^{-1},  \\
                -(\Bbf^H \Bbf - \Bbf^H \Abf (\Abf^H \Abf)^{-1} \Abf^H \Bbf)^{-1} \Bbf^H \Abf (\Abf^H \Abf)^{-1}, \\
              \end{array}
            \right. \\
            &   \hspace{10em} \left.
            \begin{array}{c}
             - (\Abf^H \Abf)^{-1}\Abf^H \Bbf(\Bbf^H \Bbf - \Bbf^H \Abf(\Abf^H \Abf)^{-1}\Abf^H\Bbf)^{-1} \\
               (\Bbf^H \Bbf - \Bbf^H \Abf (\Abf^H \Abf)^{-1} \Abf^H \Bbf )^{-1}
            \end{array}
            \right]
            \left[
              \begin{array}{c}
                \Abf^H \\
                \Bbf^H \\
              \end{array}
            \right]\\
            &= \Ibf - \Abf (\Abf^H \Abf)^{-1}\Abf^H - \Abf(\Abf^H \Abf)^{-1}\Abf^H \Bbf (\Bbf^H\Bbf - \Bbf^H \Abf (\Abf^H \Abf)^{-1}\Abf^H \Bbf)^{-1} \Bbf^H \Abf (\Abf^H \Abf)^{-1}\Abf^H \\
            &~~~~ + \Bbf(\Bbf^H \Bbf - \Bbf^H \Abf (\Abf^H \Abf)^{-1} \Abf^H \Bbf)^{-1} \Bbf^H \Abf (\Abf^H \Abf)^{-1} \Abf^H \\
            &~~~~ +  \Abf(\Abf^H \Abf)^{-1}\Abf^H \Bbf(\Bbf^H \Bbf - \Bbf^H \Abf(\Abf^H \Abf)^{-1}\Abf^H\Bbf)^{-1}\Bbf^H \\
            &~~~~ - \Bbf (\Bbf^H \Bbf - \Bbf^H \Abf (\Abf^H \Abf)^{-1} \Abf^H \Bbf)^{-1} \Bbf^H \\
            &=  \Ibf -  \Pibf_{\Abf}   -  \Pibf_{\Abf}\Bbf (\Bbf^H\Pibf_{\Abf}^\bot\Bbf)^{-1} ( \Pibf_{\Abf}\Bbf )^H  + \Bbf (\Bbf^H \Pibf_{\Abf}^\bot\Bbf)^{-1} (\Pibf_{\Abf}\Bbf)^H\\
            &~~~~ +  \Pibf_{\Abf}\Bbf(\Bbf^H \Pibf_{\Abf}^\bot\Bbf)^{-1}\Bbf^H - \Bbf (\Bbf^H \Pibf_{\Abf}^\bot\Bbf)^{-1} \Bbf^H
            \end{align*}
            \begin{align*}
            &= \Ibf - \Pibf_{\Abf} - (\Bbf - \Pibf_{\Abf}\Bbf) (\Bbf^H\Pibf_{\Abf}^\bot\Bbf)^{-1}  (\Bbf - \Pibf_{\Abf}\Bbf)^H \hspace{18em}\\
            &= \Pibf_{\Abf}^\bot - \Pibf_{\Abf}^\bot\Bbf (\Bbf^H\Pibf_{\Abf}^\bot\Bbf)^{-1}(\Pibf_{\Abf}^\bot\Bbf)^H \\
            &= \Pibf_{\Abf}^\bot - \Pibf_{\Abf}^\bot\Bbf ((\Pibf_{\Abf}^\bot\Bbf)^H\Pibf_{\Abf}^\bot\Bbf)^{-1}(\Pibf_{\Abf}^\bot\Bbf)^H
 \end{align*}
         where $\Pibf_\Abf =  \Abf (\Abf^H \Abf)^{-1}\Abf^H $, $\Pibf_\Abf^\bot =\Ibf -  \Abf (\Abf^H \Abf)^{-1}\Abf^H $, and the
         block matrix inversion formula is used in Step (a).
         In the last equality, we used $\Pibf_{\Abf}^{\bot H}\Pibf_{\Abf}^\bot=(\Pibf_{\Abf}^\bot)^2 = \Pibf_{\Abf}^\bot$.
          Therefore, we have
        \begin{align*}
            \Pibf_{[\Abf,\Bbf]}^\bot \xbf &=  \Pibf_{\Abf}^\bot \xbf - \Pibf_{\Abf}^\bot\Bbf ((\Pibf_{\Abf}^\bot\Bbf)^H\Pibf_{\Abf}^\bot\Bbf)^{-1}(\Pibf_{\Abf}^\bot\Bbf)^H  \xbf \\
           & =\Pibf_{\Abf}^\bot \xbf - \Pibf_{\Abf}^\bot\Bbf ((\Pibf_{\Abf}^\bot\Bbf)^H\Pibf_{\Abf}^\bot\Bbf)^{-1}(\Pibf_{\Abf}^\bot\Bbf)^H (\Pibf_{\Abf}\xbf + \Pibf_{\Abf}^\bot\xbf) \\
           &\stackrel{(b)}{=}\Pibf_{\Abf}^\bot \xbf - \Pibf_{\Abf}^\bot\Bbf ((\Pibf_{\Abf}^\bot\Bbf)^H\Pibf_{\Abf}^\bot\Bbf)^{-1}(\Pibf_{\Abf}^\bot\Bbf)^H \Pibf_{\Abf}^\bot\xbf\\
           &=(\Ibf- \Pibf_{\Abf}^\bot\Bbf ((\Pibf_{\Abf}^\bot\Bbf)^H\Pibf_{\Abf}^\bot\Bbf)^{-1}(\Pibf_{\Abf}^\bot\Bbf)^H ) \Pibf_{\Abf}^\bot\xbf\\
           &= (\Ibf-\Pibf_{\Pibf_\Abf^\perp \Bbf})\Pibf_\Abf^\perp \xbf,
        \end{align*}
where Step (b) holds because $\Pibf_{\Abf}^\bot\Bbf ((\Pibf_{\Abf}^\bot\Bbf)^H\Pibf_{\Abf}^\bot\Bbf)^{-1}(\Pibf_{\Abf}^\bot\Bbf)^H$ is the projection onto $\Cc(\Pibf_{\Abf}^\bot\Bbf)$ which is a subspace contained in $\Cc^\perp(\Abf)$.
\hfill{$\blacksquare$}

\section*{Appendix C: Proof of Proposition \ref{pro:distribution}}

Consider the effective channel $\gbf_i^{(j)}=\Pibf_{\tilde{\Hbf}_j}^\perp \hbf_i^{(j)}$, where $\hbf_i^{(j)}$ is the channel vector of User $i$ in group $\Gc_j$, and $\tilde{\Hbf}_j$ is defined in  \eqref{eq:tildeHbfj}.  By Lemma
\ref{lemma:sop},  $\gbf_i^{(j)}=\Pibf_{\tilde{\Hbf}_j}^\perp \hbf_i^{(j)}$ can be obtained from sequentially projecting $\hbf_i^{(j)}$ onto the sequential orthogonal spaces associated with the channel vectors of $\Gc_1,\Gc_2,\cdots,\Gc_{j-1},\Gc_{j+1},\cdots,$ $\Gc_{N_g}$, as discussed in Lemma
\ref{lemma:sop} and Example \ref{example:SOP}, i.e.,
$\gbf_i^{(j)}  =    \Pc(\Gc_{N_g}|\Gc_{N_g-1},\cdots,\Gc_{j+1},\Gc_{j-1},\cdots,\Gc_1) \cdots \Pc($ $\Gc_{j+1}|$ $\Gc_{j-1},\cdots,\Gc_1)  \Pc(\Gc_{j-1}|\Gc_{j-2},\cdots,\Gc_1) \cdots\cdot \Pc(\Gc_2|\Gc_1)$
 $\Pc(\Gc_1)\hbf_i^{(j)}$,
 where $\Pc(\Bc|\Ac)$ denotes
 the sequential projection onto the orthogonal space of the projected subspace of $\Bc$ onto $\Cc^\perp(\Ac)$. (Please see Lemma
\ref{lemma:sop} and Example \ref{example:SOP}.) Here, we have $N_g-1$ projection stages. At each projection stage, the proposed user grouping algorithm, Algorithm \ref{al:algoirhtm1}, guarantees that norm reduction is not beyond $(1-\theta^{th})$. The norm of the ZF effective channel can be written as (see Example \ref{example:SOP})
\begin{equation}  \label{eq:PropositionDistributionGbf}
||\gbf_i^{(j)}||^2 = Y ||\hbf_i^{(j)}||^2,
\end{equation}
where the reduction gain random variable $Y$ depends on the channels, but $(1-\theta^{th})^{K-1}=:Y^{th} \le Y \le 1$ since $N_g \le K$.
By \eqref{eq:PropositionDistributionGbf} and Lemma   \ref{eq:lemmaDistribution} below, we have the claim
\eqref{eq:PropDistMain}. \hfill{$\blacksquare$}

\begin{lemma} \label{eq:lemmaDistribution}
       Let $X$ be a random variable  satisfying the condition,
       $\lim_{x \rightarrow 0}\frac{\log \mathrm{Pr}(X \leq x)}{\log x} = d$,      and let $Y$ be a random variable satisfying the condition, $Y^{th} \leq Y \leq 1$, where $Y^{th}$ is some  constant  $\in (0,1]$ and $d$ is some positive constant. Then, the product $Z := XY$ satisfies $\lim_{z \rightarrow 0}\frac{\log \mathrm{Pr}(Z \leq z)}{\log z} = d$.
\end{lemma}

{\em Proof of Lemma \ref{eq:lemmaDistribution}:}
  \begin{align}
             &  \mathrm{Pr}(   X \leq z ) ~ \leq ~ \mathrm{Pr}(  Z \leq z) ~ \leq ~ \mathrm{Pr}( Y^{th} X \leq z ) \label{eq:lem3_1}\\
             & \Leftrightarrow ~~ \mathrm{Pr}(   X \leq z ) ~ \leq ~ \mathrm{Pr}(  Z \leq z) ~ \leq ~ \mathrm{Pr}( X \leq  \frac{z}{Y^{th}} ) \nonumber \\
             & \Leftrightarrow ~~  \lim_{z \rightarrow 0}  \frac{\log \mathrm{Pr}(   X \leq z )}{\log z} ~ \leq ~ \lim_{z \rightarrow 0} \frac{ \log \mathrm{Pr}(  Z \leq z)}{\log{z}} ~ \leq ~ \lim_{z \rightarrow 0}  \frac{ \log \mathrm{Pr}( X \leq  \frac{z}{Y^{th}} )}{\log z} \nonumber \\
              & \Leftrightarrow ~~  \lim_{z \rightarrow 0}  \frac{ \log \mathrm{Pr}(   X \leq z )}{\log z} ~ \leq ~ \lim_{z \rightarrow 0} \frac{ \log \mathrm{Pr}(  Z \leq z)}{\log{z}} ~ \leq ~ \lim_{z \rightarrow 0}  \frac{\log \mathrm{Pr}( X \leq  \frac{z}{Y^{th}} )}{\log \frac{z}{Y^{th}} } \cdot \frac{\log \frac{z}{Y^{th}}}{\log z} \nonumber\\
              & \Leftrightarrow ~~ d ~ \leq ~  \lim_{z \rightarrow 0} \frac{ \log \mathrm{Pr}(  Z \leq z)}{\log{z}}  ~\leq ~ d, \nonumber
 \end{align}
where  \eqref{eq:lem3_1} holds because $Y^{th}X \le Z=YX \le X$ due to $Y\in (Y^{th},1)$.       Therefore, the claim follows. \hfill{$\blacksquare$}

\section*{Appendix D: Existence of Power Distribution Factors}

\begin{lemma}  \label{lem:ExistDeltasForTheo1}
There always exists a collection of in-group power distribution factors $(\delta_1^{(j)},\cdots,\delta_\ell^{(j)})$  for $\Gc_j$ with $|\Gc_j|=\ell$ such that $( \frac{\delta_{i}^{(j)}}{2^{R^{th}}-1} -$ $ \sum_{m=1}^{i-1} \delta_{m}^{(j)} )$ in  \eqref{outage_morethan2} is strictly positive for all $i=2,\cdots, \ell$.
\end{lemma}

\begin{proof} The condition is equivalent to the following:
\begin{equation}
            \frac{\delta_{i}^{(j)}}{2^{R^{th}}-1} - \sum_{m=1}^{i-1} \delta_{m}^{(j)} > 0 ~~ \Leftrightarrow ~~
           \frac{\delta_{i}^{(j)}}{ \sum_{m=1}^{i-1} \delta_{m}^{(j)}}  > 2^{R^{th}}-1, \quad i = 2,
           \cdots, \ell \label{eq:consdition_nonzero}
\end{equation}
Consider the following recursion
\begin{equation}  \label{eq:appendix_recursion}
\delta_{i}^{(j)}= (2^{R^{th}}-1 + C) (\delta_{1}^{(j)}+\cdots+\delta_{i-1}^{(j)}),
\end{equation}
where $C > 0$ is an arbitrary  positive constant. It is easy to see that any solution to \eqref{eq:appendix_recursion} satisfies \eqref{eq:consdition_nonzero}. Solving the recursion yields
\begin{equation}
\delta_{i}^{(j)}= \delta_{1}^{(j)} (2^{R^{th}}-1 + C) (2^{R^{th}}+ C)^{i-2}.
\end{equation}
With  normalization for $\sum_{i=1}^\ell \delta_{i}^{(j)} = 1$, we have
\begin{equation}  \label{eq:deltaSolution}
            \delta_1^{(j)} = \frac{1}{(2^{R^{th}} + C)^{\ell-1}}, \quad \mbox{and} \quad   \delta_{i}^{(j)} = \frac{2^{R^{th}}-1 + C}{(2^{R^{th}} + C)^{\ell-i+1}}, \quad i = 2,\cdots, \ell,
\end{equation}
and all $\delta_i^{(j)}\ge 0$. Hence, we have a collection of power distribution factors for the condition.
\end{proof}

\section*{Appendix E: Scalable Adaptive User Grouping}

We need the angle between the channel subspaces of any two user groups be larger than a certain threshold. This guarantees that inter-group ZF does not harm the diversity order.
A scalable adaptive user grouping method for this purpose can be devised  based on the semi-orthogonal user selection (SUS) algorithm in \cite{Yoo&Goldsmith}.

\begin{figure}[ht]
\begin{psfrags}
    \psfrag{o}[l]{\small origin} %
    \psfrag{cth}[c]{\small $\theta_{\tau,1}$} %
    \psfrag{h}[l]{\small $\hbf_{g_{\hat{k}}}$} %
    \centerline{ \scalefig{0.5} \epsfbox{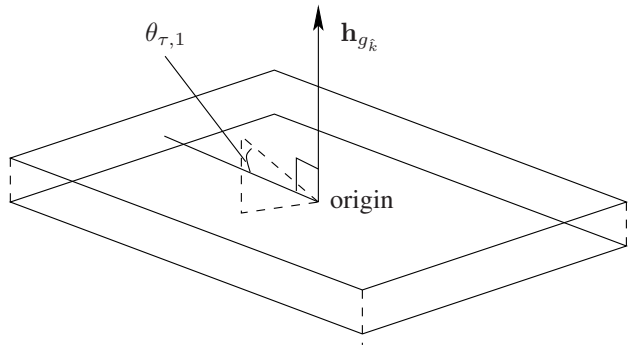} }
    \caption{A hyperslab constructed based on a channel vector (the  dotted line segment from the origin to the plane has length one)}
    \label{fig:sus_hyperslabRev2}
\end{psfrags}
\end{figure}

First, we predetermine two angle threshold values $\theta_{\tau,1} \in (0,\pi/2)$ and $\theta_{\tau,2}\in (0,\pi/2)$ such that $\theta_{\tau,2} <  \frac{\pi}{2}- \theta_{\tau,1}$.
With the channel vectors $\hbf_k, k=1,2\cdots,K$ in ${\mathbb{C}}^N$, like in the SUS algorithm,\footnote{The explanation of the SUS algorithm here is adapted from \cite{Yoo&Goldsmith,LeeSung18COM}.}  we  first
select the user that has the largest channel magnitude. Without loss of generality, we assume the index of the first selected user is  $k_1$. Then, based on the CSI
$\hbf_{k_1}$, we construct a user-selection
hyperslab defined as
\begin{equation}
\Hc_{1} = \left\{ \hbf \in {\mathbb{C}}^N:
\frac{|\hbf_{k_1}^H
\hbf|}{||\hbf_{k_1}||\cdot||\hbf||} \le \gamma  \right\},
\end{equation}
as shown in Fig. \ref{fig:sus_hyperslabRev2}, where the value $\gamma$ is determined to satisfy the following relationship with  the angle $\theta_{\tau,1}$ in Fig.
\ref{fig:sus_hyperslabRev2}:
\[
\gamma = \cos \left(\frac{\pi}{2}-\theta_{\tau,1}\right).
\]
Note that if a vector $\hbf$ is
contained in $\Hc_{1}$, $\hbf$ is semi-orthogonal to
$\hbf_{k_1}$ with the angle between $\hbf_{k_1}$ and $\hbf$ being in $[\frac{\pi}{2}-\theta_{\tau,1}, \frac{\pi}{2}+\theta_{\tau,1}]$.
Then, we select the user whose
channel vector is contained in the hyperslab $\Hc_{1}$ and who
has maximum channel vector magnitude within $\Hc_{1}$. After the
second user is selected, another hyperslab is constructed based on its channel vector. The third user is selected as the user with maximum channel norm within the intersection of the first and second hyperslabs and this guarantees that the third user's channel vector is semi-orthogonal to both first  and second users'
channel vectors with minimum angle separation of $\frac{\pi}{2}-\theta_{\tau,1}$. We continue this procedure until either we cannot find any user in the intersection or we reach the final $K$-th user.   This is basically the SUS algorithm. If the procedure  reaches the $K$-th user, we have $K$ user groups each with one user and the constructed $K$ groups satisfy the required angle separation property. If the procedure stops at the $N^\prime_g$-th step before reaching the $K$-th user, then we construct $N^\prime_g$ candidate user groups. At this point, each candidate user group has one user obtained from the SUS algorithm.

Now consider the remaining $K-N_g^\prime$ users. Each of the remaining $K-N_g^\prime$ users' channels should be close to one of the channels of the $N_g^\prime$ users obtained by the SUS procedure with angle less than $\frac{\pi}{2}-\theta_{\tau,1}$. Otherwise, one separate group had been constructed in the above SUS stage.  Let the remaining $K-N_g^\prime$ users be named Users $u_1,u_2,\cdots,u_{K-N_g^\prime}$.
Now, pick User $u_1$ and compute the angle between the channel of User $u_1$ and the channel of each of the $N^\prime_g$ users obtained by the above SUS stage. Let the angles be $\{\theta_1^{(1)},\theta_2^{(1)},\cdots,\theta_{N_g^\prime}^{(1)}\}$. Assign User $u_1$ to  the group with the smallest angle distance. Furthermore, combine the groups
\begin{equation}  \label{eq:appendEgroupComb}
\{\Gc_j~|~j=1,\cdots,N_g^\prime ~\mbox{and}~ \theta_j^{(1)} < \theta_{\tau,2}\}
\end{equation}
as a single group. That is, if the minimum angle is not guaranteed between groups due to the inclusion of User $u_1$, then combine the groups violating the minimum angle distance condition. Suppose that the assigned group is Group 1 without loss of generality and still all groups satisfy the minimum angle distance condition. Now, pick User $u_2$ and compute the angle between the channel of User $u_2$ and the channel of each of the $N^\prime_g$ groups. Since we have two users in Group 1, we compute two angle values between User $u_2$ and the two users of Group 1 and denote them by $\theta_{11}^{(2)}$ and $\theta_{12}^{(2)}$.  Let the angle between User $u_2$ and the single user in each user of Groups $2,\cdots, N_g^\prime$ be $\theta_2^{(2)},\cdots,\theta_{N_g^\prime}^{(2)}$.
Compute the minimum of $\{\theta_{11}^{(2)},\theta_{12}^{(2)},\theta_2^{(2)},\cdots,\theta_{N_g^\prime}^{(2)}\}$ and assign User $u_2$ to the group that has the user with the minimum angle distance from User $u_2$. Again, combine the groups violating the minimum angle condition due to the inclusion of User $u_2$ by similar computation to \eqref{eq:appendEgroupComb} with the same threshold $\theta_{\tau,2}$.
After that, continue to User $u_3$.    We continue this procedure until User $u_{K-N_g^\prime}$ is assigned. Finally, the procedure will return $N_g(\le N_g^\prime)$ groups.

The above method guarantees the minimum  angle $\theta_{\tau,2}$ between any two groups and is scalable with respect to $K$ since the SUS algorithm is sequential and angle checking of the remaining $K-N_g^\prime$ users with the  $N_g^\prime$ groups requires at most $K^2$ checkings.

{We actually implemented  this grouping idea and the result is shown in Fig.  \ref{fig:grouping_Reviewer1_3}. We did not fine-tune the two parameters. It seems that more tweaking is necessary for stable performance. However, it is observed that the new approach also yields far better outage performance as expected. Indeed, more efficient user grouping algorithms for the desired purpose can be devised.}

\begin{figure}[ht]
\begin{psfrags}
    \centerline{ \scalefig{0.6} \epsfbox{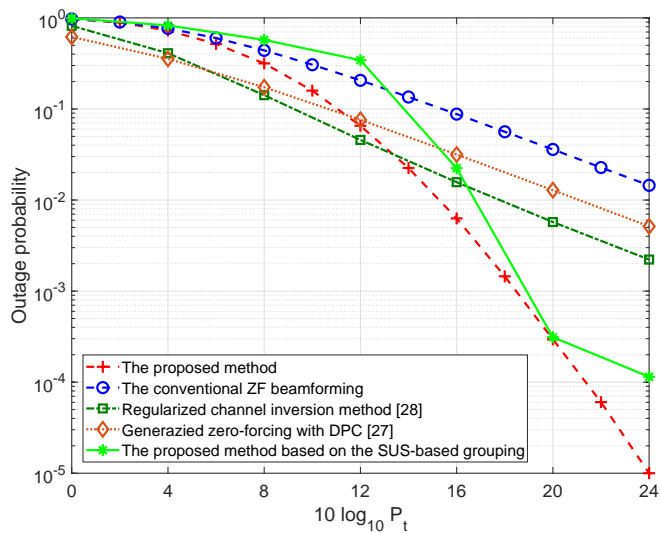} }
    \caption{{Outage probability of several methods:  The SUS-based grouping with $\theta_{\tau,1}=0.25$ and $\theta_{\tau,2}=0.55$}}
    \label{fig:grouping_Reviewer1_3}
\end{psfrags}
\end{figure}



\bibliographystyle{ieeetr}

\bibliography{referenceBibs6}

\begin{thebibliography}{10}

\bibitem{Seo18Thesis}
J.~Seo, ``Beamformer design based on non-linear reception in {MIMO} downlink
  system,'' {\em PhD Dissertation}, KAIST, Aug. 2018.

\bibitem{Weingartenetal06IT}
H.~Weingarten, Y.~Steinberg, and S.~Shamai, ``{The capacity region of the
  Gaussian multiple-input multiple-output broadcast channels},'' {\em IEEE
  Trans. Inf. Theory}, vol.~52, pp.~3936--3964, Sep. 2009.

\bibitem{Sharif&Hassibi:05IT}
M.~Sharif and B.~Hassibi, ``{On the capacity of MIMO broadcast channels with
  partial side information},'' {\em IEEE Trans. Inf. Theory}, vol.~51,
  pp.~506--522, Feb. 2005.

\bibitem{Yoo&Goldsmith}
T.~Yoo and A.~Goldsmith, ``On the optimality of multiantenna broadcast
  scheduling using zero-forcing beamforming,'' {\em IEEE J. Sel. Areas in
  Commun.}, vol.~24, pp.~528--541, Mar. 2006.

\bibitem{LTEMUMIMOphy}
3GPP, ``{Evolved universal terrestrial radio access (E-UTRA): Downlink multiple
  input multiple output (MIMO) enhancement for LTE-Advanced (Release 11)},''
  {\em TR 36.871, V11.0.0}, 2011-2012.

\bibitem{Spencer&Swindlehurst&Haardt:04IT}
Q.~H. Spencer, A.~L. Swindlehurst, and M.~Haardt, ``{Zero-forcing methods for
  downlink spatial multiplexing in multiuser MIMO channels},'' {\em IEEE Trans.
  Signal Process.}, vol.~52, pp.~461 -- 471, Feb. 2004.

\bibitem{LeeSung18COM}
G.~Lee and Y.~Sung, ``{A new approach to user scheduling in massive multi-user
  MIMO broadcast channels},'' {\em IEEE Trans. Commun.}, vol.~66, pp.~1481 --
  1495, Apr. 2018.

\bibitem{LeeSungSeo}
{G. Lee, Y. Sung, and J. Seo}, ``Randomly-directional beamforming in mm-wave
  multi-user {MISO} downlink,'' {\em IEEE Trans. Wireless Commun.}, vol.~15,
  pp.~1086--1100, Feb. 2016.

\bibitem{LeeSungKountouris}
{G. Lee, Y. Sung, and M. Kountouris}, ``On the performance of random
  beamforming in sparse millimeter wave channels,'' {\em IEEE J. Sel. Topics
  Signal Process.}, vol.~10, pp.~560--575, Apr. 2016.

\bibitem{Huh:12IT}
H.~Huh, A.~M. Tulino, and G.~Caire, ``Network mimo with linear zero-forcing
  beamforming: Large system analysis, impact of channel estimation, and
  reduced-complexity scheduling,'' {\em IEEE Trans. Inf. Theory}, vol.~58,
  pp.~2911--2934, Dec. 2012.

\bibitem{Saito:13VTC}
Y.~Saito, Y.~Kishiyama, A.~Benjebbour, T.~Nakamura, A.~Li, and K.~Higuchi,
  ``{Non-orthogonal multiple access (NOMA) for cellular future radio access},''
  in {\em Proc. IEEE VTC}, pp.~1--5, Jun. 2013.

\bibitem{mao2018rate}
Y.~Mao, B.~Clerckx, and V.~O. Li, ``Rate-splitting multiple access for downlink
  communication systems: bridging, generalizing, and outperforming sdma and
  noma,'' {\em EURASIP Journal on Wireless Communications and Networking},
  vol.~2018, no.~1, p.~133, 2018.

\bibitem{Seo&Sung:17SPAWC}
J.~Seo and Y.~Sung, ``{A new transceiver architecture for multi-user MIMO
  communication based on mixture of linear and non-linear reception},'' in {\em
  Proc. SPAWC}, (Sapporo, Japan), Jun. 2017.

\bibitem{Chen&Ding&Dai:16Access}
Z.~Chen, Z.~Ding, and X.~Dai, ``Beamforming for combating inter-cluster and
  intra-cluster interference in hybrid {NOMA} systems,'' {\em IEEE Access},
  vol.~4, pp.~4452 -- 4463, Aug. 2016.

\bibitem{Hamid08COM}
J.~Kazemitabar and H.~Jafarkhani, ``Multiuser interference cancellation and
  detection for users with more than two transmit antennas,'' {\em IEEE Trans.
  Commun.}, vol.~56, no.~4, pp.~574 -- 583, 2008.

\bibitem{Adhikary13IT}
A.~Adhikary, J.~Nam, J.~Ahn, and G.~Caire, ``Joint spatial division and
  multiplexing $-$ {T}he large-scale array regime,'' {\em IEEE Trans. Inf.
  Theory}, vol.~59, pp.~6441 -- 6463, Oct. 2013.

\bibitem{Tse:book}
D.~Tse and P.~Viswanath, {\em Fundamentals of Wireless Communication}.
\newblock Cambridge University Press, 2005.

\bibitem{Hanif16:TSP}
M.~F. Hanif, Z.~Ding, T.~Ratnarajah, and G.~K. Karagiannidis, ``A
  minorization-maximization method for optimizing sum rate in the downlink of
  non-orthogonal multiple access systems,'' {\em IEEE Trans. Signal Process.},
  vol.~64, no.~1, pp.~76--88, 2016.

\bibitem{yuille03:Concave}
A.~L. Yuille and A.~Rangarajan, ``The concave-convex procedure,'' {\em Neural
  computation}, vol.~15, no.~4, pp.~915--936, 2003.

\bibitem{Seo&Sung:18SP}
J.~Seo and Y.~Sung, ``{Beam design and user scheduling for non-orthogonal
  multiple access with multiple antennas based on Pareto-optimality},'' {\em
  IEEE Trans. Signal Process.}, vol.~66, pp.~2876 -- 2891, Jun. 2018.

\bibitem{Zhang&Cui:10SP}
R.~Zhang and S.~Cui, ``Cooperative interference management with miso
  beamforming,'' {\em IEEE Trans. Signal Process.}, vol.~58, pp.~5450 -- 5458,
  Oct. 2010.

\bibitem{AliHossainKim17Access}
S.~Ali, E.~Hossain, and D.~I. Kim, ``Non-orthogonal multiple access ({NOMA})
  for downlink multiuser {MIMO} systems: {U}ser clustering, beamforming, and
  power allocation,'' {\em IEEE Access}, vol.~5, pp.~565 -- 577, 2017.

\bibitem{JileiHou}
{J. Hou, J. E. Smee, H. D. Pfister, and S. Tomasin}, ``{Implementing
  interference cancellation to increase the EV-DO Rev A reverse link
  capacity},'' {\em IEEE Communications Magazine}, vol.~44, pp.~58 -- 64, Feb.
  2006.

\bibitem{DingAdachiPoor16WC2}
Z.~Ding, F.~Adachi, and H.~V. Poor, ``The application of {MIMO} to
  non-orthogonal multiple access,'' {\em IEEE Trans. Wireless Commun.},
  vol.~15, pp.~537 -- 552, Jan. 2016.

\bibitem{ZhengTse:03IT}
L.~Zheng and D.~N.~C. Tse, ``{Diversity and multiplexing: a fundamental
  tradeoff in multiple-antenna channels},'' {\em IEEE Trans. Inf. Theory},
  vol.~49, pp.~1073 -- 1096, May 2003.

\bibitem{Mroueh08ISIT}
{L. Mroueh, S. R.-Leveil, G. R.-Ben Othman, and J.-C. Belfiore}, ``Dmt of
  weighted parallel channels: Application to broadcast channels,'' in {\em
  Proc. ISIT}, Jul. 2008.

\bibitem{HuRusek:17WC}
S.~Hu and F.~Rusek, ``{A generalized zero-forcing precoder with successive
  dirty-paper coding in MISO broadcast channels},'' {\em IEEE Trans. Wireless
  Commun.}, vol.~16, pp.~3632 -- 3645, Jun. 2017.

\bibitem{PeelHochwaldSwindlehurst:05COM}
C.~B. Peel, B.~M. Hochwald, and A.~L. Swindlehurst, ``{A vector-perturbation
  technique for near-capacity multiantenna multiuser communication-part I:
  Channel inversion and regularization},'' {\em IEEE Trans. Commun.}, vol.~53,
  pp.~195 -- 202, Jan. 2005.

\bibitem{Jindal:06IT}
N.~Jindal, ``{MIMO broadcast channels with finite-rate feedback},'' {\em IEEE
  Trans. Inf. Theory}, vol.~52, pp.~5045 -- 5060, Nov. 2006.

\bibitem{ZengYetis:13SP}
Y.~Zeng, C.~M. Yetis, E.~Gunawan, Y.~L. Guan, and R.~Zhang, ``{Transmit
  optimization with improper Gaussian signaling for interference channels},''
  {\em IEEE Trans. Signal Process.}, vol.~61, pp.~2899 -- 2913, Jun. 2013.

\bibitem{Sidiropoulos:06TSP}
N.~D. Sidiropoulos, T.~N. Davidson, and Z.-Q. Luo, ``Transmit beamforming for
  physical-layer multicasting,'' {\em IEEE Trans. Signal Process.}, vol.~54,
  no.~6, pp.~2239--2251, 2006.

\bibitem{Luo:07JO}
Z.-Q. Luo, N.~D. Sidiropoulos, P.~Tseng, and S.~Zhang, ``Approximation bounds
  for quadratic optimization with homogeneous quadratic constraints,'' {\em
  SIAM Journal on optimization}, vol.~18, no.~1, pp.~1--28, 2007.

\end{thebibliography}


\end{document}

\begin{figure}[htbp]
\centerline{ \SetLabels
\L(0.25*-0.1) (a) \\
\L(0.76*-0.1) (b) \\
\endSetLabels
\leavevmode
\strut\AffixLabels{
\scalefig{0.45}\epsfbox{figures/rate_distribution_PT10.eps}
\scalefig{0.45}\epsfbox{figures/rate_distribution_PT15.eps}
} } \vspace{0.3cm} \centerline{ \SetLabels
\L(0.25*-0.1) (c) \\
\L(0.76*-0.1) (d) \\
\endSetLabels
\leavevmode
\strut\AffixLabels{
\scalefig{0.45}\epsfbox{figures/rate_distribution_PT20.eps}
\scalefig{0.45}\epsfbox{figures/rate_distribution_PT20.eps}
} } \vspace{0.5cm} \caption{PD vs. $\Delta_1$ ($M=2$,
$\Delta=0.02$, SNR = 10 dB): (a) $A=1$ (b) $A=15$ (c) $A=30$ (d)
$A=100$} \label{fig:PDperiodpatternM2highSNR}
\end{figure}